\documentclass[aps,reprint,superscriptaddress,nofootinbib, showkeys]{revtex4-2}

\usepackage[utf8]{inputenc}
\usepackage[T1]{fontenc}
\usepackage[english]{babel}
\usepackage{selinput}

\usepackage{graphics}
\usepackage{graphicx}
\usepackage{float}
\usepackage{adjustbox}
\usepackage{placeins}

\usepackage{amsmath}
\usepackage{amsthm}
\usepackage{amsfonts}
\usepackage{amssymb}
\usepackage{mathtools}
\usepackage{bm}
\usepackage{bbm}
\usepackage{dsfont}
\usepackage{mathdots}
\usepackage{physics}
\usepackage{braket}
\usepackage{nccmath}

\usepackage[normalem]{ulem}
\usepackage[shortlabels]{enumitem}
\usepackage{csquotes}
\usepackage{comment}
\usepackage{verbatim}
\usepackage{lipsum}

\usepackage[linesnumbered,ruled,vlined]{algorithm2e}
\SetKwInput{kwInit}{Init}

\usepackage[caption=false]{subfig}
\usepackage{tabularx,booktabs}
\usepackage{multirow}
\usepackage{ytableau}

\usepackage[toc,page]{appendix}
\usepackage{natbib}
\usepackage{xcolor}
\usepackage[colorlinks=true,citecolor=blue,linkcolor=magenta]{hyperref}

\usepackage{titlesec}
\titleformat{\paragraph}[runin]
  {\normalfont\bfseries}
  {}
  {0pt}
  {}

\makeatletter
\let\newfloat\newfloat@ltx
\makeatother

\def\Cbb{\mathbb{C}}

\def\HC{\mathcal{H}}

\def\LC{\mathcal{L}}

\def\ad{^{\dagger}}

\newcommand{\majoranaNorm}[1]{\|#1\|_{{\rm maj},2}}

\newcommand{\fsnull}[1]{}
\newcommand{\old}[1]{}

\newcommand{\ketbraq}[1]{\ketbra{#1}{#1}}
\newcommand{\bramatket}[3]{\langle #1 \hspace{1pt} | #2 | \hspace{1pt} #3 \rangle}
\newcommand{\bramatketq}[2]{\bramatket{#1}{#2}{#1}}

\newcommand{\poly}{\operatorname{poly}}

\newcommand{\Ebb}{\mathbb{E}}

\newcommand{\R}{\mathbb{R}}

\newcommand{\AC}{\mathcal{A}}
\newcommand{\BC}{\mathcal{B}}
\newcommand{\CC}{\mathcal{C}}

\newcommand{\GC}{\mathcal{G}}

\newcommand{\MC}{\mathcal{M}}

\newcommand{\OC}{\mathcal{O}}
\newcommand{\PC}{\mathcal{P}}

\newcommand{\VC}{\mathcal{V}}

\newcommand{\XC}{\mathcal{X}}

\newcommand{\Var}{{\rm Var}}

\renewcommand{\geq}{\geqslant}
\renewcommand{\leq}{\leqslant}

\newcommand{\spn}{{\rm span}}

\renewcommand{\vec}[1]{\boldsymbol{#1}}

\newcommand*{\id}{\openone}

\newcommand{\bs}{\textsf{BS}}

\newcommand{\al}{\alpha }

\newcommand{\dl}{\delta }

\newcommand{\ep}{\epsilon}

\newcommand{\lm}{\lambda }

\newcommand{\sg}{\sigma }

\newcommand{\liea}{\mathfrak{g}}

\def\be{\begin{equation}}
\def\ee{\end{equation}}
\def\bs{\begin{split}}
\def\e{\end{split}}
\def\ba{\begin{eqnarray}}
\def\bea{\begin{eqnarray}}

\def\tea{\end{eqnarray}}
\def\ea{\end{eqnarray}}
\def\eea{\end{eqnarray}}

\def\eye{\mathds{1}}

\def\R{\mathds{R}}

\def\lieg{\mathds{G}}

\def\U{\mathrm{U}}

\newtheorem{theorem}{Theorem}
\newtheorem{lemma}{Lemma}

\newtheorem{corollary}{Corollary}

\newtheorem{definition}{Definition}

\newsavebox{\mstrut}
\newcommand{\rrangle}{\rangle\kern-0.4\ht\mstrut\right\rangle}
\newcommand{\llangle}{\langle\kern-0.4\ht\mstrut\left\langle}

\newcommand{\Cl}{{\rm Cl}}
\newcommand{\FGU}{{\rm FGU}}

\newcommand{\QS}{\rm Q.S.}
\newcommand{\CT}{\rm C.T}
\newcommand{\QT}{\rm Q.T}

\newcolumntype{L}{>{$}l<{$}}  
\usepackage{siunitx}

\begin{document}

\title{The cost of simulating classically tractable quantum circuits and dynamics }

\author{Su Yeon Chang}
\email{suyeon.chang97@gmail.com}
\affiliation{Theoretical Division, Los Alamos National Laboratory, Los Alamos, NM 87545, USA}
\affiliation{European Organization for Nuclear Research (CERN), Geneva 1211, Switzerland}
\affiliation{Institute of Physics, Ecole Polytechnique Fédérale de Lausanne (EPFL),  Lausanne 1015, Switzerland}

\author{Supanut Thanasilp}
\affiliation{Chula Intelligent and Complex Systems, Department of Physics, Faculty of Science, Chulalongkorn University, Bangkok 10330, Thailand}

\author{Zo\"e Holmes}
\affiliation{Institute of Physics, Ecole Polytechnique Fédérale de Lausanne (EPFL),  Lausanne 1015, Switzerland}
\affiliation{Centre for Quantum Science and Engineering, Ecole Polytechnique Fédérale de Lausanne (EPFL),  Lausanne 1015, Switzerland}

\author{M. Cerezo}
\email{cerezo@lanl.gov}
\affiliation{Information Sciences, Los Alamos National Laboratory, Los Alamos, NM 87545, USA}

\begin{abstract}
Determining whether a quantum evolution can be efficiently simulated
classically is central to understanding the boundary between classical
and quantum computation.
However, polynomial-time simulability is an asymptotic statement, and does not by itself determine whether the (quantum-inspired) classical simulation is actually practical. 
Indeed, different
polynomial scalings can lead to vastly different computational costs,
particularly when expensive preprocessing or quantum data acquisition
is required. In this work, we ask whether classically simulable quantum
dynamics are in practice more resource-efficient to simulate
classically than to execute directly on quantum hardware. We analyze
this question using three resource metrics, quantum sample, quantum
time, and classical time complexity, for several widely studied
classically simulable circuit families.  
Using representative hardware-level estimates, we identify regimes in
which quantum simulation can be faster despite the existence of a
polynomial-time classical algorithm, as well as regimes in which
classical simulation remains more efficient. 
At the same time, the large quantum
sampling cost needed to characterize unknown input states can make this polynomial-time classical simulation prohibitively expensive with current cloud-based  hardware access prices.
Ultimately, our work indicates that guarantees of classical simulability with polynomial resources alone are insufficient to determine the
preferred implementation.

\end{abstract}

    \maketitle

\section{Introduction}

While simulating generic quantum circuits and dynamics is generally believed to be classically intractable in the worst case, a broad range of structured quantum processes admit efficient classical simulation. Prominent examples include Clifford circuits acting on stabilizer states~\cite{gottesman1998heisenberg,aaronson2004improved}, free-fermionic dynamics~\cite{terhal2002classical,valiant2001quantum,jozsa2008matchgates}, and low-entanglement dynamics accessible through tensor-network methods~\cite{vidal2003efficient,markov2008simulating}. More recently, approaches based on Lie-algebraic simulation, Pauli Propagation, symmetry, and architecture-specific surrogate models have substantially enlarged the set of quantum evolutions that can be treated efficiently on classical hardware~\cite{goh2023lie,rudolph2025pauli,teng2025leveraging,bermejo2024quantum}. Such methods are useful both for understanding the boundary between classical and quantum computation and as practical tools for developing quantum algorithms or hardware
~\cite{horodecki2009quantum,bravyi2005universal,jozsa2010matchgate,cerezo2020variationalreview,gujju2024quantum,wang2024comprehensive,chang2025primer,cerezo2023does}. These results establish when classical simulation is possible within polynomial runtime, but not whether it is actually the preferable implementation.

\begin{figure*}[t]
\subfloat[Quantum Simulation]{
    \includegraphics[width=0.5\linewidth]
    {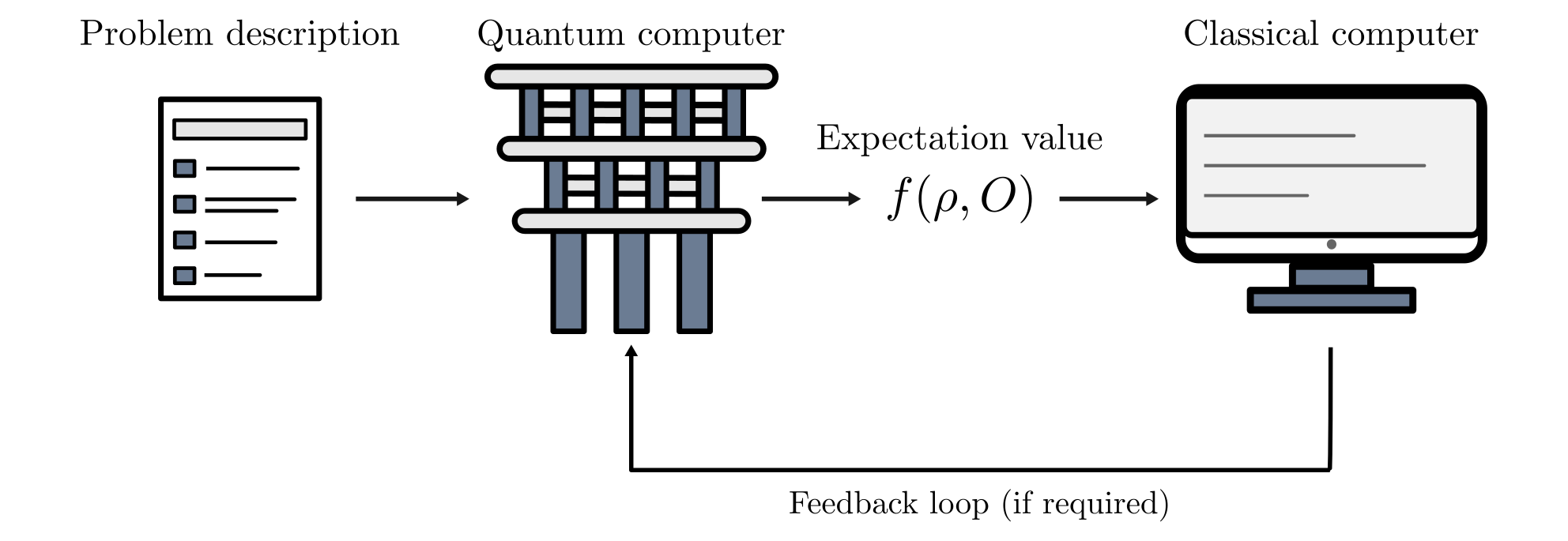}
}
\subfloat[Classical Simulation]{
    \includegraphics[width=0.5\linewidth]
    {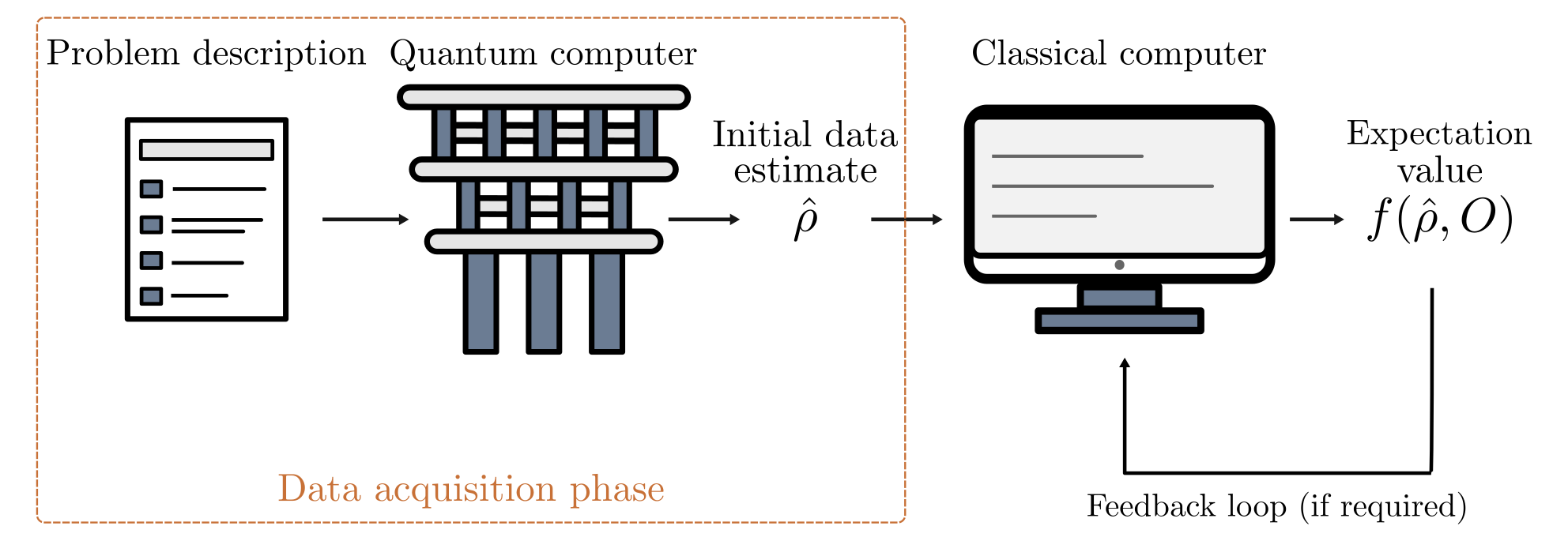}
}
\caption{\textbf{Schematic description of Quantum and Classical
Simulations.}
a) In Quantum Simulation, each circuit instance is executed directly on
the quantum hardware and the corresponding expectation value is
estimated from quantum measurements. When required, a classical
computer updates the circuit parameters between evaluations.
b) In CS, the quantum evolution is replaced by a
classical surrogate. If the required information about the initial state
is not known classically, it is obtained during an initial quantum
data-acquisition stage and subsequently reused by the Classical
Simulation for different circuit instances (e.g. for different simulation times or parameter values).}
\label{fig:qsim_qesim}
\end{figure*}

Indeed, classical simulability is fundamentally an asymptotic statement. A simulation whose runtime scales as a high-degree polynomial can still be prohibitively expensive at the system sizes of interest, while direct execution on quantum hardware can incur a very different combination of circuit depth and sampling costs. Moreover,
classical simulations for non-classical input states require information about the input
state that must itself be acquired experimentally. Thus, once a quantum
evolution is known to be classically simulable, the relevant question
is no longer simply whether it \emph{can} be simulated classically, but
whether it \emph{should} be. This motivates the central question of
this work:
\begin{quote}
\textit{If a quantum evolution is classically simulable, is it actually
more resource-efficient to simulate it classically than to run it on
quantum hardware?}
\end{quote}

To address this question, we compare the two computational routes shown
schematically in Fig.~\ref{fig:qsim_qesim}. We refer to a
Quantum Simulation (QS) as an algorithm in which the quantum
evolution is executed directly on quantum hardware. For each circuit
instance, the input state is prepared, evolved, and measured to estimate
the desired expectation values. In contrast, a (quantum-inspired) Classical
Simulation (CS) exploits structure in the circuit, observable, or relevant
operator space to replace the quantum evolution by a polynomial-time
classical computation. Importantly, CS does not
necessarily eliminate access to quantum hardware. When the initial
state is not known classically, the surrogate can require an initial
quantum data-acquisition stage, for instance through direct
measurements, tomography, or classical shadows~\cite{elben2022randomized,west2026classical}. This information can
then be processed and reused across subsequent circuit instances. When
the input state is already known classically, this quantum
data-acquisition step is absent.

The goal of this work is to quantify the resources required by these two
approaches on equal footing. We consider several families of provably
classically simulable quantum evolutions, including Clifford circuits,
shallow local circuits, Matchgate circuits, $S_n$- and
$\U(1)$-equivariant circuits, randomly initialized quantum convolutional neural networks,
and a randomized model  of ADAPT-VQE circuits. For each problem, we
construct or identify a suitable CS and derive its
Quantum Sample (number of measurements), Quantum Time (quantum circuit depth), and Classical Time (number of classical operations) complexities. A
central distinction throughout our analysis is whether a given cost is
incurred only once, either to characterize the initial state or
construct the classical surrogate, or whether it must be paid again for
each circuit instance.

Our results show that the existence of a polynomial-time CS does not by itself determine which implementation is more resource-efficient at finite system sizes. The comparison depends on how the total cost is divided between one-time contributions and per-instance contributions. QS incurs sampling and circuit-execution costs for every circuit instance. CS may instead require a large one-time cost for quantum data acquisition and preprocessing, together with a classical evaluation cost for each circuit instance. When the resource estimates are translated into representative hardware-level costs, we identify regimes in which the relative advantage of QS and CS depends on the number of circuit instances, leading to a crossover between the two approaches. Thus, classical simulability, wall-clock advantage, and monetary advantage constitute distinct notions of efficiency.

The paper is organized as follows.
Section~\ref{sec:framework} introduces the expectation-value estimation
problem, the resource metrics used throughout the paper, and the
classically simulable circuit families considered below.
Section~\ref{sec:quantum_simulation} derives the corresponding resource
requirements for a general QS.
In Section~\ref{sec:classical_methods}, we review the two Classical
Simulation frameworks that require additional machinery, namely
$\liea$-sim and Pauli Propagation. Readers who are familiar with them may skip this section.
Section~\ref{sec:cs_for_evolutions} derives the resource complexity
of architecture-specific CS for the different model
classes.
Section~\ref{sec:practical_cost} translates these complexity estimates
into representative wall-clock times and quantum hardware access costs.
Finally, Section~\ref{sec:conclusion} discusses the scope of the
comparison and the implications for practical quantum advantage.
Detailed derivations and proofs of the computational results are
provided in the appendices.

\section{Framework\label{sec:framework}}

\subsection{Estimating expectation values\label{sec:expectation_values}}

Let $\mathcal{H}$ denote an $n$-qubit Hilbert space, and let
$\mathcal{B}(\mathcal{H}) \simeq \mathcal{H}\otimes\mathcal{H}^*$ denote
the space of linear operators acting on $\mathcal{H}$. We consider an
input state $\rho\in\mathcal{B}(\mathcal{H})$ and a quantum evolution
of the form
\begin{equation}\label{eq:generic_circuit}
    U = \prod_{\ell=1}^{L} U_\ell 
      = \prod_{\ell=1}^{L} e^{-i w_\ell H_\ell}\;,
\end{equation}
where we adopt the convention that products are ordered from right to left: $\prod_{\ell = 1}^L U_\ell = U_L \cdots U_1.$
Here, $\mathcal{G}=\{H_1,\ldots,H_L\}$ is a set of Hermitian generators
and $\bm{w}=(w_1,\ldots,w_L)\in\mathbb{R}^L$ is a real-valued vector. The physical interpretation of $\vec{w}$ depends on the specific application. For instance, in variational quantum algorithms or quantum machine learning, the components of $\vec{w}$ correspond to trainable parameters (typically denoted by $\vec{\theta}$), while in quantum dynamics simulations, they could represent the discrete time steps governing the system's evolution (denoted by $t$) for a sequence of Hamiltonians or after some form of Trotterization. Moreover, we also note that the $L = 1$ case captures continuous time dynamics.

Next, we define an observable $O$ as a linear combination of $M$ orthogonal Hermitian operators from a basis $\PC$ of $\BC$, i.e.,
\begin{equation}
    O = \sum_{\al} c_\al P_\al\,,
    \label{eq:observable}
\end{equation}
with  $c_\al \in \R$ and $P_\al\in \PC$. We denote the $\ell_1$- and $\ell_2$- norms of the coefficient vectors by $\norm{\vec{c}}_1 = \sum_\al \abs{c_\al}$ and $\norm{\vec{c}}_2 = (\sum_\al c_\al^2)^{1/2}$, respectively.  For instance, we can take $\PC$ to be the set of $n$-qubit Pauli operators $\{\eye,  X,Y,Z\}^{\otimes n}$.

The central task considered throughout this work is to estimate the
expectation value
\begin{equation}\label{eq:loss}
    f_U(\rho,O) = \Tr[U\rho U^\dagger O]
\end{equation}
to additive precision $\epsilon$. That is, we seek an estimator
$\widetilde{f}_U(\rho,O)$ satisfying
\begin{equation}
    \left|f_U(\rho,O)-\widetilde{f}_U(\rho,O)\right|\leq \epsilon\;,
\end{equation}
with a high success probability $1-\delta$ and we want to quantify the different resources required to produce it.

While Eq.~\eqref{eq:loss} describes a single evolution, we will generally
consider a collection of $T$ circuit instances
$\{U^{(i)}\}_{i=1}^{T}$. Each $U^{(i)}$ belongs to the same circuit family and is
constructed from the same set of generators, but can correspond to a
different set of parameters,
\begin{equation}
    U^{(i)}
    = \prod_{\ell = 1}^L U_\ell^{(i)} = 
    \prod_{\ell=1}^{L}
    e^{-iw_{i,\ell}H_\ell}.
\end{equation}
For circuit families with discrete gate choices, such as Clifford
circuits, different instances can instead correspond to different gate
sequences drawn from the same allowed family.

Thus, $T$ denotes the total number of distinct circuit evaluations whose
expectation values are required, i.e., $\{f_{U^{(i)}}(\rho, O) \}_{i=1}^T$. In a variational setting, $T$ can count
the circuit instances evaluated throughout the optimization, while in a
dynamical simulation it can denote the number of time points considered.
If only a single circuit is evaluated, then $T=1$. ADAPT-VQE is treated
separately below because the ansatz itself changes at each circuit evaluation. In that setting, $L$ denotes the final number of selected
generators and $T_{\rm VQE}$ denotes the total number of energy
evaluations performed across all parameter re-optimization stages.

We make two additional assumptions. First, we assume repeated access to
independently prepared copies of $\rho$, but no quantum memory allowing
coherent operations on multiple copies and no general classical description of
a circuit preparing $\rho$ (unless $\rho$ is a trivial state to prepare). Thus, each copy can be evolved and measured
independently, while multi-copy protocols are outside the scope of this
work. Second, we assume that the unitary generators and the measurement
operator admit efficient classical descriptions. In practice, $H_\ell$
and $O$ are taken to be linear combinations of $\mathcal{O}(\poly(n))$
elements of $\mathcal{P}$.  With these assumptions in mind, we evaluate how different methods to estimate expectation values $f_U(\rho,O)$ fare against each other in terms of the following three metrics:
\begin{itemize}
    \item \textbf{Quantum Sample complexity}: number of copies of $\rho$
    required to estimate the quantities of interest. In all settings
    considered here, this coincides with the total number of circuit shots
    and therefore with the number of accesses to quantum hardware. This
    metric is also directly relevant when the quantum device is accessed
    through a per-shot cost model.

    \item \textbf{Quantum Time complexity}: quantum circuit depth required
    for a single sample. This includes the gates needed to implement the
    evolution of $\rho$ and any basis changes required before measurement.
    Together with hardware-dependent initialization, gate, readout, and
    reset times, this quantity can later be converted into a wall-clock
    cost per shot.

    \item \textbf{Classical Time complexity}: number of elementary
    operations performed on the classical computer. This includes the
    processing of measurement outcomes and, when applicable, the
    optimization or surrogate-evaluation steps. Once a classical hardware
    model is specified, this quantity can likewise be converted into a
    wall-clock time.
\end{itemize}

\subsection{Classically simulable evolutions\label{sec:classically_simulable_evolutions}}
In this section, we review the classes of quantum evolutions considered
throughout this work. In all cases, we only briefly discuss the structure
and purpose of each unitary, and refer the reader to the corresponding
references for additional details. Crucially for our purposes,
polynomial simulability alone does not determine practical cost. The
dimension of the reduced representation, the amount of information
required about the input state, and whether this information can be
reused across different evolutions can lead to very different
polynomial overheads.

Although these evolutions act on an exponentially large Hilbert space,
their structure allows the expectation value in
Eq.~\eqref{eq:loss} to be estimated using classical resources that scale
at most polynomially with the system size. We organize the examples
below according to the structure underlying this reduction. We first
consider exact simulations, beginning with Clifford and shallow local
circuits, where simulability follows directly from closure of the Pauli
basis or restriction to a polynomially sized light cone. We then turn
to exact simulations that exploit more structured reduced
representations, including invariant Majorana sectors, permutation
symmetry, and fixed Hamming-weight subspaces. Finally, we consider
QCNNs and ADAPT-VQE, where polynomial runtime is obtained by truncating
sufficiently high-weight Pauli or Majorana contributions and therefore
introduces a controllable approximation error. Importantly, classical
simulability is generally a property of the complete estimation problem,
rather than of the circuit alone, and depends on the circuit
architecture, the input state $\rho$, and the observable $O$. As such, throughout each example we will specify the conditions over the input state and measured operator that enable efficient simulability.

\subsubsection{Clifford evolutions}

Perhaps the simplest and most well-known example of a classically simulable quantum evolution is a Clifford circuit. Let $\Cl(n)$ denote the $n$-qubit Clifford group, defined as the normalizer of the Pauli group $\mathbb{P}_n=\langle \pm 1,\pm i\rangle\times \langle \{\eye,X,Y,Z\}^{\otimes n}\rangle$. A Clifford unitary maps every Pauli string to another Pauli string under conjugation,
\begin{equation}
U\ad P_\al U = \pm P_{\beta}\,,
\qquad U\in \Cl(n)\,, \quad P_\al,P_\beta\in\mathbb{P}_n\,.
\end{equation}
Consequently, Clifford evolution does not generate a linear combination of Pauli strings and can be represented by a linear transformation over a $2n$-dimensional binary vector space. More specifically, the action of a Clifford unitary can be encoded by a symplectic matrix in $\mathrm{Sp}(2n,\mathbb{F}_2)$.

When the initial state is a stabilizer state, the resulting circuit can be simulated efficiently through the standard Gottesman--Knill framework~\cite{gottesman1998heisenberg,nest2008classical}. In this work, since we allow $\rho$ to be a general quantum state, the circuit evolution can still be efficiently simulated in the Heisenberg picture provided that $O$ is expressed as a linear combination of Pauli operators. Indeed, in such a case one can  backpropagate each operator appearing in $O$ through the Clifford circuit, and estimate the expectation values of the resulting Pauli strings with respect to the initial state.

\subsubsection{Shallow local evolutions}

Another class of evolutions considered in this work consists of shallow circuits generated by geometrically local gates. An archetypal representative of this class is the so-called hardware efficient ansatz (HEA)~\cite{kandala2017hardware,cerezo2020cost} used in variational quantum algorithms and quantum machine learning, which is composed of layers of neighboring two-qubit gates arranged in a brick-like fashion, with circuit depth
\begin{equation}
L\in\OC(\log(n))\,,
\end{equation}
followed by the measurement of a $k$-local observable with $k\in\OC(1)$. More generally, we define the set of $k$-local Pauli operators as:  
\begin{equation}
    \PC=\{P\in\{\id,X,Y,Z\}^{\otimes n}\,|\,|P|\leq k\,,\quad k\in\OC(1)\}\,,
\end{equation}
where here $|P|$ denotes the bodyness of the operator.  Although this definition includes $k$-local observables with arbitrarily separated support, in the subsequent analysis for shallow HEA, we further restrict to geometrically $k$-local observables whose support is contained within a connected region of at most $k$ neighboring qubits in the underlying interaction geometry. 

The classical simulability of these circuits follows from a simple backward-light-cone argument. For a one-dimensional alternating-layer architecture, backpropagating a local observable through one circuit layer can enlarge its support only by a constant number of neighboring qubits. Therefore, after $L$ layers, the observable is supported on at most $\OC(k+L)=\OC(\log(n))$ qubits. Thus, the corresponding reduced Hilbert space has dimension
\begin{equation}
2^{\OC(\log(n))}=\OC(\poly(n))\,,
\end{equation}
and the reduced circuit contained within the light cone can be represented and manipulated using polynomial classical resources~\cite{basheer2023alternating}. Here we can see how, although the full circuit acts on $n$ qubits, the expectation value of a local observable depends only on a logarithmically sized portion of the input state and circuit. This enables an exact CS based on reduced unitary matrices, provided that the relevant reduced information about the input state can be efficiently obtained.

\subsubsection{Matchgate, or free-fermionic, evolutions}

Matchgate circuits form another important family of exactly classically simulable evolutions~\cite{valiant2001quantum,jozsa2008matchgates,terhal2002classical}. As a concrete example, they can be generated by the set of operators
\begin{equation}
\GC_{\rm MG}
=
\left\{Z_i\right\}_{i=1}^n
\bigcup
\left\{X_iX_{i+1}\right\}_{i=1}^{n-1}\;.
\label{eq:matchgate_generators_framework}
\end{equation}
Under the Jordan--Wigner transformation, these circuits correspond to the evolution of non-interacting fermions generated by Hamiltonians that are quadratic in Majorana operators. Indeed, for an $n$-qubit system, we introduce the $2n$ Majorana operators
\begin{align}
    &\gamma_1 = X\eye\cdots \eye\;,~\gamma_3 = ZX\eye\cdots \eye\;,~\gamma_{2n-1} = Z\cdots ZX\;, \nonumber \\
    &\gamma_2 = Y\eye\cdots \eye\;,~\gamma_4 = ZY\eye\cdots \eye\;,~\gamma_{2n} = Z\cdots ZY\;,  
    \label{eq:majorana}
\end{align}
which satisfy the anti-commutation relation, $\{\gamma_\mu, \gamma_\nu\} = 2\dl_{\mu\nu}\eye$ for all $\mu, \nu = 1,\dots, 2n\;$. 

We further define Majorana string $\gamma_{\vec{\mu}}$ (a fermionic analogue of a Pauli string) as an ordered product of Majorana operators
\begin{equation}
    \gamma_{\vec{\mu}} = \gamma_{\mu_1}\gamma_{\mu_2} \cdots \gamma_{\mu_\kappa}\;, \quad \vec{\mu} \in \CC_{2n, \kappa}\,,
    \label{eq:majorana_string}
\end{equation}
where $\kappa$ is referred to as the fermionic locality and $\CC_{2n, \kappa}$ denotes the admissible set of ordered vectors with $n$ modes and fermionic locality $\kappa$ defined as
\begin{equation}
    \CC_{2n, \kappa} = \left\{\vec{\mu} \in \mathbb{N}^{\kappa}  \; \middle| \;  1 \le \mu_1 < \dots < \mu_\kappa \le 2n\right\} \,. 
    \label{eq:CC_def}
\end{equation}Each set of $\kappa$-local Majorana strings forms a basis of the invariant submodule $\BC_\kappa = \spn\{G_{\vec{\mu}}^\kappa\}_{\vec{\mu} \in \CC_{2n, \kappa}}$, where we define $G_{\vec{\mu}}^\kappa = i^{\kappa(\kappa-1)/2} \gamma_{\vec{\mu}}$.  By construction, each subspace $\BC_{\kappa}$ has dimension $\tbinom{2n}{\kappa}$ and the full operator space $\BC$ decomposes as a direct sum, $\BC = \bigoplus_\kappa \BC_\kappa$.  

Under this mapping, any Matchgate circuit generated by $\GC$ corresponds to the evolution under a free-fermionic (Hermitian) Hamiltonian that is quadratic in Majorana strings, given as
\begin{equation}
    H_Q = \frac{i}{2} \sum_{\nu_1 < \nu_2} \Lambda_{\nu_1 \nu_2} \gamma_{\nu_1} \gamma_{\nu_2}\,,
    \label{eq:H_majorana}
\end{equation}
where $\Lambda = \{\Lambda_{\nu_1 \nu_2}\}_{\nu_1, \nu_2 = 1,\dots 2n}$ is a real antisymmetric tensor of size $2n \times 2n$. Here, the subscript $Q$ denotes the associated special orthogonal matrix defined as $Q = e^{\Lambda} \in \textrm{SO}(2n)$, governing the evolution of the system. Moreover, the action of the Matchgate circuit, expressed as $U_Q = e^{-iH_Q}$, reduces to the linear transformation at the level of Majorana strings and preserves the fermionic locality of the Majorana strings, leaving the subspace $\BC_\kappa$ invariant. Explicitly, the Matchgate circuit can be fully characterized by the following adjoint action
\begin{equation}
    U\ad_Q \gamma_{\vec{\mu}} U_Q = \sum_{\vec{\nu}\in \CC_{2n, \kappa}} \det({Q_{\vec{\mu}\vec{\nu}}}) \gamma_{\vec{\nu}}\;,
\end{equation}
where $Q_{\vec{\mu} \vec{\nu}}$ represents the $\kappa \times \kappa$ submatrix of $Q$ with rows indexed by $\vec{\mu}$ and columns by $\vec{\nu}$. As a result, the action of the Matchgate circuit over an observable, expressed as a linear combination of Majorana monomials with fixed locality $\kappa \in \OC(1)$, can be reduced to a linear transformation in the polynomially-sized Majorana subspaces, rather than acting on the full $2^n$-dimensional Hilbert space, which leads to efficient classical simulability. Hence, here we have
\begin{equation}
    \PC=\{G_{\vec{\mu}}^\kappa = i^{\frac{\kappa(\kappa-1)}{2}}\gamma_{\vec{\mu}}\,|\,\vec{\mu} \in \CC_{2n, \kappa}\,, \quad \kappa\in\OC(1)\}\,.
\end{equation}

\subsubsection{$S_n$-equivariant evolutions}

We also consider quantum evolutions that are equivariant under arbitrary permutations of the qubits~\cite{schatzki2022theoretical,anschuetz2022efficient,kazi2023universality,chang2026practical}. Let $R(\pi)$ denote the qubit-permuting representation of a permutation $\pi\in S_n$. That is
\begin{equation}
    R(\pi)\ket{i_1i_2\ldots i_n}=\ket{i_{\pi^{-1}(1)}i_{\pi^{-1}(2)}\ldots i_{\pi^{-1}(n)}}\,.
\end{equation}
Then, we say an operator $A$ is $S_n$-equivariant if
\begin{equation}
[A,R(\pi)]=0
\qquad
\text{for all }\pi\in S_n.
\end{equation}

Next, $S_n$-equivariant quantum evolutions can be obtained when the generators are given by symmetrized Pauli strings, defined as sums of Pauli strings permuted over all qubits. In particular, we restrict our attention to one- or two-local Pauli operators, which encompass most physically motivated problems~\cite{vidal2006concurrence, luo2025hamiltonian}. More concretely, $S_n$-equivariant generators can be written as
\begin{equation}
    \GC_{S_n} = \left\{\frac{1}{n}\sum_{j=1}^n P_j\;, \frac{1}{n(n-1)}\sum^n_{j\neq k } P_jP_k \right\}\,,
    \label{eq:generator_Sn}
\end{equation}
where $P_j \in \{X,Y,Z\} $ denotes a single-qubit Pauli string acting on qubit $j$.

The classical simulability of $S_n$-equivariant evolutions paired with an $S_n$-equivariant observable $O$, e.g.,
\begin{equation}
    \PC=\GC_{S_n}\,.
\end{equation}
follows from their representation-theoretic block structure~\cite{fulton1991representation}. After transforming to the Schur basis, every $S_n$-equivariant operator admits a decomposition into irreducible representation (irrep) blocks as:
\begin{equation}
    A \cong \bigoplus_\lm \eye_{m_\lm} \otimes A_\lm\;, 
\end{equation}
where $A_\lm$ denotes the irrep block of label $\lm$ with size $d_\lm \times d_\lm$, which scales at most $(n + 1)\times(n+1)$, and $m_\lm$ is the multiplicity of the irrep $\lm$. 

Under this change of basis,  $S_n$-equivariant quantum evolutions $U$ and observables $O$ can be decomposed similarly as $U  \cong \bigoplus_\lm (\eye_{m_\lm} \otimes U_\lm)$ and $O \cong \bigoplus_\lm (\eye_{m_\lm} \otimes O_\lm)$ where $U_\lm$ and $O_\lm$ are $d_\lm \times d_\lm$ matrices. Then, the expectation value $f_U(\rho, O)$ defined in Eq.~\eqref{eq:loss} can be expressed as a sum over all the irreps $\lm$ as 
\begin{equation}
    f_U(\rho, O) = \sum_\lm \Tr[\rho_\lm U\ad_\lm O_\lm U_\lm], 
\end{equation}
where $\rho_\lm$ denotes the projection of $\rho$ onto the corresponding irrep subspace, summed over all $m_\lm$ multiplicity labels. In particular, we note that the initial state need not itself be permutation invariant; only its reduced components within the irrep blocks relevant to the observable must be reconstructed. Thus, the full exponentially large unitary evolution is replaced by a collection of polynomial-size matrix evolutions~\cite{chang2026practical}.

\subsubsection{$\U(1)$-equivariant evolutions}

The next family consists of $\U(1)$-equivariant, or number-conserving, quantum circuits. These evolutions commute with the total magnetization operator
\begin{equation}
M_z=\sum_{i=1}^n Z_i,
\qquad
[U,M_z]=0,
\end{equation}
and therefore preserve the Hamming weight of computational basis states. To describe the $\U(1)$-equivariant circuit, we define a set of operators  
\begin{equation}
    \XC =  \{ I, Z, a, a\ad\}^{\otimes n}\;, 
\end{equation}
where $a$ and $a\ad$ are the lowering and raising operators, respectively, defined as
\begin{equation}
a\ad = \frac{X - iY}{2} = |1\rangle\langle 0 |\,, \quad a = \frac{X + iY}{2} = |0\rangle\langle 1 |\;.
\end{equation}
That is, $a\ad$ maps $\ket{0}$ to $\ket{1}$ and hence raises the
Hamming weight by one, while $a$ lowers it. An operator $O \in \XC$ is $\U(1)$-equivariant (or number-conserving) if and only if $n_{+}$ and $n_{-}$, the numbers of $a$ and $a\ad$ in $O$, are equal. We define the set of $\U(1)$-equivariant operators as
\begin{equation}
    \XC_{\U(1)} = \left\{A \in \{ I, Z, a, a\ad \}^{\otimes n} \, \middle\lvert \, n_{a,+} = n_{a,-} \right\}\;.  
    \label{eq:U1_basis}
\end{equation}
We also denote by $n_z$ the number of $Z$ operators in $O$.  This set forms an orthogonal basis that spans all $\U(1)$-equivariant operators on $n$ qubits.

Alternatively, any $\U(1)$-equivariant quantum circuit $U$ can be generated by exponentiation of the generators in  $\GC = \{H^{ij}_{\rm RBS}\}_{ij} $ where $H_{\rm RBS}^{ij}$ is a Hamiltonian acting on qubit $i$ and $j$, defined as
\begin{equation}
\renewcommand\arraystretch{1}
    H_{\rm RBS}^{ij} = i(a_i\ad a_j - a_i a_j\ad) = 
    \begin{pmatrix}
        0 & 0 & 0 & 0 \\
        0 & 0 & i & 0 \\ 
        0 & -i & 0 & 0 \\ 
        0 & 0 & 0 & 0 
    \end{pmatrix}\;.  
\end{equation}         
This Hamiltonian generator is a so-called reconfigurable beam splitter (RBS) gate written as follows
\begin{equation}
\renewcommand\arraystretch{1.2}
G_{ij}(w) = e^{-i w H^{ij}_{\rm RBS}} = 
    \begin{pmatrix}
        1 & 0 & 0 & 0 \\
        0 & \cos(w) & \sin(w) & 0 \\ 
        0 & -\sin(w) & \cos(w) & 0 \\ 
        0 & 0 & 0 & 1 
    \end{pmatrix}\;. 
    \label{eq:givens}
\end{equation}
The set of indices $ij$ depends on the chosen circuit architecture--for example, it may follow the qubit connectivity of the underlying hardware~\cite{monbroussou2025trainability}, or adopt a pyramidal structure~\cite{kerenidis2021classical}.

Under the action of $U(1)$-equivariant circuits, the Hilbert space with computational basis can be decomposed into subspaces with the same Hamming weight $h$
\begin{equation}
    \HC = \bigoplus_{h=0}^n \HC^{(n)}_h = \bigoplus_{h=0}^n \spn(B_h^{(n)})\;, 
\end{equation}
\begin{equation}
    \abs{B^{(n)}_h} = d^{(n)}_h = \binom{n}{h}\;,
\end{equation}
where $B_h^{(n)}$ corresponds to the set of computational basis with Hamming weight $h$. 
The matrix representation of a $\U(1)$-equivariant circuit in the computational basis takes a block-diagonal form, with each block of size $d_h^{(n)}  \times d_h^{(n)}$, constrained to the subspace of a fixed Hamming weight, $h$. For a $\U(1)$-equivariant observable, i.e.,
\begin{equation}
    \PC=\XC_{\U(1)}\,,
\end{equation}
and an initial state with a fixed Hamming weight $h$, the effective dynamics are restricted to the subspace of DLA with dimension $\Theta\left(\binom{n}{h}^2\right)$. Given that $h\in \order{1}$, this scales at most polynomially as $\order{n^{2h}}$, making the simulation classically tractable.  Thus, throughout this work,  $\U(1)$-equivariant evolutions constitute an example where conditions are imposed over the initial state so that the relevant Hamming-weight sector remains polynomially sized.

\subsubsection{Quantum convolutional neural networks}

Quantum convolutional neural networks (QCNNs) were initially proposed as a quantum machine learning architecture for data classification~\cite{cong2019quantum}. In particular, the circuit is composed of alternating convolutional and pooling layers. The convolutional layers consist of local gates acting on nearest neighbors, while the pooling layers progressively reduce the number of active qubits. For an $n$-qubit input, the total number of hierarchical layers scales as
\begin{equation}
L\in\order{\log_2(n)},
\end{equation}
and the final observable acts only on one or two qubits
\begin{equation}
    \PC\in\{\id,X,Y,Z\}^{\otimes 2}\,.
\end{equation}

Under the random-initialization assumptions considered in
Refs.~\cite{pesah2020absence,bermejo2024quantum,angrisani2024classically}, sufficiently high-weight Pauli components can be
truncated while maintaining a controlled approximation error. Moreover,
Ref.~\cite{bermejo2024quantum} provides numerical evidence that this low-weight
description remains accurate throughout training for the benchmark
problems considered there. Thus, under these conditions, QCNN expectation
values admit an efficient approximation based on the propagation of
low-weight observables.

\subsubsection{ADAPT-VQE\label{sec:adapt_vqe_description}}

Finally, we consider the restricted class of adaptive fermionic evolutions arising in ADAPT-VQE for finding the ground state of quantum chemistry Hamiltonians~\cite{grimsley2019adaptive,tang2019qubit,grimsley2022adapt}. For completeness, we recall that the Variational Quantum Eigensolver (VQE) algorithm aims to variationally optimize the parameters in a quantum circuit to prepare the ground state of a given Hamiltonian~\cite{peruzzo2014variational,cerezo2020variationalreview,tilly2022variational}. As such, given an initial state $\rho_0$ and the Hamiltonian $O$ of interest, it optimizes the circuit $U \equiv U(\vec{w})$ to minimize the energy $f_{U(\vec{w})}(\rho_0,O)$. Then, the parameters $\vec{w}$ are iteratively updated using either gradient-based or gradient-free optimization method. In particular, physical Hamiltonians are expressed as linear combinations of constant-order Majorana products, meaning that 
\begin{equation}
    \PC=\{i^{\kappa(\kappa-1)/2}\gamma_{\vec{\mu}}\,|\,\vec{\mu} \in \CC_{2n, \kappa}\,, \quad \kappa\in\OC(1)\}\,.
\end{equation}

The choice of the initial state $\ket{\psi_0}$ depends on the problem of interest, but given that we focus on quantum chemistry-related problems we assume the system is initialized to a Hartree--Fock state $\ket{\psi^{\rm HF}}$~\cite{tilly2022variational, zhang2021variational, arute2020hartree}. Then, among the different variants for the choice of $U(\vec{w})$, the Adaptive Derivative-Assembled Pseudo-Trotter ansatz VQE (ADAPT-VQE)~\cite{grimsley2019adaptive, tang2019qubit, grimsley2022adapt} dynamically constructs the quantum circuit ansatz by repeatedly choosing circuit generators $\{H_\ell\}_{\ell =1}^L$ from an operator pool $\AC$. At each iteration, a new layer is added to the circuit one at a time by selecting the operator $A^* \in \AC$, which yields the largest energy gradient when evaluated with respect to the current trial state. After adding the selected operator, all parameters are re-optimized following the standard VQE optimization procedure.  

While ADAPT-VQE is not classically simulable in full generality, restricting the operators in $\AC$ can lead to the model's simulability. To make these restrictions explicit, we consider a general parameterized fermionic circuit of the form given in Eq.~\eqref{eq:generic_circuit}, where each generator $H_\ell$ is the Majorana string $G_{\vec{\mu}_\ell}^{\kappa_\ell}$ with fermionic locality  $\kappa_\ell$, such that $U(\vec{w}) = \prod_\ell e^{iw_\ell G_{\vec{\mu}_\ell}^{\kappa_\ell}/2}$. 
In particular, we restrict the ADAPT-VQE operator pool to Majorana strings with bounded fermionic locality,
\begin{equation}
\AC
\subseteq
\bigoplus_{\kappa\leq\kappa^*}\BC_\kappa,
\qquad
\kappa^*\in\order{1},
\end{equation}
and assume that the Hamiltonian and measured observables are also sums of low-locality Majorana strings. Accordingly, at each layer $\ell$, the circuit generator is selected from this pool such that $
G^{\kappa_\ell}_{\vec{\mu}_\ell}\in\AC.$
Conjugating one Majorana string by such a gate either leaves it unchanged or produces a linear combination involving another Majorana string. Iteratively propagating the observable therefore generates paths through the space of Majorana strings. The number of such strings grows rapidly in general, but can be controlled by discarding terms whose fermionic locality exceeds a threshold $\tau$~\cite{miller2025simulation}.

\section{Quantum Simulation for estimating expectation values\label{sec:quantum_simulation}}

Given the framework above, the most direct way to estimate
$f_U(\rho,O)$ is to prepare a copy of $\rho$, evolve it under $U$, and
measure the terms appearing in $O$. We refer to this procedure as a
Quantum Simulation, since the quantum device performs the
evolution and supplies the expectation value estimates. The classical
computer is used only to process the measurement outcomes and, when
required, update the circuit parameters.

In this section, we quantify the three resource costs introduced in
Section~\ref{sec:framework}. Importantly, the QS procedure does not
exploit the structure responsible for the classical simulability of the
evolution. The resulting bounds are therefore generic and apply to all
circuit families considered in the previous section.

\subsection{Quantum sample complexity}

Given the prepare-evolve-measure structure of QS, the Quantum Sample
complexity is determined by the measurements performed on
$U\rho U^\dagger$. In general, commuting terms in $O$ can be grouped
and measured simultaneously. To obtain a simple architecture-independent
upper bound, we instead assume that the $M$ terms are measured
independently. This corresponds to a worst-case measurement model, while
the case in which commuting observables are grouped is analyzed in
Appendix~\ref{adx:qsim_quantum_resource}. Since different terms can contribute unequally to $O$, we employ a
weighted measurement strategy in which the number of samples $N_\alpha$
allocated to $P_\alpha$ is proportional to $|c_\alpha|$,
\begin{equation}\label{eq:weighted_measurement_main}
    N_\alpha =
    \frac{N|c_\alpha|}{\sum_\beta |c_\beta|}\;,
\end{equation}
where $N=\sum_\alpha N_\alpha$ is the total number of measurements.

For a target error $\ep$, the following theorem, proved in Appendix~\ref{adx:qsim_quantum_resource},  provides a bound on the Quantum Sample complexity of QS:
\begin{theorem}
    Let $\{U^{(i)}\}_{i=1}^T$ be a quantum evolution, where each $U^{(i)}$ takes the form in~\eqref{eq:generic_circuit} and let  $O$ be an observable as in~\eqref{eq:observable}. Then, the total number of Quantum Samples required to estimate all expectation values $f_{U^{(i)}}(\rho,O)$ up to additive error $\ep$ and with success probability $1-\dl$ is
     \begin{equation}
    N_{\QS}  \in     \mathcal{O}\left(
        \frac{T\|\vec{c}\|_1^2}{\epsilon^2}
        \log\left(\frac{T}{\delta}\right)
    \right)\;. 
    \label{eq:Nqs_quantum}
    \end{equation}

\label{thm:Nqs_quantum}
\end{theorem}

The important feature for the comparison below is that the sampling
cost in QS is incurred for every circuit instance $U^{(i)}$. In contrast,
several of the CS considered in the later sections admit preprocessing
or quantum data-acquisition steps that can be reused across different
parameter choices.

\subsection{Quantum time complexity}

In QS, the quantum time complexity is determined by the circuit depth
required to obtain a single sample. The first contribution comes from
implementing the evolutions $e^{-iw_\ell H_\ell}$ in
Eq.~\eqref{eq:generic_circuit}. The corresponding gate-level depth
depends on both the generators $\mathcal{G}$ and the quantum hardware.
For instance, an evolution may need to be decomposed into native gates,
and additional routing operations may be required to respect the device
connectivity. We denote the resulting hardware-dependent overhead per
layer by $\eta$. Since $U$ contains $L$ layers, the evolution contributes
a depth in $\mathcal{O}(\eta L)$.

The second contribution comes from the basis changes required before
measurement. For all circuit families considered in this work, these
basis changes can be implemented with constant depth. Clifford, shallow
local, Matchgate, $U(1)$-equivariant, QCNN, and ADAPT-VQE settings
ultimately require measurements of Pauli strings. Similarly, the one-
and two-local symmetrized Pauli observables considered for
$S_n$-equivariant circuits can be measured using products of local basis
rotations. Hence, the measurement contribution is in $\mathcal{O}(1)$.

Combining both contributions, the Quantum Time complexity per sample is
therefore
\begin{equation}
    N_{\rm Q.T.}^{(\rm QS)}
    \in \mathcal{O}(\eta L)\;.
\end{equation}
We stress that this is a circuit-depth cost for a single sample. The
total quantum wall-clock time additionally depends on $N_{\rm Q.S.}$ and
on initialization, readout, and reset times, which are incorporated in
Section~\ref{sec:wallclock_time}.

\subsection{Classical time complexity}
Unlike the Quantum Sample and Quantum Time complexities, the classical
cost associated with QS is not determined by the quantum circuit alone.
Rather, it depends on the classical routine wrapped around the quantum
evaluations. 

One contribution, however, is common to all settings: the measurement outcomes must be processed to construct the
required expectation value estimates. Since $N_{\QS}$ already denotes
the total number of circuit shots across all measured terms, a streaming
estimator requires $\order{N_{\QS}}$ elementary operations. A baseline estimate
for the classical work accompanying QS is therefore
\begin{align}
    N_{\CT}^{(\rm QS)}
    &\in
    \order{
        N_{\QS}+N_{\rm aux}
    }
    \nonumber\\
    &\subseteq
    \order{
        \frac{T\|\vec{c}\|_1^2}{\epsilon^2}
        \log\left(\frac{T}{\delta}\right)
        +
        N_{\rm aux}
    },
\end{align}
where $N_{\rm aux}$ denotes the application-dependent classical operations beyond the processing of measurement outcomes. 

For example, in variational settings such as those of quantum machine learning, a
classical optimizer must update the trainable parameters. For an ansatz with $L$ parameters and $T_{\rm opt}$ optimization steps, a standard first-order update contributes $\order{T_{\rm opt}L}$ to $N_{\rm aux}$. Here, $T_{\rm opt}$ is generally distinct from $T$, since each optimization step may require multiple circuit instances to evaluate the objective function and its gradient. The precise classical overhead also depends on the routines employed. Natural-gradient or second-order optimization, classical post-processing for error mitigation, and nontrivial measurement post-processing can increase this cost and should be accounted for separately. Conversely, strategies such as layerwise optimization and random coordinate descent can reduce the per-iteration cost by evaluating and updating only a subset of the trainable parameters, while stochastic gradient descent can reduce the number of measurement samples used for each gradient estimate. Any resulting reduction in the per-iteration cost may, however, be offset by an increase in the number of optimization steps required for convergence.

In the regimes considered below, the elementary classical operations
needed to accumulate measurement outcomes are typically much faster than the corresponding quantum
hardware executions. We therefore expect this contribution to be
subleading in the wall-clock comparison.

\section{Classical Simulations for estimating expectation values~\label{sec:classical_simulation}}

In this section, we provide the computational complexity of the corresponding CS frameworks for each of the classically simulable quantum dynamics presented in Section~\ref{sec:classically_simulable_evolutions}. 
While Theorem~\ref{thm:Nqs_quantum}, which characterized the QS cost, is completely general and agnostic to the specific circuit architecture, estimating the cost of CS requires an architecture-specific analysis, as different circuits call for different surrogate models and simulation techniques. In other words, \textit{there is no single classical simulation algorithm to rule them all}.

We will start by providing a overview of classical methods for simulating quantum algorithms in Section~\ref{sec:classical_methods}, intended primarily for readers who may be unfamiliar with them. Readers with prior knowledge of these methods may skip this section.

\subsection{Background\label{sec:classical_methods}}

There exists a broad range of methods for classically simulating quantum dynamics, including semi-classical approaches~\cite{mink2022hybrid,
vidal2006concurrence}, tensor-network methods~\cite{wood2011tensor,
gibbs2025learning,gibbs2024deep}, and Lie-algebraic
methods~\cite{goh2023lie}. In several of the examples considered below,
the CS reduces directly to standard matrix operations,
and we introduce those methods together with the corresponding circuit
architecture in Section~\ref{sec:classical_simulation}.

Here, we instead briefly review two simulation frameworks whose
implementation requires additional machinery, 
$\liea$-sim and Pauli propagation. Both operate naturally in the Heisenberg picture by
representing the backpropagated observable in a structured operator
basis, thereby avoiding an explicit representation of the full
$2^n$-dimensional quantum state. In particular, $\liea$-sim achieves this by exploiting invariant operator subspaces
induced by the dynamical Lie algebra of the circuit. On the other hand, Pauli propagation works in the Pauli basis, and truncates the propagated operator when needed.  
We provide only the
ingredients needed for the complexity analysis below, with additional
technical details deferred to the Appendices~\ref{adx:pauli_propagation} and \ref{adx:gsim}.

\subsubsection{\label{sec:gsim}$\liea$-sim}

Lie-algebraic simulation, referred to as $\liea$-sim, exploits the
operator-space structure induced by the dynamical Lie algebra (DLA) of
the circuit~\cite{goh2023lie,barligea2026enabling}. Given a set of
Hermitian circuit generators $\mathcal{G}$, we define
\begin{equation}
    \liea
    =
    \spn_{\mathbb{R}}
    \left\langle i\mathcal{G}\right\rangle_{\rm Lie}
    =
    \spn_{\mathbb{R}}
    \{iG_\gamma\}_{\gamma=1}^{\dim(\liea)},
    \label{eq:gsim_dla}
\end{equation}
where $\{iG_\gamma\}_\gamma$ forms a basis of the DLA. The associated
dynamical Lie group $\lieg=\exp(\liea)$ contains the unitaries generated
by the circuit.

To understand the origin of the reduced classical description used in
$\liea$-sim, we consider the action of $\lieg$ on the operator space
$\mathcal{B}$ through conjugation,
\begin{equation}
    A \longmapsto U^\dagger A U,
    \qquad
    U\in\lieg,\quad A\in\mathcal{B}.
\end{equation}
Under this action, the operator space can be decomposed into invariant
subspaces as
\begin{equation}
    \mathcal{B}
    \simeq
    \bigoplus_\lambda \mathcal{B}_\lambda,
    \label{eq:gsim_irrep_decomposition}
\end{equation}
where each $\mathcal{B}_\lambda$ satisfies
\begin{equation}
    U^\dagger A U \in \mathcal{B}_\lambda,
    \qquad
    \forall A\in\mathcal{B}_\lambda,\quad
    \forall U\in\lieg.
\end{equation}
We will refer to the $\mathcal{B}_\lambda$ as irreducible
representations (irreps), although the $\liea$-sim framework only
requires them to be invariant subspaces. We denote their dimensions as
$d_\lambda=\dim(\mathcal{B}_\lambda)$ and introduce a
Hilbert--Schmidt-orthonormal Hermitian basis
$\{B_\alpha^{(\lambda)}\}_{\alpha=1}^{d_\lambda}$ for each subspace.

The usefulness of $\liea$-sim arises when the observable of interest is
supported on one, or a constant number, of irreps whose dimensions grow
at most polynomially with the system size. In that case, the action of
the circuit can be represented inside these reduced operator spaces,
rather than in the full exponentially-large space $\mathcal{B}$.
Importantly, the relevant dimension controlling the simulation is
$d_\lambda$, which need not coincide with $\dim(\liea)$. 

The central idea of $\liea$-sim is to describe how the elements of the
DLA mix the basis operators within a given $\mathcal{B}_\lambda$.
For $iG_\gamma\in\liea$, we define the representation elements
$f_{\alpha\beta}^{(\lambda,\gamma)}$ through
\begin{equation}
    [iG_\gamma,iB_\alpha^{(\lambda)}]
    =
    \sum_{\beta=1}^{d_\lambda}
    f_{\alpha\beta}^{(\lambda,\gamma)}
    iB_\beta^{(\lambda)}.
    \label{eq:gsim_rep_elements}
\end{equation}
Equivalently, these coefficients can be obtained as
\begin{equation}
    f_{\alpha\beta}^{(\lambda,\gamma)}
    =
    \Tr\left[
        iB_\beta^{(\lambda)}
        [iB_\alpha^{(\lambda)},iG_\gamma]
    \right].
\end{equation}
They define the adjoint representation of the Lie algebra on the
$\lambda$-th irrep,
\begin{equation}
    \left(
        \Phi_{\lambda}^{\rm ad}(iG_\gamma)
    \right)_{\alpha\beta}
    \equiv
    f_{\alpha\beta}^{(\lambda,\gamma)},
    \label{eq:gsim_ad_rep}
\end{equation}
where
\begin{equation}
    \Phi_{\lambda}^{\rm ad}(iG_\gamma)
    \in
    \mathbb{R}^{d_\lambda\times d_\lambda}.
\end{equation}
Thus, for a fixed irrep, the DLA is represented by
$\dim(\liea)$ matrices, each of dimension
$d_\lambda\times d_\lambda$.

Through the exponential map, the Lie-algebra representation induces
the corresponding representation of the dynamical Lie group. For a
gate
\begin{equation}
    U_\ell=e^{-iw_\ell H_\ell},
    \qquad iH_\ell\in\liea,
\end{equation}
we define
\begin{equation}
    i\overline{H}^{(\lambda)}_\ell
    =
    \Phi_{\lambda}^{\rm ad}(iH_\ell),
\end{equation}
so that
\begin{equation}
    \Phi_{\lambda}^{\rm Ad}(U_\ell)
    =
    e^{-iw_\ell\overline{H}^{(\lambda)}_\ell}.
    \label{eq:gsim_gate_rep}
\end{equation}
The matrix
$\Phi_{\lambda}^{\rm Ad}(U_\ell)\in
\mathbb{R}^{d_\lambda\times d_\lambda}$
fully characterizes the action of the gate on
$\mathcal{B}_\lambda$. In particular,
\begin{equation}
    U^\dagger
    B_\alpha^{(\lambda)}
    U
    =
    \sum_{\beta=1}^{d_\lambda}
    \left(
        \Phi_{\lambda}^{\rm Ad}(U)
    \right)_{\alpha\beta}
    B_\beta^{(\lambda)}.
    \label{eq:gsim_adjoint_action}
\end{equation}

We can now use this reduced representation to evaluate the expectation
value $f_U(\rho,O)$ directly in the Heisenberg picture. Consider first
an observable supported on a single irrep,
\begin{equation}
    O
    =
    \sum_{\alpha=1}^{d_\lambda}
    c_\alpha B_\alpha^{(\lambda)}.
    \label{eq:gsim_observable}
\end{equation}
The backpropagated observable remains in the same invariant subspace,
\begin{equation}
    \widetilde{O}
    =
    U^\dagger O U
    =
    \sum_{\beta=1}^{d_\lambda}
    \widetilde{c}_\beta B_\beta^{(\lambda)},
    \label{eq:gsim_backprop_observable}
\end{equation}
with coefficients
\begin{equation}
    \widetilde{c}_\beta
    =
    \sum_{\alpha=1}^{d_\lambda}
    c_\alpha
    \left(
        \Phi_{\lambda}^{\rm Ad}(U)
    \right)_{\alpha\beta}.
    \label{eq:gsim_backprop_coefficients}
\end{equation}
Equivalently, if $\vec{c}$ and $\widetilde{\vec{c}}$ are taken as
column vectors,
\begin{equation}
    \widetilde{\vec{c}}
    =
    \left(
        \Phi_{\lambda}^{\rm Ad}(U)
    \right)^T
    \vec{c}.
\end{equation}

The remaining ingredient is the description of the input state on the
same irrep. We define the vector
\begin{equation}
    \left(
        \vec{e}^{\,({\rm in})}_\lambda
    \right)_\alpha
    \equiv
    \Tr\left[
        \rho B_\alpha^{(\lambda)}
    \right],
    \qquad
    \alpha=1,\ldots,d_\lambda.
    \label{eq:gsim_input_rep}
\end{equation}
Depending on the problem, these expectation values can be known from a
classical description of $\rho$, or estimated during an initial data
acquisition stage through direct measurements or a suitable classical
shadow.

Combining Eqs.~\eqref{eq:gsim_backprop_observable}
and~\eqref{eq:gsim_input_rep}, the desired expectation value is
\begin{align}
    f_U(\rho,O)
    &=
    \Tr[\rho U^\dagger O U]
    \nonumber\\
    &=
    \sum_{\alpha=1}^{d_\lambda}
    \widetilde{c}_\alpha
    \Tr\left[
        \rho B_\alpha^{(\lambda)}
    \right]
    \nonumber\\
    &=
    \widetilde{\vec{c}}^{\,T}
    \vec{e}^{\,({\rm in})}_\lambda.
    \label{eq:gsim_expectation}
\end{align}

More generally, an observable can have support on several irreps,
$O=\sum_\lambda O_\lambda$. Since the adjoint action does not mix
distinct invariant subspaces, each $O_\lambda$ can be backpropagated
independently using $\Phi_{\lambda}^{\rm Ad}(U)$, and the final
expectation value is obtained by summing the contributions from the
different irreps. Hence, $\liea$-sim remains efficient when the
observable has support on a constant number of polynomially-sized
irreps and the corresponding input vectors
$\vec{e}^{\,({\rm in})}_\lambda$ can be efficiently obtained.

The Classical Time complexity of $\liea$-sim naturally separates into
a preprocessing step and the repeated evaluation of circuit instances.
Let $K$ denote the number of distinct generators appearing in the
circuit, with $K\leq L$. For each generator, constructing and
diagonalizing its representation on $\mathcal{B}_\lambda$ requires, in
the generic dense implementation, a cost in
$\mathcal{O}(d_\lambda^3)$. Since these matrices depend only on the
generators and not on the circuit parameters, this cost is incurred only
once,
\begin{equation}
    N_{\rm pre}^{(\liea{\rm -sim})}
    \in
    \mathcal{O}(K d_\lambda^3).
\end{equation}

Once these decompositions have been stored, changing the parameters
$w_{i,\ell}$ only requires updating the corresponding diagonal
exponentials and applying the represented gates to the coefficient
vector. The dominant cost for a single $L$-layer circuit instance is
therefore
\begin{equation}
    N_{\rm eval}
    \in
    \mathcal{O}(L d_\lambda^2).
\end{equation}
Hence, evaluating $T$ circuit instances gives
\begin{equation}
    N_{\rm C.T.}
    \in
    \mathcal{O}\left(
        Kd_\lambda^3
        +
        TLd_\lambda^2
    \right)\,,
    \label{eq:gsim_classical_complexity}
\end{equation}
where the first term is a one-time preprocessing cost and can be omitted from subsequent runs once the decompositions have been stored.

\subsubsection{Pauli Propagation\label{sec:pauli_propagation}}

Pauli propagation~\cite{rall2019simulation} provides a practical
framework for classically estimating expectation values for a broad
range of quantum circuits, with rigorous guarantees available in
several noiseless~\cite{rudolph2023classical,angrisani2024classically,
rudolph2025pauli,angrisani2025simulating,rall2019simulation} and
noisy~\cite{fontana2023classical} settings. Its central idea is to
rewrite the expectation value as
\begin{equation}
    f_U(\rho,O)
    =
    \Tr[\rho U^\dagger O U]
\end{equation}
and propagate the observable backwards through the circuit in the
Heisenberg picture, rather than explicitly evolving the quantum state
in the Schr\"odinger picture.

Consider a circuit $U=U_L\cdots U_1$ and initialize
$O^{(L)}=O$. We recursively define
\begin{equation}
    O^{(\ell-1)}
    =
    U_\ell^\dagger O^{(\ell)}U_\ell
    =
    \sum_\alpha
    \tilde c_\alpha^{(\ell-1)}P_\alpha,
    \qquad
    \ell=L,\ldots,1,
    \label{eq:pauli_backprop}
\end{equation}
where $P_\alpha\in\mathcal{P}$. Clifford gates map each Pauli string
onto a single Pauli string, whereas non-Clifford gates such as Pauli
rotations can split one Pauli term into a weighted sum of several
terms. Repeated application of Eq.~\eqref{eq:pauli_backprop} therefore
generates a branching set of Pauli propagation paths.

After all layers have been propagated,
\begin{equation}
    O^{(0)}
    =
    U^\dagger O U
    =
    \sum_\alpha \tilde c_\alpha P_\alpha,
    \label{eq:UOU_pauli_prop}
\end{equation}
and hence
\begin{equation}
    f_U(\rho,O)
    =
    \sum_\alpha
    \tilde c_\alpha \Tr[\rho P_\alpha].
\end{equation}
For a generic circuit, the number of Pauli terms generated by this
procedure can grow exponentially, rendering exact propagation
intractable. Approximate Pauli propagation controls this growth by
discarding terms during the backwards evolution. Two common strategies
are coefficient truncation, in which terms with
$|\tilde c_\alpha|$ below a prescribed cutoff are removed, and weight
truncation, in which Pauli strings with weight larger than a threshold
$\tau$ are discarded.

The accuracy of either truncation depends on the circuit class under
consideration. In particular, low-weight Pauli propagation is efficient
in settings where the contribution of high-weight propagation paths is
sufficiently suppressed. In the applications below, we state the
conditions under which this truncation leads to a controlled
approximation error.

Several extensions of this framework have been developed. These include
symmetry-based methods that merge Pauli strings related by the action of
a symmetry group~\cite{teng2025leveraging}, as well as fermionic
variants in which Majorana strings replace Pauli strings as the basic
propagated objects~\cite{miller2025simulation,rudolph2026thermal}.
These variants will be used below when the structure of the circuit
makes them more efficient than standard Pauli propagation.

\renewcommand{\arraystretch}{2.2}
\begin{table*}[t]
    \centering

    \begin{adjustbox}{width=\textwidth}
\begin{tabular}{|c|c||c|c|c||c|}
    \hline
         Problem & Simulation Methods & Quantum Sample & Quantum Time & Classical Time & Theorem  \\
         \hline
         \hline
          General & Quantum Simulation &
          $\order{\frac{T\|\vec{c}\|_1^2}{\epsilon^2}\log\left(\frac{T}{\delta}\right)}$ &
          $\order{\eta L}$ &
          $\order{N_{\QS}+N_{\rm aux}}$ &
          Theorem~\ref{thm:Nqs_quantum}\\
         \hline
         Clifford circuit & $\rm{Sp}(2n, \mathbb{F}_2)$ matrix &
         $\order{\frac{T\norm{\vec{c}}_1^2}{\ep^2}\log\left(\frac{T}{\dl}\right)}$ &
         $\order{1}$ &
         $\order{N_{\QS}+TLMn^2}$ &
         Theorem~\ref{thm:classical_simulation_clifford}\\
         Shallow HEA & Reduced $U$ &
         $\order{\frac{16^k n^4}{\ep^2}\log\left(\frac{MT}{\dl}\right)\norm{\vec{c}}_1^2}$ &
         $\order{1}$ &
         $\order{M N_{\QS}16^k n^4+MTn^{2\omega}}$ &
         Theorem~\ref{thm:shallow_HEA}\\
         Matchgate & $\liea$-sim~\cite{goh2023lie} &
         $\order{\frac{n^{\kappa/2}}{\epsilon^2}\log\left(\frac{T}{\delta}\right)\norm{\vec{c}}_1^2}$ &
         $\order{n}$ &
         $\order{N_{\QS}n^{\kappa/2}+TLn^{3}+Tn^{\kappa + 1}}$ &
         Theorem~\ref{thm:classical_sim_matchgate}\\
         $S_n$-equivariant & Block-diagonalized $U$~\cite{chang2026practical} &
         $\order{\frac{n^2}{\epsilon^2}\log\left(\frac{T}{\delta}\right)\norm{\vec{c}}_1^2}$ &
         $\order{n\,\poly\left(\log\epsilon_{\rm QST}^{-1}\right)}$ &
         $\order{N_{\QS}n^2+Ln^3+TLn^{\omega+1}}$ &
         Theorem~\ref{thm:classical_sim_Sn}\\
         $\U(1)$-equivariant & Givens rotation~\cite{kerenidis2021classical} &
         $\order{\frac{n^{3h}}{\ep^2}\log\left(\frac{n}{\delta}\right)\norm{\vec{c}}_1^2}$ &
         $\order{1}$ &
         $\order{N_{\QS}n^{h+1}+TLn^{h-1}}$ &
         Theorem~\ref{thm:classical_sim_U1}\\
         QCNN & Pauli Propagation$^*$~\cite{rall2019simulation} &
         $\order{\frac{\exp(\order{\tau})}{\ep^2}\log\left(\frac{n}{\dl}\right)\norm{\vec{c}}_1^2}$ &
         $\order{1}$ &
         $\order{N_{\QS}n^{\tau}+Tn^{2\tau}}$ &
         Theorem~\ref{thm:classical_sim_qcnn}\\
         ADAPT-VQE & Majorana Propagation$^*$~\cite{miller2025simulation} &
         $0$ &
         $0$ &
         $\order{(L \abs{\AC} + T_{\rm VQE})n^{\tau}}$ &
         Theorem~\ref{thm:classical_sim_adapt_vqe}\\
         \hline
    \end{tabular}
\end{adjustbox}
    \caption{\textbf{Simulation algorithm and the simulation cost for the problems that are provably classically simulable.} We summarize the problem instances known to be classically simulable, together with the CS methods and resource costs derived below. Approximate methods (that enjoy average case error guarantees for certain architectures) are marked with $^*$. We consider $T$ circuit instances of the $n$-qubit circuit family in Eq.~\eqref{eq:generic_circuit}, each containing $L$ layers and sharing the same generators but potentially different parameters. All exact methods estimate $f_U(\rho,O)$ to additive error $\ep$ with failure probability at most $\delta$. The exponent $\omega$ denotes the matrix multiplication exponent~\cite{williams2024new}, while $k$ and $\kappa$ denote the Pauli and fermionic locality, $h$ the Hamming weight of the initial-state sector, and $\tau$ the truncation threshold used in propagation methods. Finally, $\epsilon_{\rm QST}$ denotes the implementation error of the quantum Schur transform. The Classical Time column includes the classical processing of quantum measurement data whenever such data acquisition is required. Terms independent of $T$ correspond to one-time preprocessing costs that can be reused across circuit instances. For QS, $N_{\rm aux}$ denotes the application-dependent auxiliary classical operations. For ADAPT-VQE, $T_{\rm VQE}$ denotes the total number of energy evaluations across all VQE re-optimization stages, and the table reports the Hartree--Fock setting in which the initial state is known classically and therefore $N_{\QS}=N_{\QT}=0$.}
    \label{tab:quantum_resource}
\end{table*}

\subsection{Classical Simulations for classically tractable evolutions\label{sec:cs_for_evolutions}}

With these tools in place, we now turn from general CS frameworks to a circuit-specific resource analysis. For each classically simulable evolution introduced in Section~\ref{sec:classically_simulable_evolutions}, we identify an appropriate CS method and derive its computational complexities.

A recurring feature of these methods is that some computational costs can be reused across circuit instances, whereas others must be incurred whenever the circuit instances change. To make this distinction explicit, we separate, whenever relevant, between
a one-time preprocessing cost $N_{\rm pre}$ and the cost
$N_{\rm eval}$ of evaluating a new circuit instance. Thus, over $T$
circuit instances, the total Classical Time complexity takes the form
\begin{equation}
    N_{\rm C.T.}
    =
    N_{\rm pre}
    +
    T N_{\rm eval}.
    \label{eq:preprocessing_evaluation_cost}
\end{equation}
This distinction is important when the classical representation depends
only on the circuit generators and can therefore be reused as the circuit parameters change. In particular, when quantum data are acquired through a classical-shadow protocol,
$N_{\rm pre}$ also includes any classical post-processing required to
construct a reusable description of the initial state. If instead the
same shadow data must be queried separately for the backpropagated
observable associated with each circuit instance, the corresponding
post-processing cost is included in $N_{\rm eval}$. Thus, a cost can be
incurred only once even when its required precision, and hence its
magnitude, depends logarithmically on $T$.

In Table~\ref{tab:quantum_resource}, we summarize the resource and time complexities derived in this work for various classically simulable quantum algorithms. Detailed descriptions of the classical surrogates used for each problem task, along with the proofs of the theorems, are provided in the Appendix.

\subsubsection{Clifford circuit}

We first consider the Clifford evolutions introduced in
Section~\ref{sec:classically_simulable_evolutions}. For the Classical
Simulation, we assume that each Clifford layer is specified by its
symplectic representation in $\mathrm{Sp}(2n,\mathbb{F}_2)$ together
with the corresponding phase data. A Pauli string can then be encoded
as a $2n$-bit vector according to
\begin{equation}
    I \to 00, \qquad
    X \to 01, \qquad
    Z \to 10, \qquad
    Y \to 11 \;,
\end{equation}
and its Heisenberg evolution under each Clifford layer reduces to a
symplectic matrix-vector multiplication.

Since Clifford gates map Pauli strings to Pauli strings without
generating linear combinations, the number of terms appearing in the
observable remains fixed throughout the backpropagation. This leads to
the following theorem.

\begin{theorem}[Complexity of Classical Simulation for Clifford circuits]
\label{thm:classical_simulation_clifford}
Consider $T$ $n$-qubit Clifford circuit instances
$\{U^{(i)}\}_{i=1}^{T}$, each composed of $L$ Clifford gates, and an
observable
\begin{equation}
    O=\sum_{\alpha=1}^{M}c_\alpha P_\alpha,
\end{equation}
where $P_\alpha\in\PC$ and
$\vec{c}=(c_\alpha)_\alpha$ is a real-valued coefficient vector.
Assume that each Clifford layer is specified by its symplectic
representation together with the corresponding phase data.

The expectation values can be estimated from copies of the initial
state with Quantum Sample complexity
\begin{equation}
    N_{\QS}
    \in
    \order{
        \frac{T\norm{\vec{c}}_1^2}{\ep^2}
        \log\left(\frac{T}{\dl}\right)
    }.
    \label{eq:NQS_clifford}
\end{equation}

The Classical Time complexity includes the backpropagation of the
observable and the processing of the measurement outcomes,
\begin{align}
    N_{\CT}
    &\in
    \order{
        N_{\QS}
        +
        TLMn^2
    }
    \nonumber\\
    &\subseteq
    \order{
        \frac{T\norm{\vec{c}}_1^2}{\ep^2}
        \log\left(\frac{T}{\dl}\right)
        +
        TLMn^2
    }.
    \label{eq:NCT_clifford}
\end{align}

The corresponding Quantum Time complexity is
\begin{equation}
    N_{\QT}\in\order{1}.
\end{equation}
\end{theorem}

For each circuit instance, the $M$ Pauli strings must be backpropagated
through its $L$ Clifford layers. A generic
$2n\times2n$ symplectic matrix-vector multiplication requires
$\order{n^2}$ operations, giving $\order{TLMn^2}$ operations over the
$T$ circuit instances. The resulting Pauli expectation values are then
measured directly on the initial state, so the Quantum Sample complexity
has the same scaling as in Theorem~\ref{thm:Nqs_quantum}. Processing
these measurement outcomes requires $\order{N_{\QS}}$ additional
classical operations. Since no quantum evolution is required before
measurement, the Quantum Time complexity is constant.

\subsubsection{Shallow hardware efficient ansatz}

We next consider the shallow hardware efficient ansatz (HEA) evolutions introduced in
Section~\ref{sec:classically_simulable_evolutions}. For each Pauli term
$P_\alpha$ in the observable, the CS is performed by
restricting the circuit to its backward light cone and explicitly
backpropagating $P_\alpha$ using reduced unitary matrices. The required
information about the initial state is obtained once using local Pauli
classical shadows and can be reused across all $T$ circuit instances.

While one could create an HEA on different topologies, we focus for simplicity on a one-dimensional alternating-layer architecture~\cite{cerezo2020cost}. Thus, a $k$-local observable has a backward light cone supported on at most
\begin{equation}
    \widetilde{k}\leq 2L+2k
\end{equation}
qubits. We take $L\leq\log_2(n)$, such that
$\widetilde{k}\leq2\log_2(n)+2k$. The shadow data can be processed once
into empirical reduced-state descriptions on the $M$ fixed backward
light cones and reused as the circuit parameters change. This results
in the following theorem.

\begin{theorem}[Computational complexity of Classical Simulation for shallow HEAs]
\label{thm:shallow_HEA}
Consider $T$ circuit instances $\{U^{(i)}\}_{i=1}^{T}$ of an $n$-qubit
shallow HEA with the same one-dimensional alternating-layer
architecture and depth $L\leq\log_2(n)$, but potentially different
circuit parameters. Let
\begin{equation}
    O=\sum_{\alpha=1}^{M}c_\alpha P_\alpha
\end{equation}
be an observable composed of $M$ $k$-geometrically local Pauli strings, with
$k\in\order{1}$ and
$\vec{c}=(c_\alpha)_\alpha$ a real-valued coefficient vector.

Using local Pauli classical shadows, the Quantum Sample complexity
required to estimate all $T$ expectation values up to additive error
$\epsilon$, with success probability at least $1-\delta$, satisfies
\begin{equation}
    N_{\QS}
    \in
    \order{
        \frac{16^k n^4}{\epsilon^2}
        \log\left(\frac{MT}{\delta}\right)
        \norm{\vec{c}}_1^2
    }.
    \label{eq:classical_shadow_shallow_hea}
\end{equation}

The Classical Time complexity includes the one-time processing of the
shadow data and the repeated evaluation of the parameter-dependent
reduced circuits. A conservative bound is
\begin{align}
    N_{\CT}
    &\in
    \order{
        M N_{\QS}16^k n^4
        +
        MTn^{2\omega}
    }
    \nonumber\\
    &\subseteq
    \order{
        \frac{M\,256^k n^8}{\epsilon^2}
        \log\left(\frac{MT}{\delta}\right)
        \norm{\vec{c}}_1^2
        +
        MTn^{2\omega}
    }.
    \label{eq:NCT_shallow_HEA}
\end{align}
If the reduced input-state descriptions are known classically, the
first term is absent.

The corresponding Quantum Time complexity is
\begin{equation}
    N_{\QT}\in\order{1}.
\end{equation}
\end{theorem}

The first term in Eq.~\eqref{eq:NCT_shallow_HEA} is a one-time
post-processing cost. In contrast, changing the circuit parameters
changes the reduced unitaries within the backward light cones, and the
corresponding backpropagation must be repeated for every circuit
instance. The $n^4$ dependence in the Quantum Sample complexity follows
from the worst-case shadow-norm bound for observables supported on
$\widetilde{k}\leq2\log_2(n)+2k$ qubits. Thus, the result is a
sufficient upper bound; in practice fewer measurements can be
sufficient, as observed numerically in Ref.~\cite{basheer2023alternating}.

\subsubsection{Matchgate circuit\label{sec:matchgate}}

We next consider the Matchgate evolutions introduced in
Section~\ref{sec:classically_simulable_evolutions}. For their Classical
Simulation, we employ $\liea$-sim over the invariant Majorana subspaces
$\BC_\kappa$ defined above. Given a Hermitian observable
\begin{equation}
    O
    =
    \sum_{\vec{\mu}\in\CC_{2n,\kappa}}
    c_{\vec{\mu}}G_{\vec{\mu}}^\kappa
    \in\BC_\kappa
\end{equation}
with fixed even fermionic locality $\kappa\in\order{1}$, the relevant
operator space has dimension
\begin{equation}
    d_\kappa
    =
    \dim(\BC_\kappa)
    =
    \binom{2n}{\kappa}
    \in
    \order{n^\kappa}.
\end{equation}
Thus, $\liea$-sim can be performed by representing the adjoint action
of the Matchgate generators directly on $\BC_\kappa$.

To obtain the required information about the input state, we employ
Matchgate classical shadows~\cite{wan2022matchgate,zhao2021fermionic,
heyraud2024unified}. Rather than estimating each Majorana component
independently, we use the collective Matchgate shadow bound of
Ref.~\cite{west2026fermionic}. For a Hermitian observable
$O\in\BC_\kappa$,
\begin{equation}
    \norm{O}_{\rm sh}^2
    \leq
    \frac{3\binom{2n}{\kappa}}{2\binom{n}{\kappa/2}}
    \norm{O}_\infty^2.
    \label{eq:matchgate_collective_shadow}
\end{equation}
For fixed $\kappa$, the prefactor
$\binom{2n}{\kappa}/\binom{n}{\kappa/2}$ scales as
$\Theta(n^{\kappa/2})$.

For each circuit instance $U^{(i)}$, $\liea$-sim backpropagates the
observable within the same invariant subspace,
\begin{equation}
    \widetilde{O}_i
    =
    (U^{(i)})^\dagger O U^{(i)}
    \in\BC_\kappa.
\end{equation}
This backpropagation can be implemented directly as a sequence of contractions on the rank-$\kappa$ coefficient tensor of the observable, involving intermediate tensors with $\order{n^\kappa}$ entries, without explicitly constructing the full adjoint representation.
Furthermore, since unitary conjugation preserves the operator norm,
$\norm{\widetilde{O}_i}_\infty=\norm{O}_\infty$, the same
Matchgate shadow bound applies to all $T$ circuit instances. The shadow samples are acquired and processed once into reusable statistics for the degree-$\kappa$ Majorana sector, which are subsequently queried for all circuit instances. Combining these contributions yields the following theorem.

\begin{theorem}[Computational complexity for Majorana $\liea$-sim]
\label{thm:classical_sim_matchgate}
Consider $T$ Matchgate circuit instances
$\{U^{(i)}\}_{i=1}^{T}$ of the form given in
Eq.~\eqref{eq:generic_circuit}, sharing the same set of $L$ generators
but potentially different circuit parameters. Let $O\in\BC_\kappa$ be a
Hermitian observable with fixed even fermionic locality
$\kappa\in\order{1}$.

Using Matchgate classical shadows for the initial data acquisition, the
Quantum Sample complexity required to estimate all $T$ expectation
values up to additive error $\epsilon$, with success probability at
least $1-\delta$, satisfies
\begin{align}
    N_{\QS}
    &\in
    \order{
        \frac{\binom{2n}{\kappa}}{\binom{n}{\kappa/2}}
        \frac{\norm{O}_\infty^2}{\epsilon^2}
        \log\left(\frac{T}{\delta}\right)
    }
    \nonumber\\
    &\subseteq
    \order{
        \frac{n^{\kappa/2}}{\epsilon^2}
        \log\left(\frac{T}{\delta}\right)
        \norm{\vec{c}}_1^2
    }.
    \label{eq:NQS_matchgate}
\end{align}

The Classical Time complexity includes the one-time shadow
post-processing and eigendecomposition of the represented generators,
together with the repeated $\liea$-sim evaluations,
\begin{align}
    N_{\CT}
    &\in
    \order{
        N_{\QS}n^{\kappa/2}
        +
        TLn^{3}
        +
        Tn^{\kappa + 1}
    }
    \nonumber\\
    &\subseteq
    \order{
        \frac{n^\kappa}{\epsilon^2}
        \log\left(\frac{T}{\delta}\right)
        \norm{\vec{c}}_1^2
        +
        TLn^{3}
        +
        Tn^{\kappa + 1}
    }.
    \label{eq:NCT_matchgate}
\end{align}
If the required degree-$\kappa$ input-state representation is known
classically, the first term is absent.

The corresponding Quantum Time complexity is
\begin{equation}
    N_{\QT}\in\order{n}.
\end{equation}
\end{theorem}

The first term in Eq.~\eqref{eq:NCT_matchgate} is a one-time
cost associated with processing the Matchgate shadow samples into reusable statistics, which are stored and shared across all circuit instances. The second term accounts for combining the layer transformations independently for the $T$ circuit instances, while the final term is the tensor-contraction cost of backpropagating the observable. For
quartic observables with $\kappa=4$, these three contributions scale
as $\order{N_{\QS}n^2}$, $\order{TLn^3}$, and
$\order{Tn^5}$, respectively.

The details of Majorana $\liea$-sim, Matchgate shadow post-processing,
and the corresponding complexity estimates are provided in
Appendix~\ref{adx:matchgate}.

\subsubsection{$S_n$-equivariant quantum circuit
\label{sec:sn_equivariant}}

We next consider the $S_n$-equivariant evolutions introduced in
Section~\ref{sec:classically_simulable_evolutions}. For their Classical
Simulation, we employ the Schur-basis representation described in
Ref.~\cite{chang2026practical}. As discussed above, an
$S_n$-equivariant operator admits the block decomposition
\begin{equation}
    A
    \cong
    \bigoplus_\lambda
    \eye_{m_\lambda}\otimes A_\lambda,
\end{equation}
where the blocks $A_\lambda$ have dimension
$d_\lambda\leq n+1$.

For each circuit instance, the observable is backpropagated independently
within these blocks,
\begin{equation}
    \widetilde{O}_\lambda
    =
    U_\lambda^\dagger O_\lambda U_\lambda,
\end{equation}
such that
\begin{equation}
    f_U(\rho,O)
    =
    \sum_\lambda
    \Tr[
        \rho_\lambda\widetilde{O}_\lambda
    ],
\end{equation}
where $\rho_\lambda$ denotes the projection of the initial state onto
the corresponding irrep block, summed over the multiplicity labels.
Importantly, the initial state need not itself be $S_n$-equivariant.

For one- and two-local symmetrized Pauli generators, the Schur blocks
can be efficiently constructed and diagonalized. Their
eigendecompositions depend only on the circuit generators and can
therefore be reused for all $T$ circuit instances. For a generic input
state, we obtain the required irrep-block information using deep
Permutation Invariant Classical Shadows (PI-CS)
~\cite{sauvage2024classical,chang2026practical}. The resulting
snapshots can be accumulated once into an empirical block description
of the input state.

\begin{theorem}[Computational Complexity of Classical Simulation of
$S_n$-equivariant circuits]
\label{thm:classical_sim_Sn}
Consider $T$ $n$-qubit $S_n$-equivariant circuit instances
$\{U^{(i)}\}_{i=1}^{T}$, each composed of $L$ one- or two-local
$S_n$-equivariant generators, with the same generators but potentially
different circuit parameters. Let
\begin{equation}
    O=\sum_\alpha c_\alpha P_\alpha
\end{equation}
be an $S_n$-equivariant observable, where the $P_\alpha$ are normalized
symmetrized Pauli operators and
$\vec{c}=(c_\alpha)_\alpha$ is a real-valued coefficient vector.

For a generic initial state, using deep PI-CS, the Quantum Sample
complexity required to estimate all $T$ expectation values up to
additive error $\epsilon$, with success probability at least
$1-\delta$, satisfies
\begin{align}
    N_{\QS}
    &\in
    \order{
        \frac{n^2}{\epsilon^2}
        \log\left(\frac{T}{\delta}\right)
        \norm{O}_\infty^2
    }
    \nonumber\\
    &\subseteq
    \order{
        \frac{n^2}{\epsilon^2}
        \log\left(\frac{T}{\delta}\right)
        \norm{\vec{c}}_1^2
    }.
    \label{eq:NQS_Sn}
\end{align}

The Classical Time complexity includes the one-time processing of the
PI-CS data and the circuit generators, together with the repeated
evaluation of the $T$ circuit instances. A conservative bound is
\begin{align}
    N_{\CT}
    &\in
    \order{
        N_{\QS}n^2
        +
        Ln^3
        +
        TLn^{\omega+1}
    }
    \nonumber\\
    &\subseteq
    \order{
        \frac{n^4}{\epsilon^2}
        \log\left(\frac{T}{\delta}\right)
        \norm{O}_\infty^2
        +
        Ln^3
        +
        TLn^{\omega+1}
    }.
    \label{eq:NCT_Sn}
\end{align}
If the required irrep-block description of the initial state is known
classically, the first term is absent.

The corresponding Quantum Time complexity is
\begin{equation}
    N_{\QT}
    \in
    \order{
        n\,\poly\left(
            \log\epsilon_{\rm QST}^{-1}
        \right)
    },
    \label{eq:NQT_Sn}
\end{equation}
where $\epsilon_{\rm QST}$ denotes the implementation error of the
quantum Schur transform.
\end{theorem}

The first two terms in Eq.~\eqref{eq:NCT_Sn} are one-time costs. The
deep PI-CS data are processed once, while the Schur blocks of the
generators can be diagonalized and stored since the generators are
shared by all $T$ circuit instances. The final term corresponds to the
parameter-dependent block matrix operations. In the second line of
Eq.~\eqref{eq:NQS_Sn}, we used
$\norm{O}_\infty\leq\norm{\vec{c}}_1$, since the normalized
symmetrized Pauli operators considered here satisfy
$\norm{P_\alpha}_\infty\leq1$. We refer the reader to
Appendix~\ref{adx:Sn_QNN} for additional details and to
Ref.~\cite{chang2026practical} for the full representation-theoretic
construction.

\subsubsection{$\U(1)$-equivariant circuit
\label{sec:U1_symmetric}}

We next consider the $\U(1)$-equivariant evolutions introduced in
Section~\ref{sec:classically_simulable_evolutions}. We focus on the
setting in which the initial state is pure and supported on a fixed
Hamming-weight sector $\HC_h^{(n)}$, with $h\in\order{1}$. In this
case, the evolution can be simulated directly within the
$d_h^{(n)}=\binom{n}{h}$-dimensional subspace by applying the sequence
of Givens rotations to the corresponding state vector.

Consider an RBS gate acting on qubits $i$ and $j$. Within
$\HC_h^{(n)}$, the gate acts non-trivially only on pairs of
computational basis states for which exactly one of the two qubits is
occupied. There are
\begin{equation}
    \binom{n-2}{h-1}
\end{equation}
such pairs. Hence, applying one Givens rotation to a state in
$\HC_h^{(n)}$ requires
$\order{\binom{n-2}{h-1}}$ elementary operations. The locations of
these pairs depend only on the circuit generators and can be determined
once, while the corresponding rotations must be reevaluated as the
circuit parameters change.

If the initial state is not known classically, we reconstruct it within
the fixed Hamming-weight sector using $\U(1)$-symmetric classical
shadows~\cite{hearth2024efficient}. The resulting description of the
initial state is acquired only once and can be reused for all $T$
circuit instances. This leads to the following theorem.

\begin{theorem}[Classical Simulation Complexity of
$\U(1)$-equivariant circuits]
\label{thm:classical_sim_U1}

Consider $T$ $\U(1)$-equivariant circuit instances
$\{U^{(i)}\}_{i=1}^{T}$ of the form given in
Eq.~\eqref{eq:generic_circuit}, sharing the same set of $L$ RBS
generators but potentially different circuit parameters. We assume that
the initial state is pure and lies entirely within a fixed
Hamming-weight sector $h\in\order{1}$. Let
$O=\sum_\alpha c_\alpha P_\alpha$ be a $\U(1)$-equivariant observable,
where $\norm{P_\alpha}_\infty\leq1$ and
$\vec{c}=(c_\alpha)_\alpha$ is a real-valued coefficient vector.

If the initial state is accessed only through copies of $\rho$, it can
be reconstructed within the fixed Hamming-weight sector using
$\U(1)$-symmetric classical shadows with Quantum Sample complexity
\begin{equation}
    N_{\QS}
    \in
    \order{
        \frac{n^{3h}}{\epsilon^2}
        \log\left(\frac{n}{\delta}\right)
        \norm{\vec{c}}_1^2
    }.
    \label{eq:NQS_U1}
\end{equation}
The same reconstructed initial state can be reused for all $T$ circuit
instances, and therefore the Quantum Sample complexity does not acquire
an additional dependence on $T$. If the initial state is known
classically, no quantum data acquisition is required.

The Classical Time complexity naturally separates into the one-time
post-processing required to reconstruct the initial state and the
repeated evaluation of the $T$ circuit instances. When the initial
state is reconstructed from quantum measurements, a conservative bound
is
\begin{align}
    N_{\CT}
    &\in
    \order{
        N_{\QS} n^{h+1}
        +
        TL\binom{n-2}{h-1}
    }
    \nonumber\\
    &\subseteq
    \order{
        \frac{n^{4h+1}}{\epsilon^2}
        \log\left(\frac{n}{\delta}\right)
        \norm{\vec{c}}_1^2
        +
        TLn^{h-1}
    }.
    \label{eq:NCT_U1}
\end{align}
If the initial state is known classically, the first term is absent.

Assuming that the disjoint two-qubit gates entering the
$\U(1)$-symmetric shadow protocol can be implemented in parallel, the
corresponding Quantum Time complexity is
\begin{equation}
    N_{\QT}\in\order{1}.
\end{equation}
\end{theorem}

The linear dependence on $T$ in Eq.~\eqref{eq:NCT_U1} arises because
the rotation angles change between circuit instances, and the
parameter-dependent Givens rotations must therefore be reapplied. In
contrast, the fixed-Hamming-weight description of the initial state is
obtained only once and reused throughout all subsequent circuit
evaluations.

We stress that Theorem~\ref{thm:classical_sim_U1} relies on the
assumption that $h\in\order{1}$ and that the initial state is pure and
supported entirely within a single Hamming-weight sector. When these
conditions are not satisfied, alternative CS methods
can be employed, including $\liea$-sim~\cite{barligea2026enabling},
Matchgate simulation, or symmetry-adapted Pauli
Propagation~\cite{teng2025leveraging}.

\subsubsection{Quantum convolutional neural networks\label{sec:qcnn}}

We next consider the QCNN evolutions introduced in
Section~\ref{sec:classically_simulable_evolutions}. For their Classical
Simulation, we employ low-weight Pauli Propagation. The hierarchical
architecture is important here as after each pooling layer, the number of
active qubits is reduced by a factor of two. As a result, the number of
low-weight Pauli transitions that must be computed also decreases
geometrically through the circuit.

We truncate the backpropagated observable after each layer by discarding
Pauli strings with weight larger than $\tau$. Under the locally
scrambling assumptions used in Ref.~\cite{angrisani2024classically},
the corresponding mean-squared truncation error satisfies
\begin{equation}
    \Ebb_U\left[
        \abs{f_U(\rho,O)-f_U^{(\tau)}(\rho,O)}^2
    \right]
    \leq
    \left(\frac{2}{3}\right)^{\tau+1}
    \norm{O}_{\rm Pauli,2}^2,
    \label{eq:qcnn_truncation_error_main}
\end{equation}
where $\norm{O}_{\rm Pauli,2}=2^{-n/2}\norm{O}_2$. This is an
average-case guarantee around random circuit instances. Ref.~\cite{bermejo2024quantum}
provides complementary numerical evidence that the low-weight
description remains accurate throughout training for the benchmark
problems considered there, including simulations with up to 1024
qubits.

Here, the total estimation error with respect to the exact expectation value can be decomposed into the approximation error $\epsilon_{\rm PP}$ arising from the Pauli path truncation, and the statistical error $\epsilon$ arising from the finite number of classical shadow measurements. More precisely, let $\tilde{f}_U^{(\tau)}(\rho, O)$ denote the estimator obtained  by combining a classical shadow estimate of the initial state with low-weight Pauli propagation at truncation threshold $\tau$. The total estimation error then satisfies
\begin{align}
    \abs{\tilde{f}^{(\tau)}_U - f_U}  \le \abs{\tilde{f}^{(\tau)}_U - f^{(\tau)}_U} + \abs{f^{(\tau)}_U- f_U} \;,
\end{align}
where we suppress the explicit dependence on $\rho$ and $O$ for notational simplicity. The first term represents the statistical error due to the finite classical shadow data, whereas the second represents the approximation
error due to the Pauli-path truncation. 
The numerical results of Ref.~\cite{bermejo2024quantum} indicate that the low-weight approximation retains sufficient information to achieve accurate classification on the benchmark problems considered therein. However, beyond the average-case bound for random circuit ensembles, no explicit instance-wise guarantee has been provided yet for the trained circuits.

The remaining input-dependent quantities are expectation values of
low-weight Pauli strings on $\rho$. These can be acquired once using
local randomized Pauli measurements and reused for all subsequent
parameter values. This leads to the following theorem.

\begin{theorem}[Computational complexity of Classical Simulation for randomly initialized QCNNs]
\label{thm:classical_sim_qcnn}
Consider $T$ $n$-qubit QCNN circuit instances $\{U^{(i)}\}_{i=1}^T$ with
the same hierarchical architecture and depth $L=\log_2(n)$, but
potentially different circuit parameters. Let
$O=\sum_\alpha c_\alpha P_\alpha$ be a constant-local observable and let
$f_{U^{(i)}}^{(\tau)}(\rho,O)$ denote the low-weight Pauli Propagation
estimate obtained with truncation threshold $\tau$.

Under the locally scrambling assumptions above, randomized local Pauli
measurements can be used for the one-time data acquisition. The Quantum
Sample complexity required to estimate the truncated QCNN expectation
values to statistical accuracy $\epsilon$, with failure probability at
most $\delta$, satisfies
\begin{equation}
    N_{\QS}
    \in
    \order{
        \frac{\exp(\order{\tau})}{\epsilon^2}
        \log\left(\frac{n}{\delta}\right)
        \norm{\vec{c}}_1^2
    }.
    \label{eq:NQS_qcnn}
\end{equation}
The same quantum data are reused for all $T$ circuit instances.

If the shadow data are processed into all Pauli expectation values with
weight at most $\tau$, the Classical Time complexity can be bounded as
\begin{equation}
    N_{\CT}
    \in
    \order{
        N_{\QS}n^\tau
        +
        Tn^{2\tau}
    }.
    \label{eq:NCT_qcnn}
\end{equation}
The first term is a one-time data-processing cost, while the second term
comes from the parameter-dependent Pauli propagation.

The corresponding Quantum Time complexity is
\begin{equation}
    N_{\QT}\in\order{1}.
    \label{eq:NQT_qcnn}
\end{equation}
\end{theorem}

Theorem~\ref{thm:classical_sim_qcnn} separates the two errors entering
the QCNN surrogate. Equation~\eqref{eq:NQS_qcnn} controls the
statistical error in estimating the truncated expectation value, while
Eq.~\eqref{eq:qcnn_truncation_error_main} controls the approximation
error introduced by removing high-weight Pauli components. The latter
is the step that makes the simulation approximate rather than exact.
Additional details are given in Appendix~\ref{adx:qcnn}.

\subsubsection{ADAPT-VQE for quantum chemistry\label{sec:adapt_vqe}} 

We finally consider the restricted ADAPT-VQE setting introduced in
Section~\ref{sec:classically_simulable_evolutions}. For its Classical
Simulation, we employ Majorana Propagation
~\cite{miller2025simulation,d2025majorana}, which backpropagates the
Hamiltonian and gradient observables in the Majorana basis and discards
strings whose fermionic locality exceeds a threshold $\tau$. As for the
QCNN setting considered above, this truncation makes the simulation
approximate rather than exact, and its efficiency relies on conditions
under which the contribution of high-locality Majorana strings can be
controlled.

In particular, Ref.~\cite{miller2025simulation} derives an error
guarantee for a randomized fermionic circuit model relevant to
ADAPT-VQE. Consider a fermionic circuit described in Sec.~\ref{sec:adapt_vqe_description}, composed of $L$ fermionic gates generated by the Majorana strings, $G_{\vec{\mu}_\ell}^{\kappa_\ell}$,  
and an observable $O = \sum_{\kappa} \sum_{\vec{\mu} \in \CC_{2n, \kappa}} c_{\vec{\mu}} G^\kappa_{\vec{\mu}}$, which consists only of low, even Majorana degree $\kappa$ for most of the physically relevant problems.
For the randomized
ensemble considered in Ref.~\cite{miller2025simulation}, in which every circuit generator has fermionic locality $\kappa_\ell=4$, and for an
observable $O$ whose coefficients are unbiased and
uncorrelated, the mean-squared truncation error satisfies
\begin{align}
    \mathbb{E}& \left[
        \left(
            f_U(\rho,O)
             -
            f_U^{(\tau)}(\rho,O)
        \right)^2
    \right]
    \nonumber \\
    & \leq \left( 
    \frac{1}{2^{n-1}} + 
    \left(
        \frac{e\tau}{n}
    \right)^{\tau/2}
    \right) \mathbb{E}\left[
        \majoranaNorm{O}^2
    \right],
    \label{eq:adapt_majorana_error}
\end{align}
where $n$ denotes  the number of fermionic modes, which equals the number of qubits. Here, 
\begin{equation}
    \mathbb{E}\left[
        \majoranaNorm{O}^2
    \right]
    =
    \mathbb{E}\left[
        \frac{\Tr[OO^\dagger]}{2^n}
    \right]
\end{equation}
is the squared $\ell_2$-norm of the Majorana coefficients averaged over
the randomness in $O$. Accordingly, choosing
$\tau\in\order{\log(\epsilon_{\rm MP}^{-1}\delta_{\rm MP}^{-1})}$ yields a controlled
approximation with an additive error at most $\epsilon_{\rm MP}$ and success probability at least $1 - \delta_{\rm MP}$, under the assumptions of
Ref.~\cite{miller2025simulation}, with polynomial runtime for fixed $\epsilon_{\rm MP}$ and $\delta_{\rm MP}$. More generally, numerical evidence in
Ref.~\cite{chakraborty2026scalable} suggests that Majorana Propagation
can remain accurate at moderate system sizes (up to 100 qubits) even when these randomness
assumptions are relaxed. Here, a similar decomposition into approximation and statistical errors applies, as in the standard Pauli Propagation discussed for QCNNs in Sec.~\ref{sec:qcnn} (though, of course, for classical initial states such as the Hartree--Fock state no shadow procedure is needed and so there are no statistical errors).

Two kinds of evaluations enter ADAPT-VQE. At each adaptive step, the
gradients associated with the full operator pool $\AC$ must be
computed. Once a new generator is selected, the parameters of the
current ansatz are re-optimized by repeatedly evaluating the energy. We
denote by $L$ the final number of selected generators and by
$T_{\rm VQE}$ the total number of energy evaluations accumulated over
all re-optimization stages. Since each new generator is appended to the
adaptive construction, the previously generated Majorana propagation
paths can be reused rather than reconstructed from the final observable
at every step.

For the quantum-chemistry setting considered here, the initial
Hartree--Fock state is known classically. Its Majorana expectation
values can therefore be computed without an initial quantum
data-acquisition stage. This leads to the following theorem.

\begin{theorem}[Computational complexity of Classical Simulation for
randomized ADAPT-VQE]
\label{thm:classical_sim_adapt_vqe}
Consider the randomized fermionic ADAPT-VQE setting described above,
with Majorana operator pool $\AC$, final circuit depth $L$, and
truncation threshold $\tau$. Assume that the initial state is the
Hartree--Fock state and let $T_{\rm VQE}$ denote the total number of
energy evaluations performed across all VQE re-optimization stages.

No copies of the initial state are required, and hence
\begin{equation}
    N_{\QS}=0.
    \label{eq:NQS_adapt_vqe}
\end{equation}

The Classical Time complexity of the truncated adaptive procedure is
\begin{equation}
    N_{\CT}
    \in
    \order{
        \left(
            L\abs{\AC}
            +
            T_{\rm VQE}
        \right)
        n^\tau
    }.
    \label{eq:NCT_adapt_vqe}
\end{equation}
The first contribution is the cost of constructing and extending the
Majorana surrogates required for gradient screening, while the second
is the cost of reevaluating the truncated surrogate during the VQE
optimizations.

Since the CS uses no quantum data acquisition in this
setting, its Quantum Time complexity is
\begin{equation}
    N_{\QT}=0.
    \label{eq:NQT_adapt_vqe}
\end{equation}
\end{theorem}

\renewcommand{\arraystretch}{1.5}
\begin{table*}[t]
    \centering
    \begin{tabular}{c|c|c|c|c|c|c|c|c}
        Provider & QPU family & Hardware type & $\#$ qubits &
        Connectivity & $t_{\rm 1q}$ & $t_{\rm 2q}$ &
        $t_{\rm readout}$ & $t_{\rm reset}$ \\
        \hline
        IonQ~\cite{robertson2026variational,chen2024benchmarking,
        hughes2025trapped}
        & Forte & Trapped-ion & 36 qubits & All-to-all
        & 130 \unit{\micro\second}
        & 900 \unit{\micro\second}
        & 150 \unit{\micro\second}
        & 50 \unit{\micro\second}
        \\
        IQM~\cite{abdurakhimov2024technology,IQMGarnetHardware,
        marxer2026above}
        & Garnet & Superconducting & 20 qubits & Square lattice
        & 20 ns & 40 ns & 500 ns & 300 \unit{\micro\second}
        \\
        Rigetti~\cite{CepheusSpecSheet}
        & Cepheus & Superconducting & 108 qubits &
        Square lattice chiplets
        & 40 ns & 60 ns & N/R& 250 \unit{\micro\second}
        \\
        Quantinuum~\cite{ransford2025helios}
        & Helios & Trapped-ion & 98 qubits & All-to-all
        & 50 \unit{\micro\second}
        & 370 \unit{\micro\second}
        & N/R & N/R
        \\
        Google~\cite{google2019supremacy,bengtsson2024model,
        GoogleDevices}
        & Sycamore & Superconducting & 53 qubits & Square lattice
        & 25 ns & 30 ns & 500 ns & 200 \unit{\micro\second}
        \\
        IBM$^{*}$~\cite{IBMHeronSpec}
        & Heron & Superconducting & 156 qubits & Heavy Hexagon
        & 30 ns & 90 ns & 3 \unit{\micro\second}
        & 20 \unit{\micro\second}
    \end{tabular}
    \caption{\textbf{Estimated specifications of the quantum hardware.}
    All hardware specifications are taken from publicly available
    references, and information unavailable from these sources is marked as N/R (not reported). The reported values are approximate and may depend on
    device configuration, calibration, and circuit implementation.
    When a device supports only passive reset, we estimate the reset
    time as $10t_{1}$, where $t_{1}$ is the qubit relaxation time. Devices
    that support active reset are marked with an asterisk (*). For IBM
    Heron processors, the reset time, referred to as the repetition
    delay time, can be selected by the user~\cite{IBMupdate} between
    20--2000~\unit{\micro\second}; here, we use the lowest value
    of 20~\unit{\micro\second}. For Quantinuum Helios, a complete
    circuit layer takes approximately 40~\unit{\milli\second} when ion
    transport, cooling, and other operational overheads are included.}
    \label{tab:hardware_specification}
\end{table*}
Theorem~\ref{thm:classical_sim_adapt_vqe} separates the two ingredients
entering the surrogate. Eq.~\eqref{eq:NCT_adapt_vqe} quantifies the
runtime of the truncated Majorana Propagation algorithm, while
Eq.~\eqref{eq:adapt_majorana_error} controls the approximation error
under the randomized assumptions of Ref.~\cite{miller2025simulation}.
The latter is the step that makes the rigorous simulability statement
restricted to the randomized ADAPT-VQE setting considered here.

The zero Quantum Sample and Quantum Time costs rely separately on the
classically known Hartree--Fock input. For a generic unknown initial
state, an additional fermionic tomography or shadow protocol would have
to be included. Further details on the propagation cost and truncation
guarantee are provided in Appendix~\ref{adx:adapt_vqe}.

\section{Wall-clock time and monetary cost
\label{sec:practical_cost}}

The polynomial scalings derived in the previous sections establish that
the models considered here are classically simulable, but they do not
determine which implementation is preferable at finite system size. In
practice, there are at least two distinct questions: which approach
reaches the desired result faster, and which requires fewer monetary
resources to do so. These two comparisons need not lead to the same
answer. A faster implementation can be more expensive, while a slower
one can remain preferable when quantum hardware access is costly.

We therefore compare QS and CS from these two
complementary perspectives. First, we convert the resource counts in
Table~\ref{tab:quantum_resource} into estimated wall-clock times using
representative quantum and classical hardware specifications. We then
estimate the monetary cost associated with quantum hardware access using
publicly available pricing models. In both cases, the goal is not to
provide a hardware benchmark, but rather to identify how the different
resource scalings derived above translate into finite-size simulation
costs.

\subsection{Wall-clock time
\label{sec:wallclock_time}}

We first consider the time-to-solution. For a simulation strategy
$X\in\{{\rm QS},{\rm CS}\}$, we model the total wall-clock time as
\begin{equation}
    t_{\rm sim}^{(X)}
    =
    t_{\rm Q}^{(X)}N_{\QS}^{(X)}
    +
    t_{\rm C}N_{\CT}^{(X)},
    \label{eq:wallclock_time}
\end{equation}
where $t_{\rm Q}^{(X)}$ denotes the execution time associated with one
quantum sample and $t_{\rm C}$ the time required for one elementary
classical operation. The quantum execution time is estimated as
\begin{equation}
    t_{\rm Q}^{(X)}
    =
    t_{\rm init}
    +
    D_{\rm 1q}^{(X)}t_{\rm 1q}
    +
    D_{\rm 2q}^{(X)}t_{\rm 2q}
    +
    t_{\rm readout}
    +
    t_{\rm reset},
    \label{eq:quantum_wallclock}
\end{equation}
where $D_{\rm 1q}^{(X)}$ and $D_{\rm 2q}^{(X)}$ denote the effective
sequential one- and two-qubit gate depths after accounting for
parallelization and routing. These quantities provide a
hardware-resolved implementation of the Quantum Time complexity
$N_{\QT}$ introduced in Section~\ref{sec:framework}. The hardware
parameters used below are summarized in
Table~\ref{tab:hardware_specification}.

For passive reset, the qubits are allowed to relax naturally to their
ground state, while active reset uses additional control operations to
reduce the reset time~\cite{divincenzo2000physical}. Within the model
of Eq.~\eqref{eq:wallclock_time}, QS exhibits a
wall-clock advantage whenever
\begin{equation}
    t_{\rm sim}^{({\rm QS})}
    <
    t_{\rm sim}^{({\rm CS})}.
    \label{eq:wallclock_advantage}
\end{equation}

To evaluate this comparison, we consider three examples that realize
qualitatively different CS costs. In all cases,
expectation values are evaluated at every circuit instance. Two-qubit
gates acting on disjoint qubits are executed in parallel, while
commuting observables are measured simultaneously.

\begin{figure*}
    \centering
    \includegraphics[width=\linewidth, trim = {0, 2cm, 0, 0}]{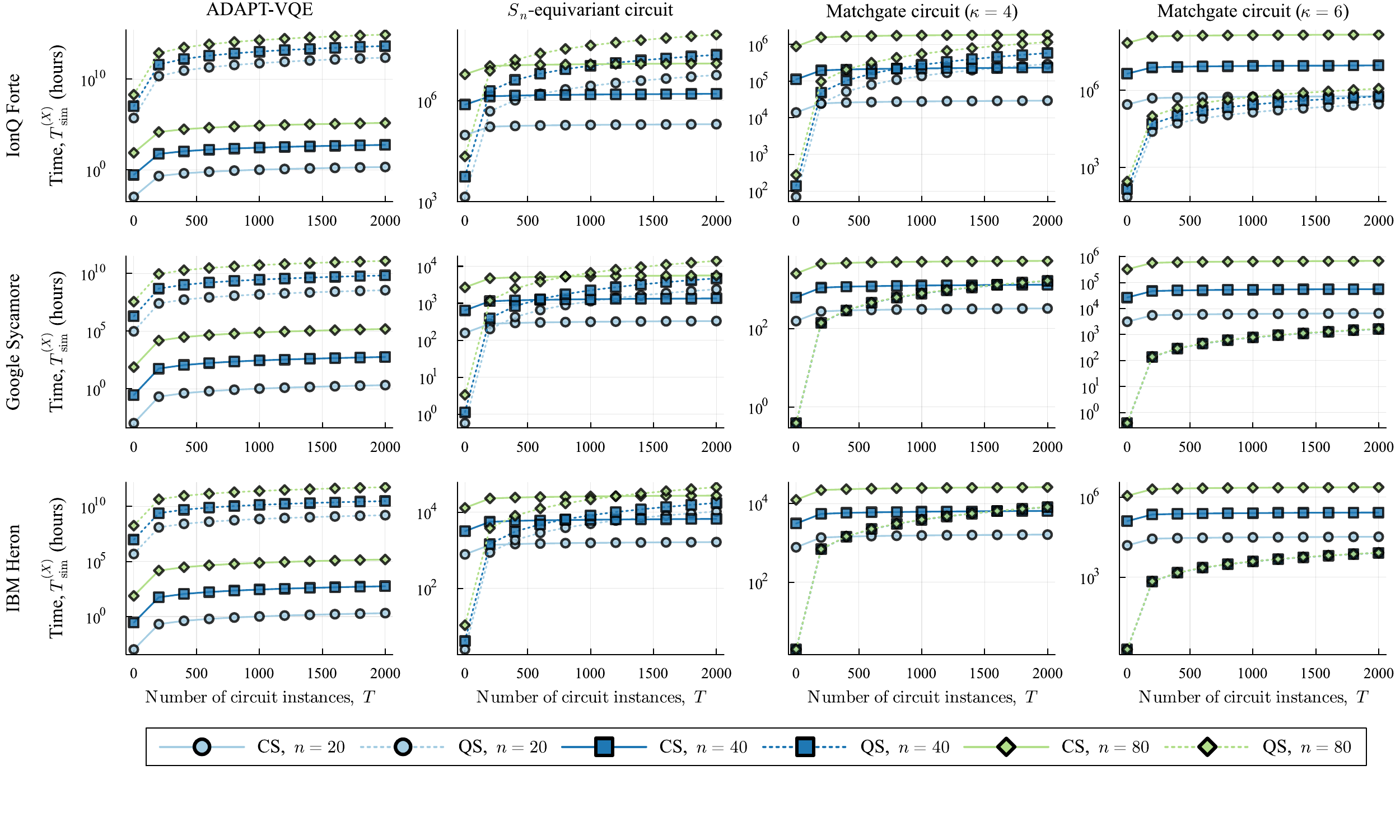}
    \caption{\textbf{Estimated wall-clock time for QS and CS
    on IonQ Forte, Google Sycamore, and IBM
    Heron.} Here, the runtime corresponds to the total runtime $t_{\rm sim}^{(X)}$ given by Eq.~\eqref{eq:wallclock_time}, which sums both the runtime on the quantum and classical hardware. Rows correspond to the three quantum hardware platforms, while all classical operations are assumed to run on the same classical hardware. The horizontal axis denotes the total number of circuit instances $T$. Columns correspond to ADAPT-VQE, $S_n$-equivariant
    circuits, and Matchgate circuits with $\kappa=4$ and $6$, while rows
    correspond to the three hardware models. Solid lines denote
    CS and dotted lines QS. Marker
    shapes indicate the system sizes $n=20,40,80$. The estimates use
    the hardware specifications in
    Table~\ref{tab:hardware_specification} and assume
    $t_{\rm C}=0.25$~ns per elementary classical operation. For
    superconducting devices, we include the routing overhead used in
    the numerical estimate arising from the fixed qubit connectivity.
    Two-qubit gates acting on disjoint qubits are assumed to be
    parallelized. Results exceeding the available device size, namely
    $n>36$ for IonQ Forte and $n>53$ for Google Sycamore, extrapolate
    the reported hardware specifications and should therefore be
    interpreted only as scaling estimates.}
    \label{fig:wallclock_time}
\end{figure*}

\begin{enumerate}[leftmargin=*]
    \item \textbf{ADAPT-VQE
    (Section~\ref{sec:adapt_vqe}).}
    This setting provides an example in which the CS
    requires no quantum data acquisition, since the Hartree--Fock
    initial state is known classically. We consider an operator pool
    containing quadratic and quartic Majorana strings,
    \begin{equation}
        \AC
        =
        \BC_2\oplus\BC_4,
        \qquad
        \abs{\AC}
        =
        \binom{2n}{2}
        +
        \binom{2n}{4}.
    \end{equation}
    The Hamiltonian is taken to be quadratic,
    \begin{equation}
        O
        =
        \sum_{\vec{\mu}\in\CC_{2n,2}}
        c_{\vec{\mu}}G_{\vec{\mu}}^2,
        \qquad
        \norm{\vec{c}}_2=1.
    \end{equation}
    For the wall-clock comparison, the horizontal axis corresponds to
    the repeated energy-evaluation count, which plays the role of
    $T_{\rm VQE}$ in Theorem~\ref{thm:classical_sim_adapt_vqe}.

    \item \textbf{$S_n$-equivariant circuit
    (Section~\ref{sec:sn_equivariant}).}
    This setting combines a one-time PI-CS data acquisition cost with
    repeated evolution of the reduced Schur blocks. We consider
    $S_n$-equivariant circuits generated by one- and two-local
    symmetrized Pauli operators. A two-body symmetrized generator
    contains $\order{n^2}$ two-qubit interactions before
    parallelization. For the observables, we consider the normalized one-local and same-axis two-local symmetrized Pauli operators, normalized such that their coefficient $\ell_1$-norm equals one, together with the global Pauli strings $P^{\otimes n}$ with $P \in \{X,Y,Z \} $. More precisely, we define
\begin{equation}
\chi = \left\{ \frac{1}{n}\sum_{i=1}^{n}P_i, \frac{2}{n(n-1)} \sum_{ i<j }P_iP_j, P^{\otimes n} \right\}.  \end{equation}
We take $O\in\chi$, where every observable in $\chi$ is normalized such that the $\ell_1$-norm of its Pauli coefficient vector satisfies $\lVert\vec{c}\rVert_1=1$.
    \item \textbf{Matchgate circuit
    (Section~\ref{sec:matchgate}).}
    This example isolates the effect of an expensive, but polynomial,
    one-time preprocessing step. We consider a homogeneous degree-$\kappa$
    Majorana observable,
    \begin{equation}
        O
        =
        \sum_{\vec{\mu}\in\CC_{2n,\kappa}}
        c_{\vec{\mu}}G_{\vec{\mu}}^\kappa,
        \qquad
        \norm{\vec{c}}_1=1,
    \end{equation}
    with $\kappa = 4$ and $\kappa = 6$. 
    Each Matchgate layer contains $\order{n}$ local two-qubit
    operations. For $\kappa=4$, the Majorana $\liea$-sim
    processing cost in Theorem~\ref{thm:classical_sim_matchgate} scales
    as $\order{Tn^{5}}$, while for $\kappa = 6$, it scales as $\order{Tn^7}$. In the following simulations, we will show how this fermionic locality strongly influences the runtime cost of the CS.

\end{enumerate}

Figure~\ref{fig:wallclock_time} shows the resulting estimates for IonQ
Forte, Google Sycamore, and IBM Heron, with
$\epsilon=0.001$ and $\delta=0.001$. 
 For the classical
computation, we assume an elementary operation time of
$t_{\rm C}=0.25$~ns. Here we can see that classical simulability does
not imply that CS is faster at finite system size.
For ADAPT-VQE, the CS is faster throughout the
parameter range considered here. In this case, the Hartree--Fock input
state is known classically, so the surrogate incurs neither quantum
data acquisition nor quantum circuit execution. Therefore, the runtime only consists of the runtime on the classical hardware. The $S_n$-equivariant
example similarly benefits from the reduced Schur-block representation,
although the one-time PI-CS data acquisition becomes increasingly
relevant as the system size grows.  

The Matchgate example exhibits qualitatively different regimes depending on the fermionic locality $\kappa$. Before discussing this dependence, we note that, for Google Sycamore and IBM Heron, the QS curves corresponding to different system sizes nearly overlap. Under the normalization $\norm{\vec{c}}_1=1$, $N_{\QS}$ is independent of $n$, as are the readout and reset times incurred for each sample. The only $n$-dependent contribution arises from the quantum circuit depth, $L$, which scales as $\order{n}$ in the present setting. However, because the two-qubit gate times for these devices are much shorter than their readout and reset times, this contribution remains subleading, making the curves difficult to distinguish on the logarithmic scale. The QS costs are also independent of $\kappa$, since the fermionic locality of the observable affects the corresponding CS but not the quantum circuit executed under the assumptions considered here.

Turning to the CS, for
$\kappa=4$, the $\order{Tn^{5}}$ tensor-contraction cost required to backpropagate the observable in 
the Majorana $\liea$-sim implementation rapidly becomes the dominant
classical cost. Nevertheless, as $T$ increases, the repeated sampling and circuit-execution costs of QS eventually dominate, while the one-time quantum data acquisition cost of CS is amortized over the circuit instances. For $\kappa=6$, this backpropagation cost increases to $\order{Tn^7}$, and the Matchgate shadow sample and post-processing costs also exhibit a stronger dependence on $n$. Over the parameter ranges considered here, both the quantum data-acquisition and classical processing contributions to CS can exceed the total runtime of QS.

These examples demonstrate that the finite-size runtime depends strongly on the fermionic locality of the measured observable. In particular, for the largest system sizes considered
in Fig.~\ref{fig:wallclock_time}, QS can be faster than
the corresponding CS despite the latter having
polynomial asymptotic complexity. At intermediate system sizes, the
relative advantage depends on both the hardware platform and the
number of circuit instances. The relevant distinction is consequently not only whether the CS has polynomial runtime, but also the polynomial degree of its dominant cost and whether that cost is incurred once or at every circuit evaluation.

As a final remark, we note that, in all cases, the coefficient-dependent prefactor $\norm{\vec{c}}_1^2$ appearing in the Quantum Sample complexity is set to unity for simplicity. For observables with larger coefficient norms, the absolute sampling requirements would increase accordingly by the corresponding constant factor.  Nevertheless, since the same dependence enters $N_{\QS}$ of both QS and CS, as summarized in Table~\ref{tab:quantum_resource}, this choice provides a controlled comparison of their sampling costs. However, the purely classical preprocessing and surrogate evaluation contributions to $N_{\CT}$ do not generally scale with $\norm{\vec{c}}_1$. Varying this norm can therefore change the balance between quantum sampling and classical computation, thereby shifting the quantitative crossover between the two implementations.

\subsection{Monetary cost of quantum hardware access
\label{sec:monetary_cost}}

Wall-clock time captures only one aspect of the practical resource
comparison. Access to present-day quantum hardware also carries a
monetary cost that depends strongly on the number of circuit
executions. Consequently, an implementation that reaches the solution
more quickly need not be the least expensive one.

The pricing models of the hardware platforms considered here are
summarized in Table~\ref{tab:pricing}. Some providers charge a fixed
task fee together with a price per circuit shot, while others charge
according to the total QPU execution time. In either case, the Quantum
Sample complexity directly controls a major contribution to the
monetary cost.

\renewcommand{\arraystretch}{1.5}
\begin{table}[t]
    \centering
    \begin{tabular}{c|c|c|c}
        Provider & QPU family & Per-task price & Per-shot price \\
        \hline
        IonQ~\cite{amazon_braket_pricing}
        & Forte & \$0.3 & \$0.08000
        \\
        IQM~\cite{amazon_braket_pricing}
        & Garnet & \$0.3 & \$0.00145
        \\
        Rigetti~\cite{amazon_braket_pricing}
        & Cepheus & \$0.3 & \$0.000425
        \\
        Quantinuum
        & Helios & \multicolumn{2}{c}{Non uniform}
        \\
        IBM~\cite{ibm_pricing}
        & Heron & \multicolumn{2}{c}{\$48/minute (Premium Plan)}
        \\
        Google
        & Sycamore & \multicolumn{2}{c}{Not publicly available}
    \end{tabular}
    \caption{\textbf{Pricing details for the quantum hardware platforms
    considered in this study.} IonQ Forte, IQM Garnet, and Rigetti
    Cepheus are accessible through Amazon Braket and use a
    task-plus-shot pricing model. IBM Heron is priced according to QPU
    execution time under the Premium Plan. All pricing information was retrieved in August 2026, at the time of manuscript preparation.}
    \label{tab:pricing}
\end{table}

\begin{figure*}[t]
    \centering
    \includegraphics[width=\linewidth]{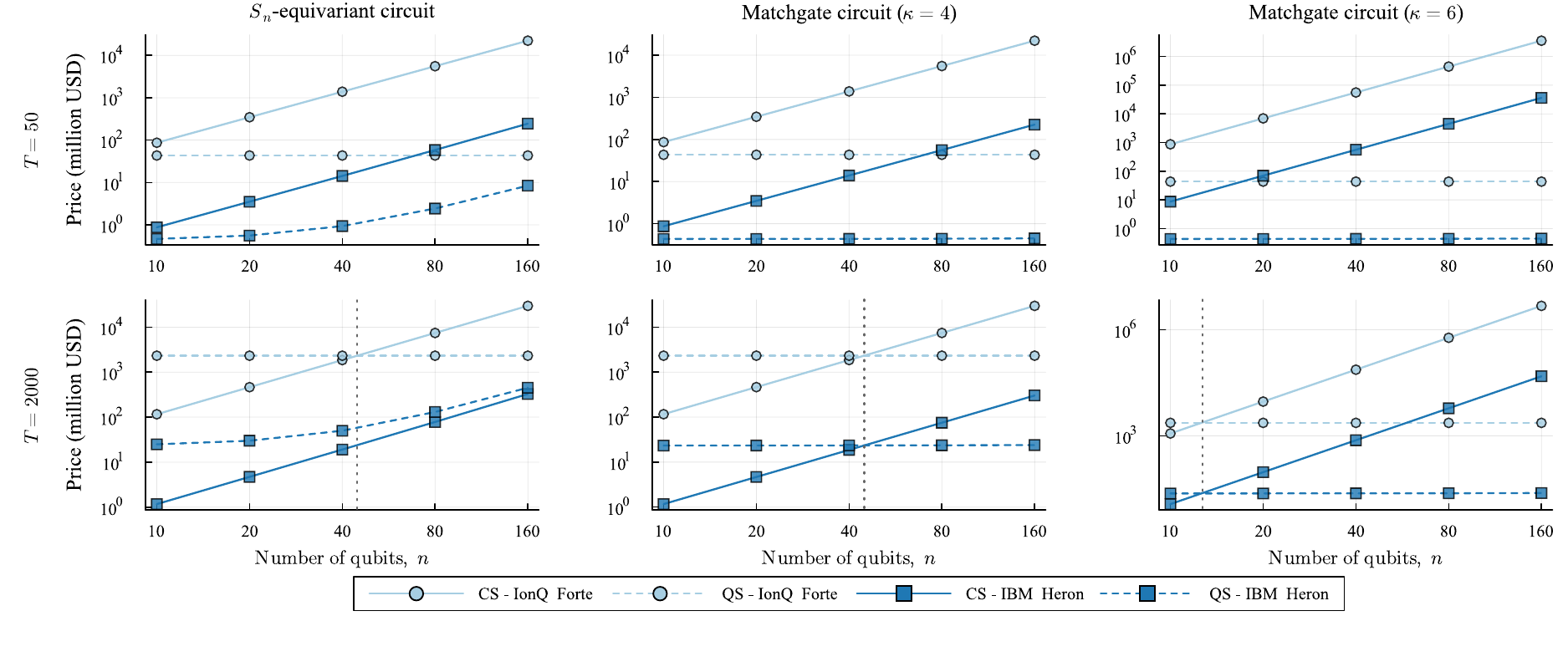}
    \caption{\textbf{Estimated quantum hardware access cost of QS and CS over $T = 50$ and $T=2000$ circuit instances for IonQ Forte and IBM Heron.}
    We consider the $S_n$-equivariant and Matchgate examples with $\kappa = 4$ and $\kappa = 6$ described
    above, using the publicly available pricing data summarized in
    Table~\ref{tab:pricing}. The plotted values account for
    quantum hardware access and do not include the monetary cost of the
    classical computation. Gray dashed lines indicate the crossover points between QS and CS. Results exceeding the available device size, namely
    $n>36$ for IonQ Forte and $n>156$ for IBM Heron, extrapolate
    the reported hardware pricing and should therefore be
    interpreted only as scaling estimates. }
    \label{fig:pricing}
\end{figure*}
For a platform employing a per-shot pricing model, the quantum hardware
cost of a simulation strategy $X$ can be written schematically as
\begin{equation}
    C_{\rm QPU}^{(X)}
    =
    N_{\rm task}^{(X)}C_{\rm task}
    +
    N_{\QS}^{(X)}C_{\rm shot},
    \label{eq:monetary_cost_shot}
\end{equation}
where $N_{\rm task}^{(X)}$ is the number of submitted tasks, $C_{\rm task}$ is the cost per task and $C_{\rm shot}$ is the cost per shot for each of the tasks. Here, the tasks refer to the different circuit designs that need to be evaluated on the quantum hardware.  For
time-based pricing, the corresponding cost is proportional to the total
quantum execution time,
\begin{equation}
    C_{\rm QPU}^{(X)}
    =
    R_{\rm QPU}
    t_{\rm Q}^{(X)}
    N_{\QS}^{(X)},
    \label{eq:monetary_cost_time}
\end{equation}
where $R_{\rm QPU}$ denotes the monetary rate per unit QPU time.

Importantly, Eqs.~\eqref{eq:monetary_cost_shot} and
\eqref{eq:monetary_cost_time} quantify the monetary cost of
\emph{quantum hardware access}. We do not assign an explicit monetary
price to the classical operations entering $N_{\CT}$, since this cost
depends strongly on the CPU or GPU platform, memory requirements,
parallelization, and whether the classical hardware is locally owned or
accessed through a cloud service. For simplicity, we assume that classical compute time is free and executed on hardware one already owns. Thus, the comparison below should not
be interpreted as a complete accounting of total economic cost.

The ADAPT-VQE Classical Simulation considered above requires no access
to quantum hardware, since the Hartree--Fock initial state is known
classically. Its QPU-access cost is therefore zero under our
assumptions. We consequently focus the nontrivial monetary comparison
on the $S_n$-equivariant and Matchgate examples, for which the Classical
Simulation itself requires an initial quantum data-acquisition stage.

Figure~\ref{fig:pricing} shows that the monetary comparison need not
follow the wall-clock comparison. In particular, CS
can reduce the repeated quantum execution time while still requiring a
large one-time quantum data-acquisition budget, particularly exhibiting consistently higher cost than the quantum simulation for a smaller number of circuit instances, $T = 50$, while for larger $T = 2000$, this crossover between quantum and classical simulation occurs at the higher number of qubits. For the
$S_n$-equivariant example, this produces regimes in which the reduced
Schur-block simulation is favorable in wall-clock time, while the
deep PI-CS acquisition can make the corresponding quantum hardware
access more expensive. Thus, the faster implementation need not be the
cheaper one.

The same competition appears for Matchgate circuits, but for a
different reason. For $\kappa = 4$, the wall-clock time is eventually dominated by
the $\order{Tn^{5}}$ tensor-contraction cost required to backpropagate the observable 
whereas the monetary cost of the CS is controlled by
the Matchgate shadow sample complexity rather than by this classical
processing. The two comparisons therefore probe different
bottlenecks. A method can be unfavorable in wall-clock time because of
a large classical polynomial while remaining competitive in
quantum hardware access, or conversely incur a large quantum
data acquisition cost despite a fast subsequent classical simulation.

Increasing the fermionic locality to $\kappa=6$ strengthens both effects: the backpropagation cost increases to $\order{Tn^7}$, while the Matchgate shadow sample complexity acquires a stronger dependence on $n$. Therefore, the quantum hardware access cost of CS exceeds that of QS at a smaller number of qubits than for $\kappa=4$. For $T=50$, CS is consistently more expensive over all system sizes considered.

Taken together, Figs.~\ref{fig:wallclock_time} and
\ref{fig:pricing} show that classical simulability, wall-clock
advantage, and monetary advantage are distinct notions. A polynomial
CS need not be the fastest implementation at the
system sizes of interest, while a faster implementation need not be the
least expensive one. Which approach is preferable therefore depends
not only on the asymptotic scaling, but also on how the resources are
distributed between one-time preprocessing, repeated circuit
evaluation, quantum sampling, and quantum hardware access. For completeness,
Fig.~\ref{fig:pricing_adx} in Appendix~\ref{adx:monetary_cost_adx} provides the full pricing comparison
across all quantum hardware platforms listed in Table~\ref{tab:pricing}.

\subsection{Limitations of the hardware-level comparison}

The estimates above are intended as controlled resource comparisons rather than direct hardware benchmarks and therefore rely on several simplifying assumptions.

First, in evaluating the quantum execution time, we retain the dominant contributions $D_{\rm 2q}^{(X)}t_{\rm 2q}$, $t_{\rm readout}$, and $t_{\rm reset}$ in Eq.~\eqref{eq:quantum_wallclock}, while neglecting $t_{\rm init}$ and $D_{\rm 1q}^{(X)}t_{\rm 1q}$. Although this approximation is reasonable for the leading-order comparison, omitting the single-qubit contribution can affect more precise runtime estimates. As shown in Table~\ref{tab:hardware_specification}, $t_{\rm 1q}$ and $t_{\rm 2q}$ are comparable for some devices; depending on the corresponding compiled circuit depths, $D_{\rm 1q}^{(X)}t_{\rm 1q}$ may therefore be non-negligible relative to $D_{\rm 2q}^{(X)}t_{\rm 2q}$ and should be included in a detailed device-level analysis. These estimates also omit constant factors in the asymptotic scaling
and several implementation-dependent costs, including compiler-specific
gate merging and routing, queueing and control overheads, error
mitigation, and quantum error correction~\cite{ruiz2025quantum}. For
example, the QCCD architecture used in Quantinuum's Helios device
incorporates ion shuttling, batching, routing, and cooling into its
layer execution time~\cite{ransford2025helios}, making a decomposition
based only on elementary gate times insufficient.

More specifically, we emphasize that Eq.~\eqref{eq:wallclock_time} measures the total busy time obtained when a single quantum processor and a single classical processing stream are used serially. It should therefore not be interpreted as a time-to-solution estimate in the presence of parallel computational resources. Both the quantum and classical device contributions admit parallelization. In particular, independent quantum samples can be distributed across multiple QPUs, while the dense linear-algebra operations dominating $N_{\CT}$ can be parallelized across multiple CPU cores or hardware accelerators such as GPUs. The resulting speedup, however, depends on the available computational resources and the associated communication and synchronization overheads.

Moreover, those estimates neglect the effects of hardware noise. Sufficiently deep circuits may not be executable with useful accuracy on the current noisy devices considered here, while error mitigation would generally increase both the number of circuit executions and the required classical post-processing. Similarly, in a fault tolerant setting, encoding logical qubits and implementing logical gates would introduce substantial space and time overheads, with logical gate times potentially exceeding the physical gate times used in our estimates. These effects can shift the quantitative crossovers identified here, although their precise impact depends on the hardware architecture and error correction scheme.

The estimates also neglect latency and communication overheads between quantum and classical hardware. Such costs can become non-negligible for algorithms requiring repeated quantum--classical feedback, such as variational algorithms involving classical gradient computation and
parameter updates~\cite{karalekas2020quantum}.

Similarly, the classical estimate $t_{\rm C}=0.25$~ns does not account
for memory access, communication, cache effects, or the
hardware-dependent parallelization of dense linear-algebra operations.
The monetary comparison likewise does not assign a price to this
classical computation. These effects can shift the quantitative
crossovers in Figs.~\ref{fig:wallclock_time} and
\ref{fig:pricing}. They do not change the mechanism responsible for
them: QS repeatedly incurs quantum sampling and circuit
execution costs, while CS can instead incur large
one-time quantum data-acquisition or classical preprocessing costs that
are subsequently amortized over the circuit instances.

Furthermore, our estimates are based on asymptotic resource scalings and therefore omit constant prefactors. At finite system sizes, these prefactors can substantially affect the quantitative resource estimates and potentially alter the relative performance of the simulation methods. This limitation is particularly relevant to classical shadow sample-complexity bounds, which can involve large and potentially conservative prefactors.

Finally, we stress that the number of circuit evaluations $T$ need not be the same for QS and CS. For example, in variational algorithms, estimating gradients on quantum hardware through the parameter-shift rule requires multiple shifted circuit evaluations for each parameter. A differentiable classical surrogate may instead permit automatic differentiation without separately evaluating each shifted parameter setting, although the associated differentiation cost must still be included in $N_{\CT}$. Therefore, $T$ can be larger for QS than for CS, even when both methods perform the same number of optimization steps. In such cases, assuming the same value of $T$ for both implementations can underestimate the wall-clock time and quantum hardware access cost of QS relative to CS.

\section{Discussion and conclusion\label{sec:conclusion}}

In this work, we have asked whether the existence of an efficient
Classical Simulation implies that one should actually use it instead of
executing the corresponding evolution on quantum hardware. Our results
show that polynomial classical simulability alone is not sufficient to
answer this question. We analyzed Quantum and Classical Simulations
using three resource metrics, Quantum Sample, Quantum Time, and
Classical Time complexity, and derived these costs for several widely
studied classically simulable circuit families. A recurring distinction
throughout the analysis is between resources that must be spent once,
such as constructing a classical representation or acquiring
information about the input state, and those that must be paid again
for every circuit instance. As a result, two simulations that are both
efficient in the complexity-theoretic sense can have very different
costs at finite system size.

This distinction becomes particularly important when the polynomial
degree of the CS is large. Some of the methods
considered here require expensive preprocessing, followed by
comparatively inexpensive evaluations of new circuit parameters. In
contrast, QS avoids this classical preprocessing but
repeatedly incurs the cost of preparing, evolving, and measuring the
quantum state. The number of circuit instances $T$ therefore plays a
central role in determining which approach is preferable, since
one-time classical costs can eventually be amortized over repeated
evaluations. The estimates in Section~\ref{sec:practical_cost} exhibit
both regimes. They also show that an advantage in wall-clock time does
not necessarily imply an advantage in monetary cost, since repeated
access to quantum hardware can remain expensive even when the quantum
execution is faster.

Another important conclusion concerns the role of the input quantum
state. Classical simulability is generally a property of the complete
estimation problem, including the circuit $U$, the state $\rho$, and
the observable $O$, rather than of the evolution alone. When $\rho$ is
known classically, as for the Hartree--Fock state in the ADAPT-VQE
setting considered here, no quantum data acquisition is required. For
an unknown quantum state, however, the classical surrogate must be
supplied with enough information about $\rho$ to evaluate the relevant
backpropagated observables. Classical shadows become particularly
useful in this setting because the required information can often be
acquired once and reused across many circuit instances. Thus, the
efficiency of a CS can depend as much on how the
initial state is characterized as on how efficiently the evolution
itself can be propagated.

There are several important limitations to our conclusions. First, the
Classical Time complexities derived here are upper bounds for the
particular algorithms that we analyze, rather than lower bounds on the
best possible CS. A different representation or a
more specialized algorithm could therefore move, or potentially remove,
the finite-size crossovers identified here. Establishing lower bounds
on the classical resources required under fixed assumptions on
$\rho$, $U$, and $O$ would clarify whether these crossovers are
intrinsic to the estimation problem or specific to the simulation
algorithms currently available. Second, our analysis focuses on
expectation-value estimation. Efficiently estimating
$f_U(\rho,O)$ does not imply that one can efficiently sample from the
complete output distribution of the same circuit, and recent results
have identified settings in which these two tasks have different
classical complexities~\cite{recio2025train,lerch2026iqp}. Our
conclusions should therefore be interpreted for the estimation problems
defined in Section~\ref{sec:framework}.

Finally, the hardware-level analysis in
Section~\ref{sec:practical_cost} should be regarded as a finite-resource
comparison rather than a definitive hardware benchmark. Compilation,
memory and communication costs, classical parallelization, error
mitigation, and pricing models can all shift the quantitative
crossovers. Moreover, the comparison can change substantially in a
fault-tolerant setting. Error-corrected quantum computation introduces
large additional space and time overheads, increasing the cost of each
logical circuit execution and potentially favoring CS
in regimes where present-day hardware estimates suggest otherwise. A
systematic comparison including fault-tolerant resource requirements is
therefore an important direction for future work.

As a side remark, we emphasize that our resource comparison does not systematically account for the space complexity of QS and CS. Memory requirements may be particularly important for CS and can be divided into two broad categories: persistent storage that must be retained and reused across circuit instances, such as classical shadow data or a precomputed Pauli propagation surrogate, and the peak working memory required to evaluate an individual circuit instance.

For example, for a Matchgate circuit with $\kappa=4$ and $n=80$, the sample complexity estimate obtained using the parameters chosen in the previous section gives approximately $9.3\times10^{10}$ shadow samples. As a conservative illustration, retaining all outcome bit strings would require approximately $0.93$ TB under optimal bit packing, or $7.4$ TB if each single-qubit outcome is stored as an unsigned 8-bit integer, even before accounting for the measurement settings and other metadata.
Substantial working memory may also be required during observable propagation to construct an antisymmetric Majorana coefficient tensor, which contains $\order{n^\kappa}$ independent components. For $n=80$, storing one dense vector of these coefficients in 64-bit floating-point format requires approximately $0.21$ GB for $\kappa=4$ and $170$ GB for $\kappa=6$\footnote{The peak memory depends strongly on the implementation and can be reduced by exploiting sparsity or using streamed contraction strategies. These estimates should therefore be understood as representative storage costs rather than fundamental lower bounds.}. Related memory bottlenecks also arise in Pauli propagation. Pauli propagation surrogates trade a large initial investment in memory for faster repeated evaluations, and Ref.~\cite{rudolph2025pauli} notes that practical Pauli propagation simulations can exhaust $1$ TB of memory within only a few hours. This further illustrates how memory, rather than runtime, can become the limiting computational resource.

In this sense, while time is money, memory is a {wall}. Additional runtime can often be accommodated at a financial cost, whereas insufficient memory imposes a hard physical constraint on the feasibility of the simulation.

Overall, determining that a quantum evolution is classically simulable
should be viewed as the beginning, rather than the end, of the resource
comparison. The relevant question is whether the resulting polynomial
algorithm is sufficiently inexpensive under the state-access, accuracy,
and circuit-evaluation requirements of the problem. Understanding when
these costs are unavoidable, and when they can be reduced by improved
simulation or measurement protocols, will be necessary to determine
whether a classically simulable quantum computation should actually be
run classically.

\section*{Artificial Intelligence Disclosure}
The authors acknowledge the use of ChatGPT 5.6 Sol for writing and reviewing the manuscript, as well as for improving existing code. All scientific content, analyses and conclusions were independently verified by the authors.

\section*{Acknowledgments}
 We thank Sumner Hearth, Joachim Favre, and Pablo Bermejo for the valuable discussions. S.Y.C. and M.C. were supported by Laboratory Directed Research and Development (LDRD) program of Los Alamos National Laboratory (LANL) under project number 20260043DR and by the U.S. Department of Energy (DOE), Office of Science, Office of Advanced Scientific Computing Research under Contract No. DE-AC05-00OR22725 through the Accelerated Research in Quantum Computing Program MACH-Q project. 
 Z.H. acknowledges support from the Sandoz Family Foundation-Monique de Meuron program for Academic Promotion. 
 S.T. acknowledges Exchange Faculty Travel Grant: TG168033 from Chulalongkorn University, as well as funding from National Research Council of Thailand (NRCT) [grant number N42A680126]. ST further acknowledges Thailand Science research and Innovation Fund Chulalongkorn University (IND\_FF\_69\_258\_2300\_062).
 M.C. acknowledges support from LANL's ASC Beyond Moore’s Law project. This work was also supported by the U.S. DOE through a quantum computing program sponsored by the LANL Information Science \& Technology Institute.

\bibliography{quantum, accessed_webpage}

\clearpage
\newpage

\onecolumngrid
\appendix
\makeatletter
\let\savedaddcontentsline\addcontentsline
\renewcommand{\addcontentsline}[3]{%
  \def\firstargument{#1}%
  \def\tocargument{toc}%
  \ifx\firstargument\tocargument
  \else
    \savedaddcontentsline{#1}{#2}{#3}%
  \fi
}
\makeatother

\section*{Appendices}

In the Appendices, we provide additional details on the Classical Simulation methods and the derivations of the resource bounds stated in the main text.
For convenience, we also summarize the notation used in the main theorems.

\begin{table}[h]
\centering

\begin{tabular}{ll|ll}
\hline
\textbf{Symbol} & \textbf{Description} & \textbf{Symbol} & \textbf{Description} \\
\hline

$N_{\QS}$ & Quantum Sample complexity 
& $N_{\CT}$ & Classical Time complexity \\

$N_{\QT}$ & Quantum Time complexity & 
$N_{\rm aux}$ & Application-dependent classical operations in QS \\ 

$M$ & Number of orthogonal operators in observable & 
$L$ & Circuit depth \\

$T$ & Number of circuit instances &
$f_U(\rho, O)$ & Expectation value \\

$\epsilon, \delta$ & Error and failure probability 
& $\epsilon_{\rm QST}$ & Quantum Schur transform implementation error \\

$P_{\vec{\al}}$ & Pauli string & 
$\gamma_{\vec{\mu}}$ & Majorana string \\

$k$ & Pauli locality & 
$\kappa$ &  Fermionic locality \\

$h$ & Hamming weight 
& $\omega$ & Matrix multiplication exponent \\

$\tau$ & Propagation truncation threshold
& $T_{\rm VQE}$ & Total ADAPT-VQE energy evaluations \\
\hline
\end{tabular}
\caption{\textbf{Summary of notation used in the main resource bounds.} Additional notation is introduced throughout the paper as needed.}
\label{tab:notation}

\end{table}
\section{\label{adx:qsim_quantum_resource}Quantum Sample complexity of Quantum Simulation: proof of Theorem~\ref{thm:Nqs_quantum}}
\subsection{Single observable}

We begin with a single Pauli observable $P$ and a quantum state $\rho$.
Let
\begin{equation}
    P
    =
    \sum_i a_i\ketbraq{a_i},
\end{equation}
where $a_i\in\{-1,1\}$. From $N$ independent measurements of $P$, we
define
\begin{equation}
    \langle\widehat P\rangle
    =
    \frac{1}{N}
    \sum_{m=1}^{N}\lambda_m,
\end{equation}
where $\lambda_m$ is the $m$-th measurement outcome. This estimator is
unbiased,
\begin{equation}
    \Ebb[\langle\widehat P\rangle]
    =
    \langle P\rangle,
\end{equation}
and, using independence of the samples,
\begin{equation}
    \Var[\langle\widehat P\rangle]
    =
    \frac{\Var[P]}{N}
    =
    \frac{
        \Tr[\rho P^2]
        -
        \Tr[\rho P]^2
    }{N}.
\end{equation}
Chebyshev's inequality therefore gives
\begin{equation}
    \Pr\left[
        \abs{
            \langle\widehat P\rangle
            -
            \langle P\rangle
        }
        >
        \epsilon
    \right]
    \leq
    \frac{\Var[P]}{N\epsilon^2}.
\end{equation}
Hence, it suffices to take
\begin{equation}
    N
    \geq
    \frac{\Var[P]}{\epsilon^2\delta}
\end{equation}
to estimate $\langle P\rangle$ to additive error $\epsilon$ with
success probability at least $1-\delta$. For a Pauli observable,
$\Var[P]\leq1$.

\subsection{\label{sec:sum_observables} Sum of observables}

We now consider an observable
\begin{equation}
    O = \sum_{\alpha=1}^{M} c_\alpha P_\alpha ,
\end{equation}
given as a linear combination of $M$ Pauli observables
$P_\alpha \in \mathcal{P}$ with nonzero coefficients
$c_\alpha \in \mathbb{R}$. We denote by
$\langle \hat{O} \rangle$ the estimator of $\langle O\rangle$,
defined as
\begin{equation}
    \langle \hat{O} \rangle
    =
    \sum_\alpha c_\alpha \langle \hat{P}_\alpha \rangle
    =
    \sum_\alpha
    \frac{c_\alpha}{N_\alpha}
    \sum_{m=1}^{N_\alpha}
    \lambda_m^\alpha ,
    \label{eq:ohat_sample}
\end{equation}
where $\lambda_m^\alpha \in \{-1,1\}$ is the $m$-th
measurement outcome of $P_\alpha$, and $N_\alpha$ is the
number of independent measurements allocated to that
observable. The total number of measurements is
$N=\sum_\alpha N_\alpha$.

Since the measurements are independent, the estimator is
unbiased,
\begin{equation}
    \mathbb{E}\left[\langle\hat{O}\rangle\right]
    =
    \langle O\rangle ,
\end{equation}
and its variance satisfies
\begin{align}
    \Var\left[\langle\hat{O}\rangle\right]
    &=
    \sum_\alpha
    \frac{c_\alpha^2}{N_\alpha}
    \Var[P_\alpha]
    \nonumber\\
    &\leq
    \sum_\alpha
    \frac{c_\alpha^2}{N_\alpha},
    \label{eq:var_O}
\end{align}
where we used $\Var[P_\alpha]\leq 1$ for Pauli observables.

For a fixed total measurement budget $N$, we employ the
weighted measurement strategy
\begin{equation}
    N_\alpha
    =
    N\frac{|c_\alpha|}{\|\vec{c}\|_1},
    \qquad
    \|\vec{c}\|_1 = \sum_\alpha |c_\alpha| ,
    \label{eq:weighted_shots}
\end{equation}
as in Eq.~\eqref{eq:weighted_measurement_main}.
Ignoring the integer rounding of $N_\alpha$, which does not
affect the asymptotic scaling, substitution into
Eq.~\eqref{eq:var_O} gives
\begin{equation}
    \Var\left[\langle\hat{O}\rangle\right]
    \leq
    \frac{\|\vec{c}\|_1^2}{N}.
    \label{eq:weighted_variance}
\end{equation}
In fact, the allocation in Eq.~\eqref{eq:weighted_shots}
minimizes the state-independent upper bound in
Eq.~\eqref{eq:var_O}. Indeed, by the Cauchy--Schwarz
inequality,
\begin{equation}
    \|\vec{c}\|_1^2
    =
    \left(
        \sum_\alpha
        \frac{|c_\alpha|}{\sqrt{N_\alpha}}
        \sqrt{N_\alpha}
    \right)^2
    \leq
    N
    \sum_\alpha\frac{c_\alpha^2}{N_\alpha},
\end{equation}
with equality when $N_\alpha\propto |c_\alpha|$.

We next derive a high-probability bound using Hoeffding's
inequality. Define the independent random variables
\begin{equation}
    X_{\alpha,m}
    =
    \frac{c_\alpha}{N_\alpha}\lambda_m^\alpha ,
\end{equation}
such that
$\langle\hat{O}\rangle=\sum_{\alpha,m}X_{\alpha,m}$.
Under the weighted allocation in
Eq.~\eqref{eq:weighted_shots}, each random variable is
bounded as
\begin{equation}
    X_{\alpha,m}
    \in
    \left[
        -\frac{\|\vec{c}\|_1}{N},
        \frac{\|\vec{c}\|_1}{N}
    \right].
\end{equation}
Since there are $N$ measurement outcomes in total,
Hoeffding's inequality yields
\begin{align}
    \Pr\left[
        \left|
            \langle\hat{O}\rangle-\langle O\rangle
        \right|>\epsilon
    \right]
    &\leq
    2\exp\left[
        -\frac{2\epsilon^2}
        {N\left(2\|\vec{c}\|_1/N\right)^2}
    \right]
    \nonumber\\
    &=
    2\exp\left[
        -\frac{N\epsilon^2}
        {2\|\vec{c}\|_1^2}
    \right].
    \label{eq:hoeffding_weighted}
\end{align}
Therefore, it suffices to take
\begin{equation}
    N
    \geq
    \frac{2\|\vec{c}\|_1^2}{\epsilon^2}
    \log\left(\frac{2}{\delta}\right)
    \label{eq:N_hoeffding_weighted}
\end{equation}
to estimate $\langle O\rangle$ up to additive error
$\epsilon$ with success probability at least $1-\delta$.

\subsection{Collection of evolutions}

We finally extend the result to the collection of $T$
evolutions $\{U^{(i)}\}_{i=1}^{T}$ considered in the main text.
For each evolution, let
\begin{equation}
    f_{U^{(i)}}(\rho,O)
    =
    \Tr[U^{(i)}\rho (U^{(i)})^\dagger O]
\end{equation}
and denote its estimator by $\widetilde{f}_{U^{(i)}}(\rho,O)$.
Applying Eq.~\eqref{eq:N_hoeffding_weighted} with failure
probability $\delta/T$ gives
\begin{equation}
    N^{(i)}
    \geq
    \frac{2\|\vec{c}\|_1^2}{\epsilon^2}
    \log\left(\frac{2T}{\delta}\right)
\end{equation}
measurements for each $U^{(i)}$. A union bound over the $T$
evolutions then guarantees
\begin{equation}
    \Pr\left[
        \bigcap_{i=1}^{T}
        \left\{
        \left|
        f_{U^{(i)}}(\rho,O)
        -
        \widetilde{f}_{U^{(i)}}(\rho,O)
        \right|
        \leq \epsilon
        \right\}
    \right]
    \geq 1-\delta .
\end{equation}
Hence, the total Quantum Sample complexity satisfies
\begin{equation}
    N_{\rm Q.S.}
    \geq
    \frac{2T\|\vec{c}\|_1^2}{\epsilon^2}
    \log\left(\frac{2T}{\delta}\right),
    \label{eq:QS_sample_complexity}
\end{equation}
or, suppressing constant factors,
\begin{equation}
    N_{\rm Q.S.}
    \in
    \mathcal{O}\left(
        \frac{T\|\vec{c}\|_1^2}{\epsilon^2}
        \log\left(\frac{T}{\delta}\right)
    \right).
\end{equation}
This proves Theorem~\ref{thm:Nqs_quantum}.

\subsection{Measurement of commuting observables}

The bound above assumes that the Pauli terms are measured independently.
When subsets of terms commute and admit a common measurement basis, the
same analysis can be applied at the level of commuting groups. Let
\begin{equation}
    O
    =
    \sum_{g=1}^{G}
    \sum_{\beta=1}^{M_g}
    c_{g,\beta}P_{g,\beta},
    \label{eq:O_commute}
\end{equation}
where the operators within each group $g$ are measured simultaneously.
For one measurement of group $g$, define the random variable
\begin{equation}
    Y_g
    =
    \sum_{\beta=1}^{M_g}
    c_{g,\beta}\lambda_{g,\beta},
\end{equation}
where $\lambda_{g,\beta}\in\{-1,1\}$ are the simultaneous measurement
outcomes. Since
\begin{equation}
    \abs{Y_g}
    \leq
    s_g,
    \qquad
    s_g
    :=
    \sum_{\beta=1}^{M_g}\abs{c_{g,\beta}}
    =
    \norm{\vec{c}_g}_1,
\end{equation}
the estimator obtained from $N_g$ measurements of each group satisfies
\begin{equation}
    \Var[\langle\hat O\rangle]
    \leq
    \sum_{g=1}^{G}\frac{s_g^2}{N_g}.
\end{equation}
For a fixed total budget $N=\sum_gN_g$, we again use the weighted
allocation
\begin{equation}
    N_g
    =
    N\frac{s_g}{\sum_{g'}s_{g'}}.
\end{equation}
This gives
\begin{equation}
    \Var[\langle\hat O\rangle]
    \leq
    \frac{\left(\sum_gs_g\right)^2}{N}.
\end{equation}

The same allocation also gives a direct Hoeffding bound. Each
single-shot contribution $Y_g/N_g$ lies in an interval of width
$2(\sum_{g'}s_{g'})/N$, and there are $N$ independent samples in total.
Therefore,
\begin{equation}
    \Pr\left[
        \abs{\langle\hat O\rangle-\langle O\rangle}
        >
        \epsilon
    \right]
    \leq
    2\exp\left[
        -\frac{N\epsilon^2}
        {2\left(\sum_gs_g\right)^2}
    \right].
\end{equation}
It consequently suffices to take
\begin{equation}
    N
    \geq
    \frac{2}{\epsilon^2}
    \left(\sum_{g=1}^{G}\norm{\vec{c}_g}_1\right)^2
    \log\left(\frac{2}{\delta}\right).
    \label{eq:N_commuting_groups}
\end{equation}
Since
$\sum_g\norm{\vec{c}_g}_1=\norm{\vec{c}}_1$, this recovers the same
worst-case coefficient dependence as Eq.~\eqref{eq:N_hoeffding_weighted}.
Commuting measurements can nevertheless reduce the realized number of
distinct circuit settings and can yield tighter bounds when the actual
spectral ranges of the grouped observables are used.

\section{Classical shadows
\label{sec:classical_shadow}}

We briefly review the classical shadow framework used throughout the
appendices~\cite{huang2020predicting,sack2022avoiding,
zhao2021fermionic,bertoni2024shallow}. Let $\VC$ denote an ensemble of
measurement unitaries with distribution $\mu(V)$, and let
$\mathbf P=\{\Pi_x\}_x$ be a projective measurement satisfying
$\sum_x\Pi_x=\eye$. The associated measurement channel is
\begin{equation}
    \MC_{\VC}(\rho)
    =
    \sum_x
    \int_{V\sim\VC}
    \mu(V)
    \Tr[\Pi_xV\rho V^\dagger]
    V^\dagger\Pi_xV.
    \label{eq:generic_shadow_channel}
\end{equation}
We follow the convention of Ref.~\cite{sauvage2024classical}, in which
the measurement projectors are included explicitly in the channel
definition. The corresponding visible operator space is
\begin{equation}
    \LC_{\rm vis}
    =
    \spn\left\{
        V^\dagger\Pi_xV
        \;\middle|\;
        V\in\VC,\ \Pi_x\in\mathbf P
    \right\}.
    \label{eq:visible_space}
\end{equation}
For the standard computational-basis measurement,
\begin{equation}
    \mathbf P_{\rm cb}
    =
    \left\{
        \ketbraq{z}
    \right\}_{z\in\{0,1\}^n}.
    \label{eq:Pi_computational}
\end{equation}

For observables in $\LC_{\rm vis}$, a single measurement outcome
$(V,x)$ defines the snapshot
\begin{equation}
    \widehat\rho_{V,x}
    =
    \MC_{\VC}^{-1}
    \left(
        V^\dagger\Pi_xV
    \right),
\end{equation}
where the inverse is understood on the visible subspace whenever the
measurement channel is not invertible on the full operator space. For
every $O\in\LC_{\rm vis}$,
\begin{equation}
    \Tr[\rho O]
    =
    \Ebb_{V,x}
    \left[
        \Tr[
            \widehat\rho_{V,x}O
        ]
    \right].
\end{equation}
Repeating the procedure $N$ times produces the classical shadow. In
practice, expectation values can be estimated using the empirical
average or a median-of-means estimator.

The variance is controlled by the shadow norm. For an observable $O$,
we use
\begin{equation}
    \norm{O}_{\rm shadow}^2
    :=
    \max_{\sigma:\,{\rm state}}
    \Ebb_{V,x\sim p_\sigma}
    \left[
        \Tr[
            \widehat\rho_{V,x}O
        ]^2
    \right],
    \label{eq:shadow_norm}
\end{equation}
where
$p_\sigma(V,x)=\mu(V)\Tr[\Pi_xV\sigma V^\dagger]$. This gives the
standard simultaneous-estimation bound.

\begin{theorem}[Sample complexity of classical shadows
~\cite{huang2020predicting}]
\label{thm:classical_shadow}
Let $\{O_i\}_{i=1}^{C}\subseteq\LC_{\rm vis}$ be a collection of
observables. To estimate all $\Tr[\rho O_i]$ to additive error
$\epsilon$, with failure probability at most $\delta$, it suffices to
take
\begin{equation}
    N_{\QS}
    \in
    \order{
        \frac{\log(C/\delta)}{\epsilon^2}
        \max_{i=1,\ldots,C}
        \norm{O_i}_{\rm shadow}^2
    }.
    \label{eq:shadow_complexity}
\end{equation}
\end{theorem}

Two standard choices are global Clifford measurements and randomized
single-qubit Pauli measurements. For global Clifford measurements, the
inverse channel is
\begin{equation}
    \MC_{\VC}^{-1}[X]
    =
    (2^n+1)X-\Tr[X]\eye.
\end{equation}
For local Pauli measurements, the channel factorizes over the qubits and
only single-qubit basis rotations are required before computational-basis
readout.

For an arbitrary observable supported on at most $k$ qubits, local
Pauli measurements give the sufficient bound
\begin{equation}
    \norm{O}_{\rm shadow}^2
    \leq
    4^k\norm{O}_\infty^2.
\end{equation}
Hence,

\begin{corollary}[Local Pauli classical shadows
~\cite{sack2022avoiding}]
\label{thm:classical_shadow_pauli}
Let $\{O_i\}_{i=1}^{C}$ be a collection of observables supported on at
most $k$ qubits. Randomized single-qubit Pauli measurements estimate all
$C$ expectation values to additive error $\epsilon$, with failure
probability at most $\delta$, using
\begin{equation}
    N_{\QS}
    \in
    \order{
        \frac{4^k}{\epsilon^2}
        \log\left(\frac{C}{\delta}\right)
        \max_i\norm{O_i}_\infty^2
    }.
    \label{eq:shadow_complexity_pauli}
\end{equation}
\end{corollary}

For local Pauli measurements, this is a worst-case sufficient bound.
More specialized measurement primitives can substantially improve the
scaling when additional structure is available. Below, we use
Matchgate shadows, permutation-invariant shadows, and
$\U(1)$-symmetric shadows for the corresponding circuit families.

We stress that Eqs.~\eqref{eq:shadow_complexity}
and~\eqref{eq:shadow_complexity_pauli} count only quantum samples. The
classical processing of the snapshots depends on the measurement
primitive and on how the data are reused, and is included separately in
the Classical Time complexity of each simulation method.

\section{Pauli Propagation \label{adx:pauli_propagation}}

\begin{figure}[h]
    \centering
    \includegraphics[width=0.47\linewidth]{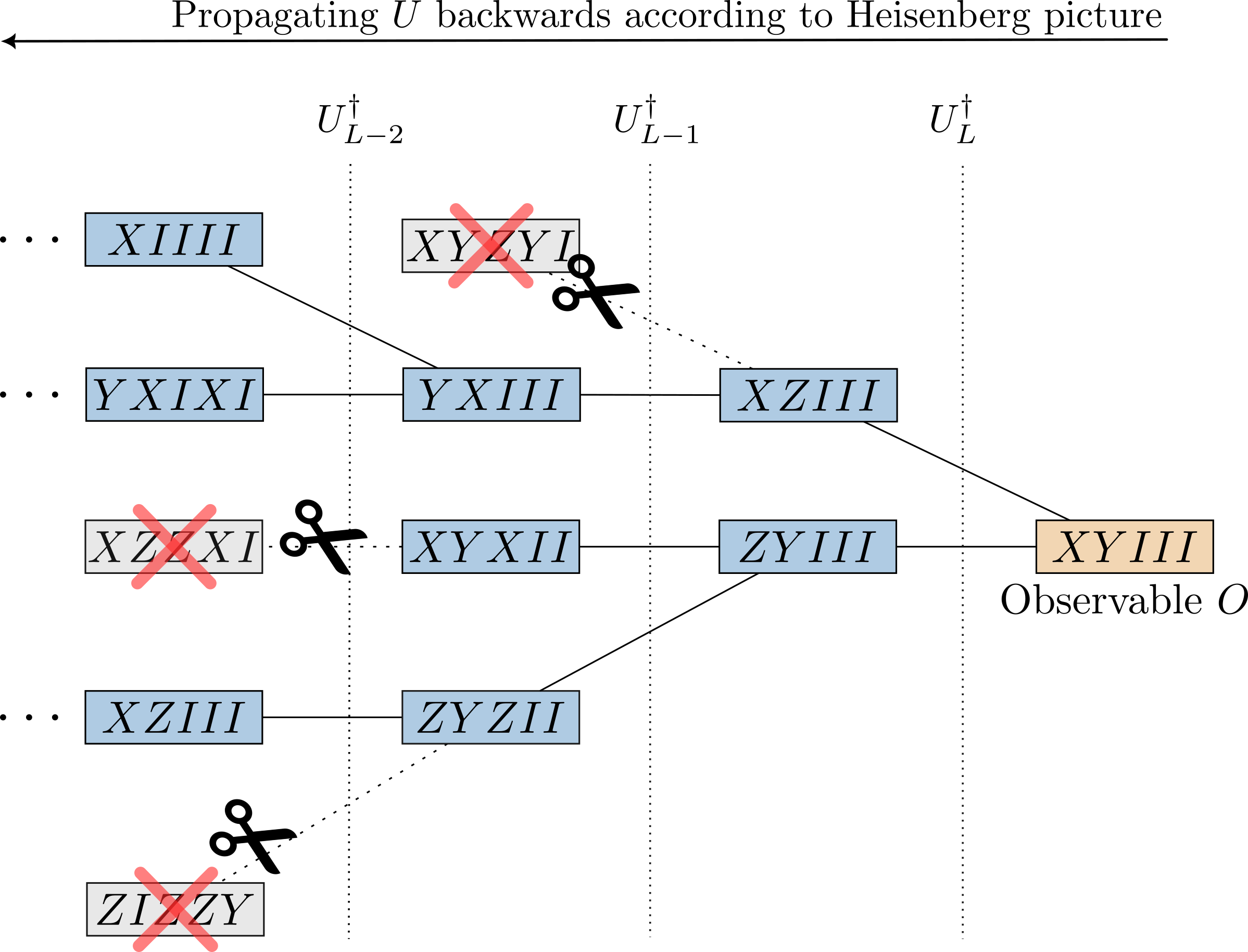}
    \caption{\textbf{Schematic illustration of Pauli Propagation with
    weight truncation.} The observable is evolved in the Heisenberg
    picture and propagation paths are discarded whenever the generated
    Pauli string exceeds the weight threshold $\tau$ (here,
    $\tau=4$).}
    \label{fig:pauli_propagation}
\end{figure}

In this section, we provide additional details on the low-weight Pauli
Propagation framework used in the main text
~\cite{angrisani2024classically,aharonov2022polynomial,
angrisani2025simulating,shao2023simulating}. To avoid normalization
factors, let
\begin{equation}
    \overline{\PC}
    =
    \left\{
        \frac{P}{\sqrt{2^n}}
        \;\middle|\;
        P\in\{\eye,X,Y,Z\}^{\otimes n}
    \right\}
\end{equation}
denote the Hilbert--Schmidt-orthonormal Pauli basis. For arbitrary
operators $H$ and $H'$,
\begin{equation}
    \Tr[HH']
    =
    \sum_{s\in\overline{\PC}}
    \Tr[Hs]\Tr[sH'].
\end{equation}
Inserting this identity between the circuit layers gives the Pauli-path
representation
\begin{equation}
    f_U(\rho,O)
    =
    \sum_{s_0,\ldots,s_L\in\overline{\PC}}
    \Tr[Os_L]
    \prod_{\ell=1}^{L}
    \Tr[U_\ell^\dagger s_\ell U_\ell s_{\ell-1}]
    \Tr[s_0\rho].
    \label{eq:pauli_path_appendix}
\end{equation}
A path is specified by
$\gamma=(s_0,\ldots,s_L)$, while the factors
$\Tr[U_\ell^\dagger s_\ell U_\ell s_{\ell-1}]$ are the corresponding
transition amplitudes.

Exact evaluation of Eq.~\eqref{eq:pauli_path_appendix} can involve
exponentially many paths. Low-weight Pauli Propagation restricts the
sum to paths for which the Pauli weight never exceeds a threshold
$\tau$. We denote the resulting estimate by
$f_U^{(\tau)}(\rho,O)$. Under the locally scrambling assumptions of
Ref.~\cite{angrisani2024classically}, this truncation admits the
average-case guarantee
\begin{equation}
    \abs{
        f_U(\rho,O)-f_U^{(\tau)}(\rho,O)
    }
    \leq
    \epsilon\norm{O}_{\rm Pauli,2}
\end{equation}
with probability at least $1-\delta$ for
\begin{equation}
    \tau
    \in
    \order{
        \log(\epsilon^{-1}\delta^{-1})
    }.
    \label{eq:error_pauli_propagation}
\end{equation}
Here, the circuit layers are sampled from a locally scrambling
distribution, as specified in Ref.~\cite{angrisani2024classically}.

The worst-case runtime follows by counting the retained transition
amplitudes. Let $\widetilde O_L=O$ and recursively define
\begin{equation}
    \widetilde O_{\ell-1}
    =
    \sum_{\substack{s\in\overline{\PC}\\ |s|\leq\tau}}
    \Tr[
        U_\ell^\dagger
        \widetilde O_\ell
        U_\ell s
    ]s,
    \qquad
    \ell=L,\ldots,1.
\end{equation}
If every $\widetilde O_\ell$ is supported on at most
$\widetilde n\leq n$ qubits, Ref.~\cite{angrisani2024classically}
gives
\begin{equation}
    N_{\CT}
    \in
    \order{
        L\min\left\{
            \widetilde n^{2\tau},
            \widetilde n^\tau\poly(n)
        \right\}
    }.
    \label{eq:NCT_pauli_propagation}
\end{equation}
The first term counts all possible transitions between retained
low-weight Pauli strings, while the second applies when the number of
nonzero Pauli strings in the propagated observable is itself
polynomially bounded. In the generic case
$\widetilde n\leq n$, this gives a polynomial runtime whenever
$L\in\poly(n)$ and $\tau\in\order{1}$. The architecture-specific bound
can be substantially tighter, as illustrated for QCNNs in
Appendix~\ref{adx:qcnn}.

For circuits composed of Pauli rotations and Clifford gates, the
branching structure can be written explicitly. Clifford gates map one
Pauli string to another, while a Pauli rotation satisfies
\begin{equation}
 e^{i\frac{w}{2}P_\alpha}
 P_\beta
 e^{-i\frac{w}{2}P_\alpha}
 =
 \begin{cases}
    P_\beta,
    & [P_\alpha,P_\beta]=0,\\[1mm]
    \cos(w)P_\beta
    +
    \dfrac{i}{2}\sin(w)[P_\alpha,P_\beta],
    & [P_\alpha,P_\beta]\neq0.
 \end{cases}
 \label{eq:pauli_relation}
\end{equation}
When the two Pauli strings anticommute,
$\frac{i}{2}[P_\alpha,P_\beta]$ is, up to a sign, another Hermitian
Pauli string. Thus, each noncommuting rotation generates at most two
branches. Efficient binary representations of Pauli operators can
further reduce the practical cost of this propagation
~\cite{rudolph2025pauli}.

\section{\label{adx:gsim}Generic $\liea$-sim framework}

In this section, we provide additional details on the $\liea$-sim
framework introduced in Section~\ref{sec:gsim}. The presentation follows
Ref.~\cite{goh2023lie}, but we formulate the simulation directly on the
invariant operator subspace supporting the observable.

\subsection{Lie theory}

\begin{definition}[Dynamical Lie Algebra]
Consider a quantum circuit of the form given in
Eq.~\eqref{eq:generic_circuit}, with Hermitian generators
$\GC=\{H_1,\ldots,H_L\}$. The Dynamical Lie Algebra (DLA) is
\begin{equation}
    \liea
    =
    \langle i\GC\rangle_{\rm Lie}
    =
    \spn_{\R}\langle
        iH_1,\ldots,iH_L
    \rangle_{\rm Lie}
    \subseteq
    \mathfrak{su}(2^n).
\end{equation}
We denote by $\{iG_\gamma\}_{\gamma=1}^{\dim(\liea)}$ a basis of
$\liea$.
\label{def:dla}
\end{definition}

The associated dynamical Lie group is
$\lieg=\exp(\liea)$. We consider its adjoint action on the operator space
$\BC$ and decompose the latter into invariant representation blocks,
\begin{equation}
    \BC
    \simeq
    \bigoplus_\lambda \BC_\lambda,
    \qquad
    U^\dagger A U\in\BC_\lambda
\end{equation}
for every $A\in\BC_\lambda$ and $U\in\lieg$. When these blocks are irreducible they correspond to irreps, although
the simulation only requires invariance. Let
$d_\lambda=\dim(\BC_\lambda)$ and let
$\{B_\alpha^{(\lambda)}\}_{\alpha=1}^{d_\lambda}$ be a
Hilbert--Schmidt-orthonormal Hermitian basis of $\BC_\lambda$.

For $iG_\gamma\in\liea$, its action on $\BC_\lambda$ is defined by
\begin{equation}
    [iG_\gamma,iB_\alpha^{(\lambda)}]
    =
    \sum_{\beta=1}^{d_\lambda}
    f_{\alpha\beta}^{(\lambda,\gamma)}
    iB_\beta^{(\lambda)}.
\end{equation}
Equivalently,
\begin{equation}
    \left(
        \Phi_\lambda^{\rm ad}(iG_\gamma)
    \right)_{\alpha\beta}
    =
    f_{\alpha\beta}^{(\lambda,\gamma)}.
\end{equation}
Hence, each element of the DLA is represented on $\BC_\lambda$ by a
$d_\lambda\times d_\lambda$ matrix. Exponentiating this representation
gives the action of the dynamical Lie group. In particular,
\begin{equation}
    U^\dagger B_\alpha^{(\lambda)}U
    =
    \sum_{\beta=1}^{d_\lambda}
    \left(
        \Phi_\lambda^{\rm Ad}(U)
    \right)_{\alpha\beta}
    B_\beta^{(\lambda)}.
    \label{eq:adjoint_action}
\end{equation}
For a gate $U_\ell=e^{-iw_\ell H_\ell}$, we define
$i\overline{H}_\ell^{(\lambda)}
=\Phi_\lambda^{\rm ad}(iH_\ell)$, such that
\begin{equation}
    \Phi_\lambda^{\rm Ad}(U_\ell)
    =
    e^{-iw_\ell\overline{H}_\ell^{(\lambda)}}.
\end{equation}
The relevant matrix dimension is therefore $d_\lambda$, which need not
coincide with $\dim(\liea)$.

\subsection{\label{sec:gsim_principle}$\liea$-sim principles}

Consider first an observable supported on a single invariant subspace,
\begin{equation}
    O
    =
    \sum_{\alpha=1}^{d_\lambda}
    c_\alpha B_\alpha^{(\lambda)}.
\end{equation}
Its Heisenberg evolution remains in the same subspace,
\begin{equation}
    \widetilde{O}
    =
    U^\dagger O U
    =
    \sum_{\alpha=1}^{d_\lambda}
    \widetilde{c}_\alpha B_\alpha^{(\lambda)},
\end{equation}
with
\begin{equation}
    \widetilde{\vec{c}}
    =
    \left(
        \Phi_\lambda^{\rm Ad}(U)
    \right)^T
    \vec{c}.
\end{equation}
The expectation value can therefore be written as
\begin{equation}
    f_U(\rho,O)
    =
    \Tr[\rho U^\dagger O U]
    =
    \sum_{\alpha=1}^{d_\lambda}
    \widetilde{c}_\alpha
    \Tr[\rho B_\alpha^{(\lambda)}].
    \label{eq:loss_gsim}
\end{equation}
We define the corresponding input-state vector as
\begin{equation}
    \left(
        \vec{e}^{\,({\rm in})}_\lambda
    \right)_\alpha
    =
    \Tr[\rho B_\alpha^{(\lambda)}],
    \label{eq:ein}
\end{equation}
so that
\begin{equation}
    f_U(\rho,O)
    =
    \widetilde{\vec{c}}^{\,T}
    \vec{e}^{\,({\rm in})}_\lambda.
\end{equation}

The cost of this procedure is controlled by $d_\lambda$. Let $K$ denote
the number of distinct circuit generators. In a generic dense
implementation, diagonalizing the represented generators requires
\begin{equation}
    N_{\rm pre}
    \in
    \order{Kd_\lambda^3}.
\end{equation}
These decompositions depend only on the generators and can be reused as
the circuit parameters change. Applying the represented gates to the
coefficient vector then requires
\begin{equation}
    N_{\rm eval}
    \in
    \order{Ld_\lambda^2}
\end{equation}
for one circuit instance. Hence, over $T$ instances,
\begin{equation}
    N_{\CT}
    \in
    \order{
        Kd_\lambda^3
        +
        TLd_\lambda^2
    }.
\end{equation}
For observables supported on a constant number of invariant subspaces,
the same construction is applied independently within each subspace and
the final expectation values are summed.

\section{Classical Simulation of Clifford circuits
\label{adx:clifford}}

In this section, we prove
Theorem~\ref{thm:classical_simulation_clifford}. An $n$-qubit Pauli
operator can be represented, up to a phase, by a binary vector
\begin{equation}
    \vec{p}
    =
    (\vec{z},\vec{x})
    \in
    \mathbb{F}_2^{2n},
\end{equation}
according to
\begin{equation}
    I\rightarrow 00,\qquad
    X\rightarrow 01,\qquad
    Z\rightarrow 10,\qquad
    Y\rightarrow 11.
\end{equation}
The adjoint action of a Clifford unitary $U_\ell$ is characterized by
a symplectic matrix $S_\ell\in{\rm Sp}(2n,\mathbb{F}_2)$ together with
the corresponding phase data,
\begin{equation}
    U_\ell^\dagger P(\vec{p})U_\ell
    =
    s_\ell(\vec{p})P(S_\ell\vec{p}),
    \qquad
    s_\ell(\vec{p})\in\{-1,+1\}.
    \label{eq:clifford_symplectic_action}
\end{equation}
The phase is restricted to $\pm1$ since conjugation preserves
Hermiticity.

Consider an observable
\begin{equation}
    O=\sum_{\alpha=1}^{M}c_\alpha P_\alpha.
\end{equation}
Backpropagating each Pauli operator through the $L$ Clifford layers gives
\begin{equation}
    U^\dagger P_\alpha U
    =
    s_\alpha\widetilde{P}_\alpha,
    \qquad
    s_\alpha\in\{-1,+1\},
\end{equation}
where $\widetilde{P}_\alpha$ is again a single Pauli string. Hence,
\begin{equation}
    U^\dagger O U
    =
    \sum_{\alpha=1}^{M}
    \widetilde{c}_\alpha\widetilde{P}_\alpha,
    \qquad
    \widetilde{c}_\alpha=s_\alpha c_\alpha.
    \label{eq:clifford_backprop}
\end{equation}
In particular, Clifford evolution does not generate additional terms
and
\begin{equation}
    \|\widetilde{\vec{c}}\|_1
    =
    \|\vec{c}\|_1.
    \label{eq:clifford_l1}
\end{equation}

We first consider the Quantum Sample complexity. Using
Eq.~\eqref{eq:clifford_backprop},
\begin{equation}
    f_U(\rho,O)
    =
    \sum_{\alpha=1}^{M}
    \widetilde{c}_\alpha
    \Tr[\rho\widetilde{P}_\alpha].
    \label{eq:clifford_expectation}
\end{equation}
Thus, the CS only requires estimating Pauli
expectation values directly on the initial state. Since Clifford
backpropagation preserves the $\ell_1$-norm of the coefficient vector,
the weighted measurement strategy of Eq.~\eqref{eq:weighted_shots} and
Theorem~\ref{thm:Nqs_quantum} give
\begin{equation}
    N_{\QS}
    \in
    \order{
        \frac{T\|\vec{c}\|_1^2}{\epsilon^2}
        \log\left(\frac{T}{\delta}\right)
    }.
    \label{eq:NQS_clifford_appendix}
\end{equation}

We next consider the Classical Time complexity. For a generic dense
symplectic representation, multiplying a $2n\times2n$ binary matrix by
a $2n$-dimensional binary vector requires $\order{n^2}$ elementary
operations. Propagating the $M$ Pauli strings through all $L$ layers
therefore requires $\order{LMn^2}$ operations for one circuit instance.
Since the $T$ instances can correspond to different Clifford
transformations, the total backpropagation cost is
\begin{equation}
    N_{\rm prop}
    \in
    \order{TLMn^2}.
\end{equation}
The measurement outcomes can be accumulated using a streaming estimator
with constant classical cost per sample, giving an additional
$\order{N_{\QS}}$ operations. Hence,
\begin{equation}
    N_{\CT}
    \in
    \order{
        N_{\QS}
        +
        TLMn^2
    }.
    \label{eq:NCT_clifford_appendix}
\end{equation}

Finally, each $\widetilde{P}_\alpha$ can be measured by applying
single-qubit basis rotations in parallel followed by computational-basis
measurement. No implementation of the original Clifford evolution is
required on the quantum device, and therefore
\begin{equation}
    N_{\QT}\in\order{1}.
\end{equation}
Together with $N_{\QT}\in\order{1}$, Eqs.~\eqref{eq:NQS_clifford_appendix}
and~\eqref{eq:NCT_clifford_appendix} prove
Theorem~\ref{thm:classical_simulation_clifford}.

\section{Shallow Hardware Efficient Ansatz\label{sec:shallow_hea}}

Shallow HEAs are classically simulable in the setting considered here
because the backpropagated observables remain confined to
polynomially-sized light cones. We consider a circuit with depth
$L\leq\log_2(n)$ and a $k$-local observable
\begin{equation}
    O=\sum_{\alpha=1}^{M}c_\alpha P_\alpha,
\end{equation}
where $k\in\order{1}$. The expectation value can be written as
\begin{equation}
    f_U(\rho,O)
    =
    \sum_\alpha
    c_\alpha
    \Tr[\rho U^\dagger P_\alpha U]
    =
    \sum_\alpha
    c_\alpha
    \Tr[\rho\widetilde{P}_\alpha],
    \label{eq:loss_shallow_hea}
\end{equation}
with $\widetilde{P}_\alpha=U^\dagger P_\alpha U$.

Each $\widetilde{P}_\alpha$ is supported on the backward light cone of
$P_\alpha$. For the one-dimensional alternating-layer architecture,
\begin{equation}
    \widetilde{k}
    \leq
    2L+2k
    \leq
    2\log_2(n)+2k
    \label{eq:shallow_lightcone_size}
\end{equation}
qubits are sufficient. We denote by $\rho_\alpha$ the reduced state of
$\rho$ on this light cone. If
$\{P_j^\alpha\}_j$ is an orthogonal Pauli basis of the corresponding
operator space $\BC_\alpha$, then
\begin{equation}
    \rho_\alpha
    =
    \frac{1}{2^{\widetilde{k}}}
    \sum_{P_j^\alpha\in\BC_\alpha}
    \Tr[\rho_\alpha P_j^\alpha]
    P_j^\alpha.
    \label{eq:rho_lm}
\end{equation}
The required reduced-state information is obtained using local Pauli
classical shadows~\cite{huang2020predicting}.

\subsection{Quantum Sample Complexity}

We consider $T$ circuit instances $\{U^{(i)}\}_{i=1}^{T}$ with the same
architecture and generators but potentially different parameters.
There are at most $MT$ backpropagated Pauli terms whose expectation
values must be controlled simultaneously.

\begin{theorem}[Quantum Sample Complexity for Classical Simulation of shallow HEAs]
\label{thm:shallow_hea_sample_appendix}
Consider $T$ circuit instances of a shallow HEA on $n$ qubits with
depth $L\leq\log_2(n)$. Given
$O=\sum_\alpha c_\alpha P_\alpha$, consisting of $M$ $k$-geometrically local Pauli
operators with $k\in\order{1}$, the Quantum Sample complexity satisfies
\begin{equation}
    N_{\QS}
    \in
    \order{
        \frac{16^k n^4}{\epsilon^2}
        \log\left(\frac{MT}{\delta}\right)
        \norm{\vec{c}}_1^2
    }.
    \label{eq:classical_shadow_shallow_hea_appendix}
\end{equation}
\end{theorem}

\begin{proof}
For one circuit instance, write
\begin{equation}
    \widetilde{O}
    =
    \sum_\alpha c_\alpha\widetilde{P}_\alpha.
\end{equation}
If every $\widetilde{P}_\alpha$ is estimated up to additive error
$\widetilde{\epsilon}$, then
\begin{align}
    \abs{
        \Tr[\rho\widetilde{O}]
        -
        \Tr[\widehat{\rho}\widetilde{O}]
    }
    &\leq
    \sum_\alpha
    \abs{c_\alpha}
    \widetilde{\epsilon}
    \nonumber\\
    &=
    \norm{\vec{c}}_1\widetilde{\epsilon}.
\end{align}
Thus, it suffices to choose
\begin{equation}
    \widetilde{\epsilon}
    =
    \frac{\epsilon}{\norm{\vec{c}}_1}.
\end{equation}
Moreover,
\begin{equation}
    4^{\widetilde{k}}
    \leq
    4^{2\log_2(n)+2k}
    =
    16^k n^4,
\end{equation}
and
$\norm{\widetilde{P}_\alpha}_\infty=1$ by unitary invariance of the
operator norm. Applying Corollary~\ref{thm:classical_shadow_pauli} to
the collection of at most $MT$ backpropagated observables gives
\begin{align}
    N_{\QS}
    &\in
    \order{
        \frac{4^{\widetilde{k}}}{\widetilde{\epsilon}^2}
        \log\left(\frac{MT}{\delta}\right)
    }
    \nonumber\\
    &\subseteq
    \order{
        \frac{16^k n^4}{\epsilon^2}
        \log\left(\frac{MT}{\delta}\right)
        \norm{\vec{c}}_1^2
    }.
\end{align}
\end{proof}

The $n^4$ scaling in
Eq.~\eqref{eq:classical_shadow_shallow_hea_appendix} is a worst-case
upper bound arising from the largest allowed light cone. In practice,
the number of measurements can be substantially smaller, as observed
numerically in Ref.~\cite{basheer2023alternating}.

\subsection{Classical Time Complexity}

We first consider the parameter-dependent reduced-unitary calculations.
After backpropagating through $\ell$ layers, the reduced Hilbert-space
dimension is bounded by
\begin{equation}
    d_\ell
    \leq
    2^{2\ell+2k}.
\end{equation}
Using dense matrix multiplication, the cost at layer $\ell$ is bounded
by $\order{d_\ell^\omega}$. Hence,
\begin{align}
    \sum_{\ell=1}^{L}
    \order{d_\ell^\omega}
    &\subseteq
    \order{
        2^{2k\omega}
        \sum_{\ell=1}^{L}
        2^{2\omega\ell}
    }
    \nonumber\\
    &=
    \order{
        2^{2k\omega}2^{2\omega L}
    }
    \subseteq
    \order{n^{2\omega}},
    \label{eq:shallow_hea_matrix_cost}
\end{align}
where the last relation uses $k\in\order{1}$ and
$L\leq\log_2(n)$. Thus, the reduced-unitary calculation requires
$\order{MTn^{2\omega}}$ operations over all $M$ terms and $T$ circuit
instances.

We next account for the classical processing of the shadow data. The
backward light-cone supports depend only on the circuit architecture and
the measured Pauli terms, and are therefore the same for all $T$
parameter choices. The $N_{\QS}$ snapshots can consequently be
accumulated once into empirical reduced shadows on these $M$ light
cones. A dense operator on a $\widetilde{k}$-qubit light cone contains
$4^{\widetilde{k}}$ entries, giving the conservative bound
\begin{equation}
    N_{\rm shadow\mbox{-}post}
    \in
    \order{
        M N_{\QS}4^{\widetilde{k}}
    }
    \subseteq
    \order{
        M N_{\QS}16^k n^4
    }.
    \label{eq:shallow_shadow_postprocessing}
\end{equation}
This cost is incurred only once. After these reduced descriptions have
been constructed, contracting them with the backpropagated observables
costs at most $\order{MT4^{\widetilde{k}}}$, which is absorbed into
the reduced-unitary bound above.

Combining the two contributions gives the following result.

\begin{theorem}[Classical Time Complexity for Classical Simulation of shallow HEAs]
\label{thm:classical_time_shallow}
Under the assumptions of
Theorem~\ref{thm:shallow_hea_sample_appendix}, the total Classical Time
complexity satisfies
\begin{equation}
    N_{\CT}
    \in
    \order{
        M N_{\QS}16^k n^4
        +
        MTn^{2\omega}
    }.
    \label{eq:NCT_shallow_hea_appendix}
\end{equation}
If the reduced input-state descriptions are known classically, the
first term is absent.
\end{theorem}

\begin{proof}
Equation~\eqref{eq:shallow_hea_matrix_cost} gives
$\order{MTn^{2\omega}}$ operations for the parameter-dependent circuit
evaluations, while Eq.~\eqref{eq:shallow_shadow_postprocessing} gives
the one-time cost of processing the shadow data. Adding the two terms
proves Eq.~\eqref{eq:NCT_shallow_hea_appendix}.
\end{proof}

This is a worst-case estimate based on explicit dense matrix
construction and multiplication and therefore applies to arbitrary
two-qubit circuit generators. The light-cone geometry can be precomputed
and reused, whereas the parameter-dependent reduced unitaries must in
general be reevaluated as the circuit parameters change.

For circuits composed of Pauli rotations and Clifford gates, one can
instead consider direct Pauli propagation. Without truncation, however,
the number of Pauli paths can grow exponentially with the number of
gates contained in the backward light cone. Since this number can scale
as $\order{\log^2(n)}$ for a shallow HEA, direct untruncated Pauli
propagation can scale as
\begin{equation}
    \order{2^{\log^2(n)}},
\end{equation}
which is superpolynomial. For the generic shallow HEA setting considered
here, reduced matrix multiplication provides the polynomial worst-case
guarantee used in our resource estimates.

Finally, the local Pauli shadow measurements require only parallel
single-qubit basis rotations before readout. Hence,
\begin{equation}
    N_{\QT}\in\order{1}.
\end{equation}
Together with Eqs.~\eqref{eq:classical_shadow_shallow_hea_appendix}
and~\eqref{eq:NCT_shallow_hea_appendix}, this proves
Theorem~\ref{thm:shallow_HEA}.

\section{\label{adx:matchgate}Matchgate circuit}
In this appendix, we provide additional details on the Classical
Simulation of Matchgate circuits and prove the resource bounds stated
in Section~\ref{sec:matchgate}.
\subsection{\label{adx:Majoran_gsim}Majorana formalism}
We begin by reviewing the Majorana formalism that underpins our simulation framework. Consider a parameterized Matchgate circuit generated by the set $\GC=\{Z_i\}_{i=1}^n \cup \{X_iX_{i+1}\}_{i=1}^{n-1} $. The associated DLA is
\begin{equation}
    \liea = \spn_{\R}i\left(
        \{Z_i\}_{i=1}^n
        \cup
        \left\{
            \widehat{X_iX_j},
            \widehat{X_iY_j},
            \widehat{Y_iX_j},
            \widehat{Y_iY_j}
        \right\}_{1\leq i<j\leq n}
    \right)\;,
    \label{eq:matchgate_dla}
\end{equation}
where $\widehat{A_iB_j} = A_i Z_{i+1}\cdots Z_{j-1}B_j$\;. Using the Majorana operators defined in Eq.~\eqref{eq:majorana}, we
consider the group of fermionic Gaussian unitaries $\FGU(n)$. A
quadratic fermionic Hamiltonian generates an orthogonal transformation
$Q\in\mathrm{SO}(2n)$ on the Majorana operators. With the convention
used in Section~\ref{sec:classically_simulable_evolutions}, we write
the Heisenberg action as
\begin{equation}
    U_Q\ad\gamma_{\mu}U_Q
    =
    \sum_{\nu=1}^{2n}
    Q_{\mu\nu}\gamma_{\nu}.
    \label{eq:UcU}
\end{equation}
Thus, a fermionic Gaussian unitary is characterized, up to a global
phase, by a $2n\times2n$ real orthogonal matrix $Q$, which we define
through Eq.~\eqref{eq:UcU}.

Having established the linear action on single Majoranas, we can now generalize to products of $\kappa$ Majorana operators $\gamma_{\vec{\mu}}$. Using Eq.~\eqref{eq:UcU}, we obtain
\begin{equation}
    U_Q^\dagger\gamma_{\vec{\mu}}U_Q
    =
    \prod_{j=1}^{\kappa}
    \left(
        U_Q^\dagger\gamma_{\mu_j}U_Q
    \right)
    =
    \sum_{\alpha_1,\ldots,\alpha_\kappa=1}^{2n}
    Q_{\mu_1\alpha_1}\cdots Q_{\mu_\kappa\alpha_\kappa}
    \gamma_{\alpha_1}\cdots\gamma_{\alpha_\kappa}\;.
    \label{eq:UckU}
\end{equation}
The anticommutation relations of the Majorana operators imply $\gamma_{\al_1} \cdots \gamma_{\al_\kappa} = \pi(\sg)\gamma_{\sg(\al_1)}\cdots \gamma_{\sg(\al_\kappa)}$ where $\sg$ is a permutation of the indices $\al_1,\dots\al_\kappa$ and $\pi(\sg)$ its sign. Reordering indices into strictly increasing order allows us to rewrite Eq.~\eqref{eq:UckU} in determinant form, summing over $\binom{2n}{\kappa}$ terms~\cite{zhao2021fermionic}: 
\begin{equation}
    U_Q\ad\gamma_{\vec{\mu}}U_Q = \sum_{\al_1 = 1}^{2n} \sum_{\al_2 = \al_1 + 1}^{2n} \cdots \sum_{\al_\kappa = \al_{\kappa-1} + 1}^{2n} \det(Q_{\vec{\mu}\vec{\al}}) \gamma_{\al_1}\cdots \gamma_{\al_{\kappa}} = \sum_{\vec{\al} \in \CC_{2n,\kappa }} \det(Q_{\vec{\mu}\vec{\al}}) \gamma_{\vec{\al}}\;,
    \label{eq:UckU_det}
\end{equation}
where $\vec{\alpha}=(\alpha_1,\ldots,\alpha_\kappa)$ and $\CC_{2n,\kappa}$ is the admissible index set defined in Eq.~\eqref{eq:CC_def}. Here, 
$Q_{\vec{\mu}\vec{\al}}$ denotes the $\kappa\times\kappa$ submatrix of $Q$ with rows indexed by $\vec{\mu}$ and columns by $\vec{\alpha}$. Therefore, the adjoint action of a fermionic Gaussian unitary on $\kappa$-local Majorana strings can be completely determined by the orthogonal matrix $Q$.

This naturally motivates ordering the Majorana strings by fermionic
locality. For each $\kappa$, the Hermitian strings
$\{G^\kappa_{\vec{\nu}}\}_{\vec{\nu}\in\CC_{2n,\kappa}}$
form a basis of the subspace $\BC_\kappa$,
\begin{equation}
    \BC_\kappa = \text{span}_{\Cbb}\{G^\kappa_{\vec{\nu}}\}_{\vec{\nu} \in \CC_{2n, \kappa}}   = \text{span}_{\Cbb}\{i^{\kappa(\kappa-1)/2} \gamma_{\nu_1}\cdots \gamma_{\nu_\kappa}\}_{1\le \nu_1 < \cdots < \nu_\kappa \le 2n}\;.
\end{equation}
where $\dim(\BC_\kappa) = \binom{2n}{\kappa}$ and the prefactor $i^{\kappa(\kappa-1)/2}$ ensures that the basis operators are Hermitian. 
In this representation, the DLA is
$\liea=\spn_{\R}\{\gamma_{\nu_1}\gamma_{\nu_2}\}_{\nu_1<\nu_2}$,
while its Hermitian counterpart is
$i\liea=\spn_{\R}\{G_{\vec{\nu}}^2\}_{\vec{\nu}\in\CC_{2n,2}}
\subset\BC_2$. In particular, the quadratic Hamiltonian $H_Q$ lies in
this Hermitian subspace.   
From the definition of the spaces $\BC_\kappa$, we obtain the following decomposition of the full operator algebra: 
\begin{lemma}[Majorana-sector decomposition~\cite{diaz2023showcasing}]
The operator space decomposes as
\begin{equation}
    \BC
    =
    \bigoplus_{\kappa=0}^{2n}
    \BC_\kappa,
\end{equation}
where
$\dim(\BC_\kappa)=\binom{2n}{\kappa}$.
\end{lemma}
Equation~\eqref{eq:UckU_det} also shows that every $\BC_\kappa$ is
invariant under Matchgate evolution.
\begin{lemma}[Invariant Majorana subspaces]
For every $O\in\BC_\kappa$ and every Matchgate unitary
$U_Q\in\lieg$,
\begin{equation}
    U_Q^\dagger O U_Q
    \in
    \BC_\kappa.
\end{equation}
\end{lemma}

This structure allows the $\liea$-sim method to be naturally adapted to Matchgate circuits by choosing the Majorana operator basis $\{G^\kappa_{\vec{\nu}}\}$ as the simulation basis, with the observable $O = \sum_{\vec{\al}} c_{\vec{\al}}G^\kappa_{\vec{\al}}\in \BC_\kappa$.

\subsection{\label{sec:fermionic_shadow}Quantum Sample Complexity
(Fermionic Shadow)}

In Majorana $\liea$-sim, the observable is backpropagated classically
within the invariant Majorana subspace $\BC_\kappa$, and its expectation
value with respect to the initial state is estimated using fermionic
classical shadows. Since Majorana strings are generally nonlocal when
expressed in the Pauli basis, local Pauli classical shadows lead to
unfavorable sample complexities. We instead employ Fermionic shadows,
also known as Matchgate shadows, which perform randomized measurements
using fermionic Gaussian unitaries~\cite{zhao2021fermionic,
wan2022matchgate}.

Although the Majorana formalism established above encompasses arbitrary
fermionic locality $\kappa$, physical fermionic observables have even
Majorana degree due to fermionic parity superselection. We therefore
restrict the discussion below to $\kappa=2r$, with
$r\in\{0,\ldots,n\}$.

It is natural to choose fermionic Gaussian unitaries as the measurement
primitives since they preserve fermionic locality. However, because
$\FGU(n)$ is a continuous group, direct sampling from this ensemble can
be inconvenient in practice. An equivalent implementation can be
obtained using the discrete subgroup of fermionic Gaussian Clifford
unitaries whose associated orthogonal matrices are signed permutations.
We therefore take
\begin{equation}
    \VC_{\rm F}
    =
    \FGU(n)\cap\Cl(n)
    =
    \left\{
        V_Q
        \;\middle|\;
        Q\in B(2n)\cap\mathrm{SO}(2n)
    \right\},
    \label{eq:U_fermionic}
\end{equation}
where $B(2n)$ denotes the signed permutation group. The continuous and
discrete ensembles generate the same shadow channel and agree in the
moments required for the second-moment bounds considered below
~\cite{wan2022matchgate}. Moreover, every
$V_Q\in\VC_{\rm F}$ admits an efficient decomposition into Givens
rotations and reflections~\cite{jiang2018quantum,wan2022matchgate},
yielding a Quantum Time complexity
\begin{equation}
    N_{\QT}\in\order{n}\;.
\end{equation}

The corresponding fermionic measurement channel $\MC_F$ is diagonal in
the Majorana basis and invertible on the even-parity operator subspace,
\begin{equation}
    \MC_F
    =
    \sum_{r=0}^{n}
    \lm_{n,r}\mathcal{P}_{2r},
    \qquad
    \lm_{n,r}
    =
    \binom{n}{r}
    \binom{2n}{2r}^{-1},
    \label{eq:fermionic_shadow_channel}
\end{equation}
where $\mathcal{P}_{2r}$ denotes the projector onto the homogeneous
Majorana sector $\BC_{2r}$. Equivalently, for every
$\vec{\mu}\in\CC_{2n,2r}$,
\begin{equation}
    \MC_F\left(G_{\vec{\mu}}^{2r}\right)
    =
    \lm_{n,r}G_{\vec{\mu}}^{2r}.
\end{equation}

For a single Majorana monomial, the corresponding shadow norm satisfies
\begin{equation}
    \norm{G_{\vec{\mu}}^{2r}}_{\rm FGU}^2
    =
    \lm_{n,r}^{-1}
    =
    \binom{2n}{2r}\binom{n}{r}^{-1}
    \in
    \Theta(n^r)
\end{equation}
for fixed $r$. More importantly for the setting considered here, this
scaling extends collectively to an arbitrary Hermitian observable in
the same homogeneous Majorana sector. In particular,
Ref.~\cite{west2026fermionic} shows that, for any
$O\in\BC_{2r}$,
\begin{equation}
    \norm{O}_{\rm FGU}^2
    \leq
    \frac{3}{2\lm_{n,r}}
    \norm{O}_\infty^2.
    \label{eq:collective_fermionic_shadow}
\end{equation}
Thus, the full Majorana sum can be estimated without separately
reconstructing each of its Majorana components.

In Matchgate $\liea$-sim, consider $T$ circuit instances
$\{U^{(i)}\}_{i=1}^T$ that share the same generators but can have different
circuit parameters. For each instance, the observable is backpropagated
as
\begin{equation}
    \widetilde{O}_i
    =
    (U^{(i)})\ad O U^{(i)}.
\end{equation}
Since $\BC_\kappa$ is invariant under Matchgate evolution,
$\widetilde{O}_i\in\BC_\kappa$. Moreover, unitary conjugation preserves
the spectral norm,
\begin{equation}
    \norm{\widetilde{O}_i}_\infty
    =
    \norm{O}_\infty.
\end{equation}
Therefore, Eq.~\eqref{eq:collective_fermionic_shadow} applies uniformly
to all $T$ backpropagated observables. Importantly, the same Matchgate
shadow data of the initial state can be reused to estimate all of these
expectation values.

This gives the following result.

\begin{theorem}[Quantum Sample Complexity of Matchgate $\liea$-sim]
\label{thm:classical_shadow_fermionic}
Consider $T$ Matchgate circuit instances and a Hermitian observable
\begin{equation}
    O
    =
    \sum_{\vec{\mu}\in\CC_{2n,\kappa}}
    c_{\vec{\mu}}G_{\vec{\mu}}^\kappa
    \in\BC_\kappa
\end{equation}
with fixed even fermionic locality $\kappa\in\order{1}$. Using
Matchgate classical shadows, the Quantum Sample complexity required to
estimate all $T$ expectation values up to additive error $\epsilon$
simultaneously, with success probability at least $1-\delta$, satisfies
\begin{align}
    N_{\QS}
    &\in
    \order{
        \frac{
        \binom{2n}{\kappa}
        }{
        \binom{n}{\kappa/2}
        }
        \frac{\norm{O}_\infty^2}{\epsilon^2}
        \log\left(\frac{T}{\delta}\right)
    }
    \nonumber\\
    &\subseteq
    \order{
        \frac{n^{\kappa/2}}{\epsilon^2}
        \log\left(\frac{T}{\delta}\right)
        \norm{O}_\infty^2
    }\;.
    \label{eq:classical_shadow_fermionic}
\end{align}
\end{theorem}

\begin{proof}
For $\kappa=2r$, Eq.~\eqref{eq:collective_fermionic_shadow} gives
\begin{equation}
    \max_i
    \norm{\widetilde{O}_i}_{\rm FGU}^2
    \leq
    \frac{3}{2}
    \binom{2n}{\kappa}
    \binom{n}{\kappa/2}^{-1}
    \norm{O}_\infty^2,
\end{equation}
where we used
$\norm{\widetilde{O}_i}_\infty=\norm{O}_\infty$. Applying the general
classical-shadow bound in Eq.~\eqref{eq:shadow_complexity} to the
collection of $T$ backpropagated observables yields
\begin{equation}
    N_{\QS}
    \in
    \order{
        \frac{
        \binom{2n}{\kappa}
        }{
        \binom{n}{\kappa/2}
        }
        \frac{\norm{O}_\infty^2}{\epsilon^2}
        \log\left(\frac{T}{\delta}\right)
    }\;.
\end{equation}
For fixed $\kappa$,
\begin{equation}
    \binom{2n}{\kappa}
    \binom{n}{\kappa/2}^{-1}
    \in
    \Theta(n^{\kappa/2}),
\end{equation}
which gives Eq.~\eqref{eq:classical_shadow_fermionic}.
\end{proof}

For comparison with bounds expressed directly in terms of the
coefficient vector, we can use
\begin{equation}
    \norm{O}_\infty
    \leq
    \sum_{\vec{\mu}}
    \abs{c_{\vec{\mu}}}
    \norm{G_{\vec{\mu}}^\kappa}_\infty
    =
    \norm{\vec{c}}_1,
\end{equation}
where
$\norm{G_{\vec{\mu}}^\kappa}_\infty=1$. Hence,
Eq.~\eqref{eq:classical_shadow_fermionic} also implies the more
conservative coefficient-based bound
\begin{equation}
    N_{\QS}
    \in
    \order{
        \frac{n^{\kappa/2}}{\epsilon^2}
        \log\left(\frac{T}{\delta}\right)
        \norm{\vec{c}}_1^2
    }\;.
\end{equation}
For the physically common cases $\kappa=2$ and $\kappa=4$, the
polynomial dependence arising from the Matchgate shadow norm is
therefore linear and quadratic in $n$, respectively.

Importantly, the collective bound in
Eq.~\eqref{eq:collective_fermionic_shadow} does not require estimating
the Majorana components with separate quantum data sets. For the
CS, we process the same Matchgate shadow samples once
into an empirical description of the degree-$\kappa$ Majorana sector.

For the signed-permutation implementation, a degree-$\kappa$ Majorana
string contributes to a measurement outcome only when its support is a
union of $\kappa/2$ pairs in the corresponding Majorana matching.
Consequently, a single snapshot contributes to at most
\begin{equation}
    \binom{n}{\kappa/2}
    \in
    \order{n^{\kappa/2}}
\end{equation}
degree-$\kappa$ components. Accumulating the $N_{\QS}$ snapshots
therefore requires
\begin{equation}
    N_{\rm shadow\mbox{-}post}
    \in
    \order{
        N_{\QS}n^{\kappa/2}
    }.
    \label{eq:matchgate_shadow_postprocessing}
\end{equation}
This cost is incurred only once. After the empirical degree-$\kappa$
description has been constructed, evaluating its inner product with a
backpropagated observable requires at most $\order{n^\kappa}$
operations per circuit instance, which is subleading compared with the
dense $\liea$-sim matrix-vector operations considered below.

\subsection{\label{adx:time_complexity_gsim_majorana}Classical Time Complexity for $\liea$-sim with Majorana operator formalism} 
A central cost of $\liea$-sim is constructing the matrix representation of the adjoint action. In a generic implementation, this representation is obtained by evaluating the adjoint action on each basis element separately. For Matchgate circuits, however, the Majorana formalism provides an equivalent and more direct implementation. Rather than constructing the adjoint representation element by element, we first combine the linear transformations associated with the circuit layers and then apply the resulting transformation to the coefficient tensor of the observable as a whole. Therefore, this can be viewed as a specialized, matrix-free implementation of $\liea$-sim on the invariant Majorana sector $\BC_{\kappa}$

Consider a quantum circuit $U  = U_{Q_L} \cdots U_{Q_1}$  where each layer is generated by a quadratic free-fermionic Hamiltonian, and is associated with an orthogonal matrix $Q_\ell \in \mathrm{SO}(2n)$. The complete circuit is itself a fermionic Gaussian unitary characterized by 
\begin{equation}
    Q = Q_L \cdots Q_1. 
\end{equation}
Thus, for a fixed circuit instance, the combined matrix $Q$ needs to be constructed only once. Using standard dense matrix multiplication, composing the $L$ matrices requires $\order{Ln^3}$ operations.

The adjoint action of the combined unitary on a $\kappa$-local Majorana string is 
\begin{equation}
    U\ad \gamma_{\vec{\mu}} U = \sum_{\vec{\al}} Q_{\mu_1 \al_1} \cdots Q_{\mu_\kappa \al_\kappa} \gamma_{\vec{\al}}\;.
\end{equation}
where the sum runs over $\al_i\in{1,\ldots,2n}$, with the resulting Majorana products brought into canonical order. Therefore, for an observable $O = \sum_{\vec{\mu}} c_{\vec{\mu}} \gamma_{\vec{\mu}}$,  the backpropagated observable can be written as 
\begin{equation}
    U\ad O U = \sum_{\vec{\mu}} c_{\vec{\mu}} U\ad \gamma_{\vec{\mu}} U =  \sum_{\vec{\mu}} \sum_{\vec{\al}}  c_{\vec{\mu}} Q_{\mu_1 \al_1} \cdots Q_{\mu_\kappa \al_\kappa} \gamma_{\vec{\al}} = \sum_{\vec{\al}} \tilde{c}_{\vec{\al}} \gamma_{\vec{\al}}\;,   
\end{equation}
where $ \tilde{c}_{\vec{\al}} = \sum_{\vec{\mu}} c_{\vec{\mu}} Q_{\mu_1 \al_1} \cdots Q_{\mu_\kappa \al_\kappa} $ denotes the backpropagated coefficient tensor.  
It  is therefore sufficient to compute $\tilde{c}_{\vec{\al}}$ for all $\vec{\al} \in \CC_{2n, \kappa}$.  

To perform this transformation efficiently, we initialize
\begin{equation}
\mathsf{T}^{(0)}_{\mu_1\cdots\mu_\kappa}
=
c_{\mu_1\cdots\mu_\kappa}
\end{equation}
and contract one tensor index at a time. At the $i$-th step, we define
\begin{equation}
\mathsf{T}^{(i)}_{\al_1\cdots\al_i \mu_{i+1}\cdots\mu_\kappa}
=
\sum_{\mu_i=1}^{2n}
\mathsf{T}^{(i-1)}_{\al_1\cdots\al_{i-1}
\mu_i\mu_{i+1}\cdots\mu_\kappa}
Q_{\mu_i\al_i}.
\end{equation}
After all $\kappa$ contractions have been performed,
\begin{equation}
\mathsf{T}^{(\kappa)}_{\al_1\cdots\al_\kappa}
=
\widetilde{c}_{\al_1\cdots\al_\kappa}.
\end{equation}

Each intermediate tensor contains $\order{n^\kappa}$ entries, and computing each entry requires a summation over $2n$ values. A single contraction therefore requires $\order{n^{\kappa+1}}$ operations. Since the contraction is performed once for each of the $\kappa$ tensor indices, the cost of backpropagating the observable for one circuit instance is 
$\order{\kappa n^{\kappa+1}}$. 
For fixed fermionic locality $\kappa\in\order{1}$ and $T$ circuit instances, this gives
\begin{equation}
\order{T\kappa n^{\kappa+1}}
\subseteq
\order{Tn^{\kappa+1}}.
\end{equation}
If the same matrix $Q$ is reused across all $T$ evaluations, its construction contributes the one-time cost $\order{Ln^3}$. If the circuit parameters change between instances, and hence produce a different combined matrix $Q$, this composition cost must instead be incurred for each circuit instance.

Combining the contributions above, the Classical Time complexity
contains
\begin{enumerate}
    \item the one-time Matchgate shadow post-processing cost,
    $\order{N_{\QS}n^{\kappa/2}}$;
    \item the construction of the combined orthogonal matrix
    $Q^{(i)}$ for each circuit instance,
    $\order{TLn^3}$; and
    \item the backpropagation of the observable $O$ for all
    circuit instances,
    $\order{Tn^{\kappa+1}}$.
\end{enumerate}
Hence,
\begin{equation}
    N_{\CT}
    \in
    \order{
        N_{\QS}n^{\kappa/2}
        +
        TLn^3
        +
        Tn^{\kappa+1}
    }.
    \label{eq:NCT_matchgate_appendix}
\end{equation}
Together with Eq.~\eqref{eq:classical_shadow_fermionic} and
$N_{\QT}\in\order{n}$, Eq.~\eqref{eq:NCT_matchgate_appendix}
proves Theorem~\ref{thm:classical_sim_matchgate}.

\section{\label{adx:Sn_QNN}$S_n$-equivariant circuits}

In this appendix, we provide the proof of
Theorem~\ref{thm:classical_sim_Sn}. For qubits, the Schur-basis
decomposition of an $S_n$-equivariant operator can be written as
\begin{equation}
    A
    \cong
    \bigoplus_\lambda
    \eye_{m_\lambda}\otimes A_\lambda,
\end{equation}
where the irrep labels can be parameterized as
$\lambda=(n-m,m)$ with $m=0,\ldots,\lfloor n/2\rfloor$. The
corresponding dimension-register blocks have size
\begin{equation}
    d_\lambda
    =
    n-2m+1
    \leq
    n+1.
    \label{eq:Sn_block_dimension}
\end{equation}
Hence, there are $\order{n}$ distinct blocks and each has dimension at
most $\order{n}$.

For an $S_n$-equivariant circuit and observable,
\begin{equation}
    f_U(\rho,O)
    =
    \sum_\lambda
    \Tr\left[
        \rho_\lambda
        U_\lambda^\dagger O_\lambda U_\lambda
    \right],
    \label{eq:Sn_expectation_appendix}
\end{equation}
where $\rho_\lambda$ denotes the projection of $\rho$ onto the
corresponding dimension-register block after summing over the
multiplicity labels. The initial state itself need not be
$S_n$-equivariant.

\subsection{Quantum Sample Complexity}

For a generic initial state, we use deep Permutation Invariant
Classical Shadows (PI-CS)~\cite{sauvage2024classical,
chang2026practical}. For an $S_n$-equivariant observable $A$, the
corresponding estimator satisfies
\begin{equation}
    \Var[\widehat A]
    \leq
    \order{n^2}\norm{A}_\infty^2.
    \label{eq:Sn_shadow_variance}
\end{equation}
For each circuit instance, define
\begin{equation}
    \widetilde O_i
    =
    (U^{(i)})^\dagger O U^{(i)}.
\end{equation}
Since $U^{(i)}$ and $O$ are both $S_n$-equivariant,
$\widetilde O_i$ is also $S_n$-equivariant, and unitary conjugation
preserves its spectral norm,
\begin{equation}
    \norm{\widetilde O_i}_\infty
    =
    \norm{O}_\infty.
    \label{eq:Sn_norm_preservation}
\end{equation}
Applying Eq.~\eqref{eq:shadow_complexity} to the collection
$\{\widetilde O_i\}_{i=1}^{T}$ therefore gives
\begin{equation}
    N_{\QS}
    \in
    \order{
        \frac{n^2}{\epsilon^2}
        \log\left(\frac{T}{\delta}\right)
        \norm{O}_\infty^2
    }.
    \label{eq:Sn_NQS_appendix}
\end{equation}
The PI-CS data are acquired only once and reused for all $T$ circuit
instances, so $T$ enters only through the simultaneous-estimation
factor.

For the normalized symmetrized Pauli operators used in the main text,
$\norm{P_\alpha}_\infty\leq1$. Hence, for
$O=\sum_\alpha c_\alpha P_\alpha$,
\begin{equation}
    \norm{O}_\infty
    \leq
    \sum_\alpha
    \abs{c_\alpha}
    \norm{P_\alpha}_\infty
    \leq
    \norm{\vec{c}}_1,
\end{equation}
and Eq.~\eqref{eq:Sn_NQS_appendix} implies
\begin{equation}
    N_{\QS}
    \in
    \order{
        \frac{n^2}{\epsilon^2}
        \log\left(\frac{T}{\delta}\right)
        \norm{\vec{c}}_1^2
    }.
    \label{eq:Sn_NQS_coeff_appendix}
\end{equation}

\subsection{Classical Time Complexity}

We next consider the circuit-dependent classical work. For the one- and
two-local symmetrized Pauli generators used here, the matrices
$(H_\ell)_\lambda$ are banded Hermitian matrices with constant
bandwidth~\cite{chang2026practical}. Their eigendecompositions,
\begin{equation}
    (H_\ell)_\lambda
    =
    W_{\ell,\lambda}
    \Xi_{\ell,\lambda}
    W_{\ell,\lambda}^\dagger,
    \label{eq:Sn_eigendecomposition}
\end{equation}
where $W_{\ell, \lambda}$ is the unitary matrix for change of basis and $\Xi_{\ell, \lambda}$ is the diagonal matrix of eigenvalues, can be evaluated in $\order{d_\lambda^2}$ operations per block.
Moreover,
\begin{equation}
    \sum_\lambda d_\lambda^2
    =
    \sum_{m=0}^{\lfloor n/2\rfloor}
    (n-2m+1)^2
    \in
    \order{n^3}.
    \label{eq:Sn_sum_d2}
\end{equation}
Thus, diagonalizing the Schur blocks of the $L$ generators requires the
one-time cost
\begin{equation}
    N_{\rm circ,pre}
    \in
    \order{Ln^3}.
    \label{eq:Sn_preprocessing}
\end{equation}
If fewer than $L$ distinct generators appear, $L$ can be replaced by
their number.

For the parameter $w_{i,\ell}$, the corresponding gate block is
\begin{equation}
    (U_{i,\ell})_\lambda
    =
    W_{\ell,\lambda}
    e^{-iw_{i,\ell}\Xi_{\ell,\lambda}}
    W_{\ell,\lambda}^\dagger.
\end{equation}
The diagonal exponential costs $\order{d_\lambda}$, while the dense
matrix operations required to update the propagated block are bounded
by $\order{d_\lambda^\omega}$. Therefore, one circuit instance requires
\begin{equation}
    N_{\rm eval}
    \in
    \order{
        L\sum_\lambda d_\lambda^\omega
    }
    \subseteq
    \order{Ln^{\omega+1}},
    \label{eq:Sn_evaluation_cost}
\end{equation}
and the $T$ circuit instances contribute
$\order{TLn^{\omega+1}}$.

We also include the classical processing of the deep PI-CS data. Each
snapshot resolves one irrep label and produces an estimator acting on
the corresponding dimension register. Since $d_\lambda\leq n+1$,
accumulating a dense estimator for the observed block requires at most
$\order{n^2}$ operations per snapshot. A conservative one-time bound is
therefore
\begin{equation}
    N_{\rm shadow\mbox{-}post}
    \in
    \order{
        N_{\QS}n^2
    }.
    \label{eq:Sn_shadow_postprocessing}
\end{equation}
The resulting empirical block description is reused for all circuit
instances. The final contraction with the backpropagated observable
costs at most
$\order{\sum_\lambda d_\lambda^2}=\order{n^3}$ per instance and is
absorbed into the circuit-evaluation term above. Combining the
contributions gives
\begin{equation}
    N_{\CT}
    \in
    \order{
        N_{\QS}n^2
        +
        Ln^3
        +
        TLn^{\omega+1}
    }.
    \label{eq:Sn_NCT_appendix}
\end{equation}
If the required irrep-block description of the input state is known
classically, the first term is absent.

\subsection{Quantum Time Complexity}

Deep PI-CS requires a quantum Schur transform. Efficient approximate
implementations have Quantum Time complexity
\begin{equation}
    N_{\QT}
    \in
    \order{
        n\,
        \poly\left(
            \log\epsilon_{\rm QST}^{-1}
        \right)
    },
    \label{eq:Sn_NQT_appendix}
\end{equation}
where $\epsilon_{\rm QST}$ denotes the approximation error of the Schur
transform. The subsequent measurements of the irrep and dimension
registers do not change this leading scaling.

Equations~\eqref{eq:Sn_NQS_appendix},
\eqref{eq:Sn_NCT_appendix}, and
\eqref{eq:Sn_NQT_appendix} prove
Theorem~\ref{thm:classical_sim_Sn}.

\section{$\U(1)$-equivariant circuit\label{adx:U1_symmetric}}

As discussed in Section~\ref{sec:U1_symmetric}, a $\U(1)$-equivariant
circuit preserves the Hamming weight of computational basis states.
Accordingly,
\begin{equation}
    \HC
    =
    \left(\Cbb^2\right)^{\otimes n}
    \cong
    \bigoplus_{h=0}^n\HC_h^{(n)},
\end{equation}
where
\begin{equation}
    \HC_h^{(n)}
    =
    \spn(B_h^{(n)}),
    \qquad
    B_h^{(n)}
    =
    \left\{
        \ket{\vec z}
        \;\middle|\;
        \vec z\in\Omega_h^{(n)}
    \right\},
\end{equation}
with
\begin{equation}
    \Omega_h^{(n)}
    =
    \left\{
        \vec z\in\{0,1\}^n
        \;\middle|\;
        {\rm HW}(\vec z)=h
    \right\},
    \qquad
    d_h^{(n)}
    =
    \binom{n}{h}.
\end{equation}
If the initial state is pure and lies entirely within one such sector,
the complete evolution remains confined to the
$d_h^{(n)}$-dimensional subspace.

\subsection{Quantum Sample Complexity}

If the initial state is known classically, no quantum data acquisition
is required. Otherwise, we reconstruct it within the relevant
Hamming-weight sector using the All-Pairs $\U(1)$-symmetric classical
shadow protocol~\cite{hearth2024efficient}. Its measurement channel is
\begin{equation}
    \MC
    =
    \frac{1}{\abs{\Gamma_n}}
    \sum_{\pi\in\Gamma_n}\MC_\pi,
    \qquad
    \MC_\pi
    =
    \prod_{[ij]\in\pi}\MC_{[ij]},
    \label{eq:MC_all_pairs}
\end{equation}
where $\Gamma_n$ is the set of pairings of the qubits and $[ij]\in\pi$
denotes one pair. Each two-qubit channel is
\begin{equation}
    \MC_{[ij]}[\rho]
    =
    \Ebb_{U_{[ij]},z_i,z_j}
    \left[
        U_{[ij]}^\dagger
        \ketbraq{z_i z_j}
        U_{[ij]}
    \right],
    \label{eq:MC_ij}
\end{equation}
where $U_{[ij]}$ is sampled from a 2-design over the two-qubit
$\U(1)$-equivariant unitary group~\cite{hearth2025unitary}.

We use the following shadow-norm bound.

\begin{lemma}[Shadow norm of $\U(1)$-equivariant classical
shadows~\cite{hearth2024efficient}]
\label{lem:u1_shadow_norm}
Consider an $n$-qubit state with fixed particle number and a basis
operator $A\in\XC_{\U(1)}$ containing $n_z$ Pauli-$Z$ operators and
$n_{a,+}=n_{a,-}$ raising and lowering operators. Its squared shadow
norm satisfies
\begin{equation}
    \norm{A}_{\rm shadow}^2
    \leq
    \left(\frac{3}{2}\right)^{n_{a,+}+2n_z}
    \frac{n^{n_{a,+}}}{n_{a,+}!}
    \norm{A}_\infty^2.
    \label{eq:shadow_norm_u1}
\end{equation}
\end{lemma}

We now reconstruct a pure state of known Hamming weight,
\begin{equation}
    \ket{\psi}
    =
    \sum_{\vec z\in\Omega_h^{(n)}}
    \alpha_{\vec z}\ket{\vec z}
    =
    \sum_{\vec z\in\Omega_h^{(n)}}
    \alpha_{\vec z}
    a_{\vec z}^\dagger\ket{\emptyset},
    \label{eq:u1_initial_state_raising}
\end{equation}
where $\ket{\emptyset}=\ket{0}^{\otimes n}$ and
\begin{equation}
    a_{\vec z}^\dagger
    =
    (a_1^\dagger)^{z_1}\cdots
    (a_n^\dagger)^{z_n}.
\end{equation}
For
$\vec\mu,\vec\nu\in\Omega_h^{(n)}$,
\begin{equation}
    \Tr[
        a_{\vec\mu}^\dagger
        a_{\vec\nu}
        \rho
    ]
    =
    \alpha_{\vec\mu}^*
    \alpha_{\vec\nu}.
    \label{eq:u1_amplitude_product}
\end{equation}
Choose a reference configuration $\vec\mu_r$ for which
$r_{\vec\mu_r}:=\abs{\alpha_{\vec\mu_r}}$ is maximal. Normalization
implies
\begin{equation}
    r_{\vec\mu_r}
    \geq
    \frac{1}{\sqrt{d_h^{(n)}}}.
    \label{eq:u1_reference_bound}
\end{equation}
The diagonal expectation value determines
$r_{\vec\mu_r}$, while
$\alpha_{\vec\mu_r}^*\alpha_{\vec\nu}$ determines every remaining
amplitude relative to the same global phase.

The operators
$a_{\vec\mu}^\dagger a_{\vec\nu}$ can be expanded in
$\XC_{\U(1)}$. For fixed $h$, every term contains at most $h$
raising/lowering pairs and at most $h$ Pauli-$Z$ operators. Thus,
Lemma~\ref{lem:u1_shadow_norm} gives
\begin{equation}
    \norm{
        a_{\vec\mu}^\dagger a_{\vec\nu}
    }_{\rm shadow}^2
    \in
    \order{n^h},
    \label{eq:u1_amplitude_shadow_norm}
\end{equation}
where factors depending only on $h$ are absorbed into the
$\order{\cdot}$ notation.

\begin{lemma}[Quantum Sample Complexity of $\U(1)$-equivariant state
tomography]
\label{lem:u1_state_tomography}
Consider a pure initial state $\ket{\psi}$ with fixed Hamming weight
$h\in\order{1}$. The Quantum Sample complexity required to reconstruct
the state, up to a global phase, with error at most
$\widetilde\epsilon$ per amplitude and success probability at least
$1-\delta$ satisfies
\begin{equation}
    N_{\QS}
    \in
    \order{
        \frac{n^{2h}}{\widetilde\epsilon^2}
        \log\left(\frac{n}{\delta}\right)
    }.
    \label{eq:NQS_U1_state_tomography}
\end{equation}
\end{lemma}

\begin{proof}
Assume that the quantities in
Eq.~\eqref{eq:u1_amplitude_product} are estimated to additive error
$\epsilon'$. From Eq.~\eqref{eq:u1_reference_bound},
\begin{equation}
    \abs{
        r_{\vec\mu_r}
        -
        \widetilde r_{\vec\mu_r}
    }
    \lesssim
    \epsilon'\sqrt{d_h^{(n)}}.
\end{equation}
The same scaling holds for the remaining reconstructed amplitudes,
\begin{equation}
    \abs{
        \alpha_{\vec\nu}
        -
        e^{i\phi}\widetilde\alpha_{\vec\nu}
    }
    \lesssim
    \epsilon'\sqrt{d_h^{(n)}},
\end{equation}
for one common global phase $\phi$. Hence, it suffices to choose
\begin{equation}
    \epsilon'
    \in
    \order{
        \frac{\widetilde\epsilon}
        {\sqrt{d_h^{(n)}}}
    }.
\end{equation}
Only $\order{d_h^{(n)}}$ expectation values relative to the fixed
reference configuration are required. Applying
Eq.~\eqref{eq:shadow_complexity} together with
Eq.~\eqref{eq:u1_amplitude_shadow_norm} gives
\begin{align}
    N_{\QS}
    &\in
    \order{
        \frac{n^h}{(\epsilon')^2}
        \log\left(
            \frac{d_h^{(n)}}{\delta}
        \right)
    }
    \nonumber\\
    &\subseteq
    \order{
        \frac{n^h d_h^{(n)}}
        {\widetilde\epsilon^2}
        \log\left(
            \frac{n}{\delta}
        \right)
    }
    \nonumber\\
    &=
    \order{
        \frac{n^{2h}}
        {\widetilde\epsilon^2}
        \log\left(
            \frac{n}{\delta}
        \right)
    },
\end{align}
where $d_h^{(n)}=\binom{n}{h}\in\order{n^h}$ for fixed $h$.
\end{proof}

We can now translate the amplitude accuracy into an error on the final
expectation value.

\begin{theorem}[Quantum Sample Complexity for Classical Simulation of
$\U(1)$-equivariant circuits]
\label{thm:u1_sample_appendix}
Consider $T$ $\U(1)$-equivariant circuit instances acting on the same
pure initial state $\ket{\psi}$ with fixed Hamming weight
$h\in\order{1}$. Let
\begin{equation}
    O
    =
    \sum_\alpha c_\alpha P_\alpha
\end{equation}
be a $\U(1)$-equivariant observable with
$\norm{P_\alpha}_\infty\leq1$. The Quantum Sample complexity required
to estimate all $f_{U^{(i)}}(\rho,O)$ to additive error $\epsilon$, with
success probability at least $1-\delta$, satisfies
\begin{equation}
    N_{\QS}
    \in
    \order{
        \frac{n^{3h}}{\epsilon^2}
        \log\left(\frac{n}{\delta}\right)
        \norm{\vec c}_1^2
    }.
    \label{eq:state_tomo_u1}
\end{equation}
\end{theorem}

\begin{proof}
Let $\ket{\widetilde\psi}$ denote the reconstructed state and define
$\ket{\Delta}=\ket{\psi}-\ket{\widetilde\psi}$. If every amplitude is
reconstructed to error $\widetilde\epsilon$, then
\begin{equation}
    \norm{\Delta}_2
    \leq
    \widetilde\epsilon\sqrt{d_h^{(n)}}.
\end{equation}
After normalizing the reconstructed state, if necessary, the
expectation-value error for any circuit instance is bounded, up to a
constant factor, by
\begin{align}
    \abs{
        \bramatketq{\psi}{(U^{(i)})^\dagger O U^{(i)}}
        -
        \bramatketq{\widetilde\psi}{(U^{(i)})^\dagger O U^{(i)}}
    }
    &\leq
    2\norm{O}_\infty\norm{\Delta}_2
    \nonumber\\
    &\leq
    2\widetilde\epsilon
    \sqrt{d_h^{(n)}}
    \norm{\vec c}_1,
    \label{eq:u1_expectation_error}
\end{align}
where
$\norm{O}_\infty\leq\norm{\vec c}_1$. It therefore suffices to take
\begin{equation}
    \widetilde\epsilon
    =
    \frac{\epsilon}
    {2\sqrt{d_h^{(n)}}\norm{\vec c}_1}.
    \label{eq:u1_amplitude_precision}
\end{equation}
Substitution into Eq.~\eqref{eq:NQS_U1_state_tomography} gives
Eq.~\eqref{eq:state_tomo_u1}. The reconstructed state is independent of
the circuit parameters, so the same state description is reused for all
$T$ instances and no additional dependence on $T$ is required.
\end{proof}

\subsection{Classical Time Complexity
\label{adx:U1_classical_time}}

We first consider the parameter-dependent circuit evaluation. Within
$\HC_h^{(n)}$, an RBS gate acting on qubits $i$ and $j$ is nontrivial
only on basis-state pairs for which exactly one of the two qubits is
occupied. The remaining $h-1$ excitations can be placed on any subset
of $h-1$ of the remaining $n-2$ qubits. Hence, one RBS gate acts on
\begin{equation}
    \binom{n-2}{h-1}
\end{equation}
two-dimensional pairs and can be applied in
$\order{\binom{n-2}{h-1}}$ classical operations. The locations of these
pairs depend only on the generator and can be precomputed once.

For one $L$-gate circuit instance,
\begin{equation}
    N_{\rm eval}
    \in
    \order{
        L\binom{n-2}{h-1}
    }.
\end{equation}
The rotation angles change between circuit instances, so the
parameter-dependent rotations must be reapplied for all $T$ instances,
giving
\begin{equation}
    N_{\rm circ}
    \in
    \order{
        TL\binom{n-2}{h-1}
    }
    \subseteq
    \order{TLn^{h-1}}.
    \label{eq:NCT_U1_appendix}
\end{equation}

We next include the one-time shadow post-processing used to reconstruct
the initial state. The All-Pairs protocol requires
$\order{n+n_z4^{n_z}}$ operations per sample and observable, together
with a one-time $\order{n_z^3}$ channel-inversion cost
~\cite{hearth2024efficient}. Reconstructing the state requires
$\order{d_h^{(n)}}$ expectation values, and
$n_z\leq h\in\order{1}$. A conservative bound is therefore
\begin{equation}
    N_{\rm shadow\mbox{-}post}
    \in
    \order{
        d_h^{(n)}N_{\QS}n
    }
    \subseteq
    \order{
        N_{\QS}n^{h+1}
    }.
    \label{eq:U1_shadow_postprocessing}
\end{equation}
This cost is independent of $T$ because it constructs a reusable
description of the initial state.

Combining both contributions gives
\begin{align}
    N_{\CT}
    &\in
    \order{
        N_{\QS}n^{h+1}
        +
        TL\binom{n-2}{h-1}
    }
    \nonumber\\
    &\subseteq
    \order{
        \frac{n^{4h+1}}{\epsilon^2}
        \log\left(\frac{n}{\delta}\right)
        \norm{\vec c}_1^2
        +
        TLn^{h-1}
    }.
    \label{eq:NCT_U1_total_appendix}
\end{align}
If the initial state is known classically, the first term is absent.

\subsection{Quantum Time Complexity}

The All-Pairs measurement circuit consists of disjoint two-qubit gates.
Assuming the required pairwise gates can be implemented in parallel,
and ignoring additional routing overhead, the Quantum Time complexity
per sample is
\begin{equation}
    N_{\QT}\in\order{1}.
\end{equation}
If the initial state is known classically, no shadow acquisition is
required and $N_{\QS}=N_{\QT}=0$.

Equations~\eqref{eq:state_tomo_u1} and
\eqref{eq:NCT_U1_total_appendix}, together with the Quantum Time bound
above, prove Theorem~\ref{thm:classical_sim_U1}.

\section{Quantum Convolutional Neural Networks\label{adx:qcnn}}

\begin{figure}[h]
    \centering
    \includegraphics[width=0.4\linewidth]{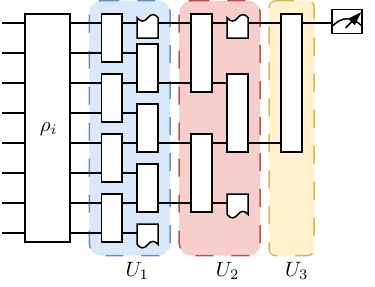}
    \caption{\textbf{Schematic diagram of a quantum convolutional neural network.} The input state $\rho_i$ is embedded into the $n$-qubit system. The QCNN circuit consists of convolutional filters, which are two-qubit unitaries acting on alternating pairs of neighboring qubits, and pooling layers that reduce the number of active qubits. At each hierarchical layer $U_j$, the pooling step reduces the number of active qubits by a factor of two.}
    \label{fig:qcnn}
\end{figure}

A QCNN has a hierarchical architecture in which the number of active
qubits decreases with the layer index. We write
\begin{equation}
    n_j=\frac{n}{2^{j-1}},
    \qquad
    j=1,\ldots,L,
    \qquad
    L=\log_2(n),
\end{equation}
and denote by $\overline{\PC}_{n_j}^{(\tau)}$ the
Hilbert--Schmidt-normalized Pauli strings supported on the $n_j$ active
qubits with weight at most $\tau$. The number of such strings satisfies
\begin{equation}
    D_\tau(n_j)
    =
    \sum_{r=0}^{\tau}3^r\binom{n_j}{r}
    \in
    \order{n_j^\tau}
\end{equation}
for fixed $\tau$.

\subsection{Low-weight Pauli propagation and Classical Time complexity}

The Pauli-path representation can be adapted directly to this
hierarchical structure. Using the Hilbert--Schmidt-normalized Pauli
basis introduced in Section~\ref{sec:pauli_propagation}, we write
\begin{equation}
    \Tr[U\ad O U\rho]
    =
    \sum_{\substack{s_0\in\overline{\PC}_n\\
                    s_j\in\overline{\PC}_{n_j}}}
    \Tr[Os_L]
    \prod_{j=1}^{L}
    \Tr[U_j\ad s_j U_j s_{j-1}]
    \Tr[s_0\rho].
    \label{eq:pauli_path_qcnn}
\end{equation}
Low-weight Pauli Propagation restricts the sum to paths satisfying
$\abs{s_j}\leq\tau$ at every layer. The important difference from a
non-hierarchical circuit is that the size of the retained Pauli space
shrinks after every pooling step.

For the first transition, the number of retained pairs is at most
$D_\tau(n)^2\in\order{n^{2\tau}}$. For $j\geq2$, the number of
transition amplitudes between layers $j-1$ and $j$ is upper bounded by
\begin{equation}
    D_\tau(n_{j-1})D_\tau(n_j)
    \in
    \order{\frac{n^{2\tau}}{2^{\tau(2j-3)}}}.
\end{equation}
Therefore,
\begin{align}
    D_\tau(n)^2
    +
    \sum_{j=2}^{L}
    D_\tau(n_{j-1})D_\tau(n_j)
    &\in
    \order{
        n^{2\tau}
        \left(
            1+
            \sum_{j=2}^{L}
            2^{-\tau(2j-3)}
        \right)
    }
    \\
    &\in
    \order{n^{2\tau}}.
    \label{eq:qcnn_geometric_runtime}
\end{align}
Thus, evaluating $T$ parameterized QCNN instances requires
\begin{equation}
    N_{\rm prop}
    \in
    \order{Tn^{2\tau}}.
    \label{eq:qcnn_propagation_cost}
\end{equation}
The $\log(n)$ circuit depth does not introduce an additional logarithmic
factor because the layer-dependent costs form a convergent geometric
series. In practice, many transition amplitudes vanish and the realized
cost can be smaller than this worst-case count.

\subsection{Truncation error}

The low-weight restriction introduces an approximation error. For the
locally scrambling circuit ensembles considered in
Ref.~\cite{angrisani2024classically}, the low-weight estimator obeys
\begin{equation}
    \Ebb_U\left[
        \abs{f_U(\rho,O)-f_U^{(\tau)}(\rho,O)}^2
    \right]
    \leq
    \left(\frac{2}{3}\right)^{\tau+1}
    \norm{O}_{\rm Pauli,2}^2,
    \label{eq:qcnn_truncation_error_appendix}
\end{equation}
where
$\norm{O}_{\rm Pauli,2}=2^{-n/2}\norm{O}_2$.
Equation~\eqref{eq:qcnn_truncation_error_appendix} is an average-case
statement over the circuit ensemble. The end-to-end simulations of
Ref.~\cite{bermejo2024quantum} provide numerical evidence that the same
low-weight mechanism remains effective along the training trajectories
of the benchmark QCNNs considered there.

When the input state is estimated from quantum data, the total error
can be separated as
\begin{align}
    \abs{
        f_U(\rho,O)
        -
        \widehat f_U^{(\tau)}(\rho,O)
    }
    &\leq
    \abs{
        f_U(\rho,O)
        -
        f_U^{(\tau)}(\rho,O)
    }
    \nonumber\\
    &\quad+
    \abs{
        f_U^{(\tau)}(\rho,O)
        -
        \widehat f_U^{(\tau)}(\rho,O)
    }.
    \label{eq:qcnn_error_decomposition}
\end{align}
The first term is the truncation error, while the second is the
statistical error associated with the initial-state data acquisition.

\subsection{Quantum Sample Complexity}

After truncation, the only input-dependent quantities are expectation
values $\Tr[s_0\rho]$ for low-weight Pauli strings. Randomized local
Pauli measurements can therefore be performed once on the initial state
and reused for every later circuit parameter value. Combining the
classical-shadow construction with the locally scrambling
low-weight estimator gives the sufficient scaling
\begin{equation}
    N_{\QS}
    \in
    \order{
        \frac{\exp(\order{\tau})}{\epsilon^2}
        \log\left(\frac{n}{\delta}\right)
        \norm{\vec{c}}_1^2
    }.
    \label{eq:NQS_qcnn_appendix}
\end{equation}
The exponential dependence is only on the truncation threshold
$\tau$. For fixed $\tau$, the quantum data acquisition therefore grows
only logarithmically with the system size. The shadow is used to
estimate a fixed collection of low-weight Pauli expectations of the
input state. Once these quantities are controlled simultaneously, they
can be reused for all $T$ QCNN instances, which is why
Eq.~\eqref{eq:NQS_qcnn_appendix} does not acquire an additional
dependence on $T$. If the input data
admit an efficient classical description, this quantum data-acquisition
step is absent.

\subsection{Classical shadow post-processing and Quantum Time complexity}

A direct implementation can process the $N_{\QS}$ snapshots into the
$D_\tau(n)\in\order{n^\tau}$ low-weight Pauli expectation values used by
the surrogate. This gives the one-time Classical Time cost
\begin{equation}
    N_{\rm pre}^{(\rm shadow)}
    \in
    \order{N_{\QS}n^\tau}.
    \label{eq:qcnn_shadow_postprocessing}
\end{equation}
Together with Eq.~\eqref{eq:qcnn_propagation_cost}, the total Classical
Time complexity is
\begin{equation}
    N_{\CT}
    \in
    \order{
        N_{\QS}n^\tau
        +
        Tn^{2\tau}
    },
\end{equation}
which is Eq.~\eqref{eq:NCT_qcnn} in the main text. Since the quantum
data acquisition uses only local basis rotations followed by
single-qubit measurements, all basis changes can be performed in
parallel and
\begin{equation}
    N_{\QT}\in\order{1}.
\end{equation}
If the required low-weight input-state expectations are known
classically, the shadow acquisition and its post-processing are absent,
so the first term in $N_{\CT}$ vanishes and
$N_{\QS}=N_{\QT}=0$. This proves the resource scalings stated in
Theorem~\ref{thm:classical_sim_qcnn}.

\section{ADAPT-VQE for Quantum Chemistry
\label{adx:adapt_vqe}}

\subsection{Architecture}

ADAPT-VQE constructs the ansatz iteratively rather than fixing its
architecture before the optimization~\cite{grimsley2019adaptive,
tang2019qubit,grimsley2022adapt}. Starting from the Hartree--Fock
state $\ket{\psi_{\rm HF}}$, each adaptive step evaluates the energy
gradient associated with the operators in a pool $\AC$, selects the
operator with the largest gradient magnitude, appends the corresponding
parameterized gate, and then re-optimizes all parameters in the enlarged
ansatz.

We restrict the operator pool to Majorana strings with bounded fermionic
locality, as in Section~\ref{sec:adapt_vqe}. After $\ell$ adaptive
steps, the circuit can be written as
\begin{equation}
    U^{(\ell)}(\vec{w})
    =
    \prod_{j=1}^{\ell}
    e^{-iw_jG^{\kappa_j}_{\vec{\mu}_j}/2},
    \qquad
    G^{\kappa_j}_{\vec{\mu}_j}\in\AC.
    \label{eq:adapt_appendix_circuit}
\end{equation}
The final circuit contains $L$ selected generators. We use
$T_{\rm VQE}$ for the total number of energy evaluations performed over
all parameter re-optimization stages of the complete adaptive run.

The operation-counting arguments below apply once a truncation threshold
$\tau$ has been fixed. The randomness assumptions enter separately when
establishing that this truncation provides a controlled approximation.
In particular, the rigorous error guarantee used in this work is the
one obtained for the randomized fermionic ADAPT-VQE setting considered
in Ref.~\cite{miller2025simulation}.
\subsection{Majorana Propagation and truncation error}

Majorana Propagation~\cite{miller2025simulation} is the fermionic
analogue of Pauli Propagation. For an observable expanded in Hermitian
Majorana strings,
\begin{equation}
    O
    =
    \sum_{\kappa,\vec{\mu}}
    c^{(\kappa)}_{\vec{\mu}}
    G^\kappa_{\vec{\mu}},
    \label{eq:adapt_majorana_expansion}
\end{equation}
we backpropagate each retained Majorana string through the circuit.
Conjugation by a Majorana rotation either leaves a string unchanged or
generates a second Majorana string. Weight truncation discards strings
whose fermionic locality exceeds $\tau$.

Under the assumption $\tau\leq n/2$, the number of retained Majorana
strings is bounded as
\begin{equation}
    \sum_{\kappa=0}^{\tau}
    \binom{2n}{\kappa}
    \in
    \order{n^\tau}.
    \label{eq:adapt_retained_majoranas}
\end{equation}
Hence, for a fixed truncation threshold, backpropagating one observable
through an $L$-gate circuit requires at most
\begin{equation}
    N_{\CT}
    \in
    \order{Ln^\tau}
    \label{eq:adapt_single_propagation}
\end{equation}
operations.

The truncation makes this simulation approximate. For the randomized
fermionic circuit ensemble considered in
Ref.~\cite{miller2025simulation}, together with an observable whose
Majorana coefficients are unbiased and uncorrelated, the mean-squared
truncation error satisfies
\begin{equation}
    \Ebb\left[
        \left(
            f_U(\rho,O)
            -
            f_U^{(\tau)}(\rho,O)
        \right)^2
    \right]
    \leq
    \left(\frac{1}{2^{n-1}} + \left(
        \frac{e\tau}{n}
    \right)^{\tau/2} 
    \right)
    \Ebb\left[
        \majoranaNorm{O}^2
    \right],
    \label{eq:adapt_appendix_mse}
\end{equation}
where $n$ denotes the number of fermionic modes, and
\begin{equation}
    \Ebb\left[
        \majoranaNorm{O}^2
    \right]
    =
    \Ebb\left[
        \frac{\Tr[OO^\dagger]}{2^n}
    \right].
\end{equation}
For the Jordan--Wigner encoding considered here, the number of
fermionic modes coincides with the number of qubits.

Equation~\eqref{eq:adapt_appendix_mse} can also be converted into a
high-probability guarantee using Markov's inequality,
\begin{equation}
    \Pr\left[
        \abs{
            f_U(\rho,O)-f_U^{(\tau)}(\rho,O)
        }
        >
        \epsilon
    \right]
    \leq
    \frac{1}{\epsilon^2} 
    \left( 
    \frac{1}{2^{n-1}} + 
    \left(
        \frac{e\tau}{n}
    \right)^{\tau/2}
    \right)
    \Ebb\left[
        \majoranaNorm{O}^2
    \right].
    \label{eq:adapt_appendix_high_probability}
\end{equation}
Thus, in the parameter regime considered in
Ref.~\cite{miller2025simulation}, a truncation threshold
\begin{equation}
    \tau
    \in
    \order{
        \log\left(
            \epsilon^{-1}\delta^{-1}
        \right)
    }
\end{equation}
provides a controlled approximation while retaining a polynomial
Classical Time complexity.

We stress that the randomness assumptions are required for the rigorous
error guarantee in Eq.~\eqref{eq:adapt_appendix_mse}, rather than for
the counting argument in Eq.~\eqref{eq:adapt_retained_majoranas}.
Numerical results in Ref.~\cite{chakraborty2026scalable} suggest that
Majorana Propagation can remain accurate beyond this randomized setting.

\subsection{Classical Time Complexity}

We now apply the truncated Majorana representation above to the
randomized ADAPT-VQE procedure. The Classical Time cost contains a
gradient-screening contribution and an energy-reevaluation contribution.  At each of the $L$ adaptive steps, the gradients associated with the
$\abs{\AC}$ pool operators must be evaluated from the current retained
Majorana representation. Under the reuse strategy assumed here, the
representation is extended when the new gate is appended rather than
reconstructed from scratch. With at most $\order{n^\tau}$ retained
strings, this gives
\begin{equation}
    N_{\rm grad}
    \in
    \order{L\abs{\AC}n^\tau}.
\end{equation}

After each generator is selected, the VQE parameters are re-optimized.
Each surrogate energy evaluation requires at most
$\order{n^\tau}$ operations, so the $T_{\rm VQE}$ evaluations
accumulated over the complete run contribute
\begin{equation}
    N_{\rm opt}
    \in
    \order{T_{\rm VQE}n^\tau}.
\end{equation}
Combining both contributions gives
\begin{equation}
    N_{\CT}
    \in
    \order{
        \left(
            L\abs{\AC}
            +
            T_{\rm VQE}
        \right)n^\tau
    }.
    \label{eq:NCT_adapt_vqe_appendix}
\end{equation}
This is the Classical Time scaling reported in
Theorem~\ref{thm:classical_sim_adapt_vqe}. The first term is generally
the cost of building and extending the gradient surrogates, while the
second cannot be neglected when the inner VQE optimization requires
many energy evaluations.

\subsection{Quantum Sample and Quantum Time Complexities}

For the quantum-chemistry setting considered here, the initial state is
the Hartree--Fock state. Its Majorana correlators are determined by the
classical occupation pattern, so the input quantities required by
Majorana Propagation can be evaluated analytically. Therefore, the
CS does not require a quantum data-acquisition phase,
and
\begin{equation}
    N_{\QS}=0,
    \qquad
    N_{\QT}=0.
    \label{eq:adapt_zero_quantum_resources_appendix}
\end{equation}
If a generic unknown input state were used instead, these two quantities
would have to be replaced by the cost of an appropriate fermionic
shadow or tomography protocol. Together with
Eq.~\eqref{eq:NCT_adapt_vqe_appendix}, this proves the resource scalings
stated in Theorem~\ref{thm:classical_sim_adapt_vqe}.

\section{Monetary cost across quantum hardware platforms\label{adx:monetary_cost_adx}}
Figure~\ref{fig:pricing_adx} presents the estimated monetary cost for all quantum hardware platforms listed in Table~\ref{tab:pricing}. These results are omitted from the main text for readability. 
\begin{figure*}[h]
    \centering
    \includegraphics[width=\linewidth]{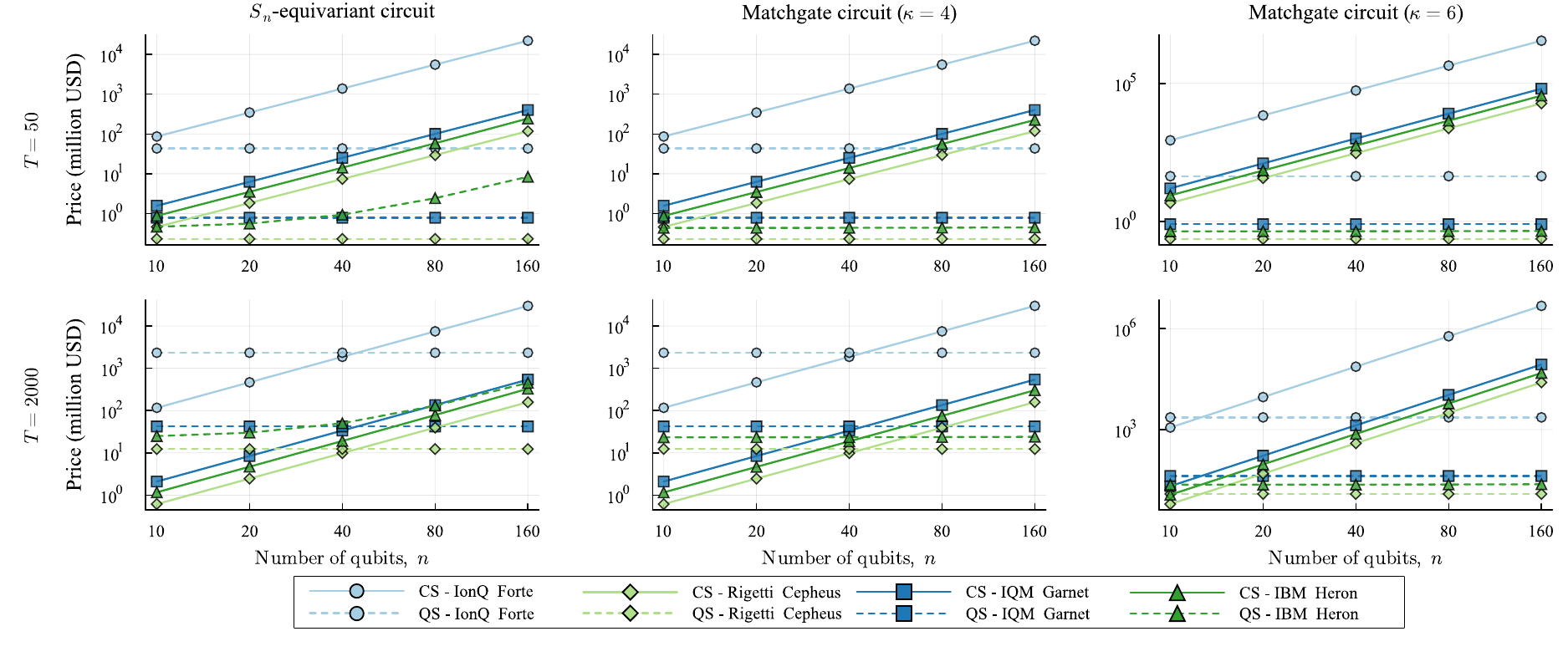}
    \caption{\textbf{Estimated quantum hardware access cost of QS and CS over $T = 50$ and $T=2000$ circuit instances.}
    We consider the $S_n$-equivariant and Matchgate examples with $\kappa = 4$ and $\kappa = 6$ described
    above, using the publicly available pricing data summarized in
    Table~\ref{tab:pricing}. The plotted values account for
    quantum hardware access and do not include the monetary cost of the
    classical computation. Results exceeding the available device size, namely
    $n>36$ for IonQ Forte, $n > 20$ for IQM Garnet, $n > 108$ for Rigetti Cepheus,  and $n>156$ for IBM Heron, extrapolate
    the reported hardware pricing and should therefore be
    interpreted only as scaling estimates. }
    \label{fig:pricing_adx}
\end{figure*}

\end{document}